\documentclass{lmcs}

\usepackage{enumerate}
\usepackage{hyperref}
\usepackage{amsmath,amssymb,inputenc,enumerate,mathtools,mathrsfs,bbold,array,multirow}
\usepackage{pb-diagram}
\usepackage{tikz,stmaryrd}
\usepackage{graphicx}
\usetikzlibrary{matrix,arrows,decorations,decorations.pathmorphing,decorations.markings,calc}
\usetikzlibrary{automata,positioning}

\DeclareMathOperator{\Imag}{Im}

\DeclareMathOperator{\fvec}{\mathtt{Vec}_\mathbb{K}}
\DeclareMathOperator{\set}{\mathtt{Set}}
\DeclareMathOperator{\CABA}{\mathtt{CABA}}

\DeclareMathOperator{\JSL}{\mathtt{JSL}}
\DeclareMathOperator{\TBJSL}{\mathtt{StJSL}}
\DeclareMathOperator{\ACDL}{\mathtt{AlgCDL}}
\DeclareMathOperator{\tvec}{\mathtt{StVec}_\mathbb{K}}
\DeclareMathOperator{\poset}{\mathtt{Poset}}

\DeclareMathOperator{\At}{At}

\DeclareMathOperator{\Ps}{\mathtt{P}}
\DeclareMathOperator{\Psc}{\mathtt{R}}
\DeclareMathOperator{\md}{Mod}

\DeclareMathOperator{\Coalg}{Coalg}
\DeclareMathOperator{\rec}{Rec}

\DeclareMathOperator{\Alg}{Alg}

\DeclareMathOperator{\T}{\tau}
\DeclareMathOperator{\Lan}{\mathscr{L}}

\DeclareMathOperator{\coeq}{coeq}
\DeclareMathOperator{\eq}{eq}

\DeclareMathOperator{\E}{\mathtt{E}}
\DeclareMathOperator{\M}{\mathtt{M}}

\DeclareMathOperator{\lepi}{\longrightarrow \mkern-15mu\rightarrow}

\theoremstyle{plain}
\theoremstyle{plain}\newtheorem{proposition}[thm]{Proposition}
\theoremstyle{plain}\newtheorem{lemma}[thm]{Lemma}
\theoremstyle{plain}\newtheorem{example}[thm]{Example}
\theoremstyle{plain}\newtheorem*{remark}{Remark}

\begin{document}

\title[Unveiling Eilenberg--type correspondences]{Unveiling Eilenberg--type correspondences: \ \ \ \ \ \ \ \  Birkhoff's Theorem for (finite) algebras + Duality}

\author[Julian Salamanca]{Julian Salamanca}
\address{CWI Amsterdam, The Netherlands}	
\email{salamanc@cwi.nl}  
\thanks{The research of this author is funded by the Dutch NWO project 612.001.210.}	

%





\begin{abstract}
  \noindent The purpose of the present paper is to show that: 
	\[
		\text{Eilenberg--type correspondences = Birkhoff's theorem for (finite) algebras + duality.} 
	\]
	We consider algebras for a monad $\mathsf{T}$ on a category $\mathcal{D}$ and we study (pseudo)varieties of 
	$\mathsf{T}$--algebras. Pseudovarieties of algebras are also known in the literature as varieties of finite algebras. Two well--known theorems that characterize varieties 
	and pseudovarieties of algebras play an important role here: Birkhoff's theorem and Birkhoff's theorem for finite algebras, the latter also known as Reiterman's theorem. 
	We prove, under mild assumptions, a categorical version of Birkhoff's theorem for (finite) algebras to establish a one--to--one correspondence between 
	(pseudo)varieties of $\mathsf{T}$--algebras and (pseudo)equational $\mathsf{T}$--theories. Now, if $\mathcal{C}$ is a category that is dual to $\mathcal{D}$ and 
	$\mathsf{B}$ is the comonad on $\mathcal{C}$ that is the dual of $\mathsf{T}$, we get a one--to--one correspondence between (pseudo)equational $\mathsf{T}$--theories 
	and their dual, (pseudo)coequational $\mathsf{B}$--theories. Particular instances of 
	(pseudo)coequational $\mathsf{B}$-theories have been already studied in language theory under the name of ``varieties of languages'' to establish Eilenberg--type correspondences. 
	All in all, we get a one--to--one correspondence between (pseudo)varieties of $\mathsf{T}$--algebras and (pseudo)coequational $\mathsf{B}$--theories, which will be shown 
	to be exactly the nature of Eilenberg--type correspondences.
\end{abstract}

\maketitle

\section*{Introduction}

The goal of the present paper is to show that: 
\[
	\text{Eilenberg--type correspondences = Birkhoff's theorem for (finite) algebras + duality.} 
\]
Eilenberg's theorem is an important result in algebraic language theory, stating that there is a one--to--one correspondence between certain classes of 
regular languages, called varieties of languages, and certain classes of monoids, called pseudovarieties of monoids \cite[Theorem 34]{eilenberg}. The concept of regular language, 
which is defined in terms of deterministic automata, has an equivalent machine--independent algebraic definition, namely, a language recognized by a finite monoid. Recognizable 
languages on an alphabet $\Sigma$ are inverse images of monoid homomorphisms with domain $\Sigma^*$ and as codomain 
any finite monoid. This algebraic approach allows us to study various kinds of recognizable languages where the notion of homomorphism between algebras 
is a key ingredient.

\noindent The study of algebras and classes of algebras is a main subject of study in universal algebra. A well--known theorem in this area is 
Birkhoff's variety theorem \cite{birk}, which states that a class of algebras of a given type is defined by a set of equations if and only if it is a variety, i.e., 
it is closed under homomorphic images, subalgebras, and products. Later, a Birkhoff's theorem for finite algebras was also obtained \cite{banab,reit}, also known as Reiterman's theorem. 
In Birkhoff's theorem for finite algebras the kind of equations are of a more general kind, they can be defined by using  topological techniques or can be equivalently defined by using the 
so--called implicit operations. In this case, the classes of algebras considered are pseudovarieties of algebras, 
also known as varieties of finite algebras, which are defined as classes of finite algebras of the same type that are closed under homomorphic images, subalgebras and finite products. 

\noindent To state Eilenberg--type theorems, which establish one--to--one correspondences between (pseudo)varieties of algebras and (pseudo)varieties of languages, 
one has to define and find the corresponding notion of a (pseudo)variety of languages which is, in general, a non--trivial problem. There are Eilenberg--type correspondences in the 
literature such as, e.g., \cite{pin} for pseudovarieties of ordered monoids and ordered semigroups, the one in \cite{reut} for pseudovarieties of finite 
dimensional $\mathbb{K}$--algebras, \cite{polak} for pseudovarieties of idempotent semirings and \cite[Theorem 39]{jan1} for varieties of monoids. 

\noindent The work in the present paper has its basis in \cite{mikolaj,sbr}. We take the main idea given in \cite{mikolaj}, where algebras 
for a monad $\mathsf{T}$ on $\mathcal{D}$ are considered, to define the natural notion of a (pseudo)variety of $\mathsf{T}$--algebras. In order to characterize the kind of equations 
defining a (pseudo)variety of $\mathsf{T}$--algebras, we use the natural approach of capturing equations as epimorphisms in $\mathcal{D}$, i.e., congruences, and add the condition that 
those equations are closed under substitution. These properties are captured in categorical terms to define the notion of a (pseudo)equational $\mathsf{T}$--theory. We obtain that, 
under mild assumptions, (pseudo)varieties of $\mathsf{T}$--algebras are exactly classes of (finite) $\mathsf{T}$--algebras that are defined by (pseudo)equational 
$\mathsf{T}$--theories, and that they are in one--to--one correspondence. This will give us a categorical version of Birkhoff's theorem for (finite) $\mathsf{T}$--algebras. Once we get this  
one--to--one correspondence between (pseudo)varieties of $\mathsf{T}$--algebras and (pseudo)equational $\mathsf{T}$--theories, we use a category $\mathcal{C}$ that is dual to 
$\mathcal{D}$ and a result given in \cite{sbr} that allows us to define a canonical comonad $\mathsf{B}$ on $\mathcal{C}$ that is dual to $\mathsf{T}$ and lift the duality between 
$\mathcal{C}$ and $\mathcal{D}$ to their corresponding Eilenberg--Moore categories. With this duality, there is a canonical correspondence between (pseudo)equational 
$\mathsf{T}$--theories and their corresponding dual, i.e., (pseudo)coequational $\mathsf{B}$--theories. Our most important examples of (pseudo)coequational $\mathsf{B}$--theories 
are those given in Eilenberg--type correspondences, i.e., ``varieties of languages''. All in all, we get a one--to--one correspondence between (pseudo)varieties of $\mathsf{T}$--algebras 
and (pseudo)coequational $\mathsf{B}$--theories. We will show how this concept of (pseudo)coequational $\mathsf{B}$--theories coincides with the different notions of 
``varieties of languages'' in Eilenberg--type correspondences, which bring us to our slogan, Eilenberg--type correspondences = Birkhoff's theorem for (finite) algebras + duality. 
As a consequence, we can summarize Eilenberg--type correspondences in the following picture:
\begin{center}	
	\begin{tikzpicture}[>=stealth,shorten >=3pt,
		node distance=2.5cm,on grid,auto,initial text=,
		accepting/.style={thick,double distance=2.5pt}]
		\node () at (0,-0.05) {(pseudo)varieties};
		\node () at (0,-0.45) {of $\mathsf{T}$--algebras};
		\draw[double,thick,<->] (1.5,-0.25) -- (7.5,-0.25);
		\node () at (4.5,0) {Eilenberg--type correspondences};

		\node () at (11.6,0.3) {$op$};
		\node () at (9.5,-0.2) {$\left(\begin{array}{c}
						\text{(pseudo)equational}\\
						\text{$\mathsf{T}$--theories}\end{array}\right)$};
	\end{tikzpicture}
\end{center}
where `$op$' denotes the dual operator. This easy to understand and straightforward one--to--one correspondence gives us what we called an abstract Eilenberg--type correspondence 
for (pseudo)varieties of $\mathsf{T}$--algebras, Proposition \ref{eilvar} and \ref{eilpvar}, from which we recover and discover particular instances of Eilenberg--type correspondences 
for different kinds of algebraic structures, i.e., $\mathsf{T}$--algebras. It is worth mentioning that Eilenberg--type correspondences have not been fully understood for the last forty years, 
which can be witnessed by the numerous published results on the subject that deal with specific kinds of algebras such as \cite{jan1,eilenberg,perpin,pin,polak,reut,wilke} and 
categorical generalizations such as \cite{adamek,mikolaj,urbat,jste} in which the direct relation between ``varieties of languages'' and equational theories, by using duality, is not studied 
or explored to find and justify the defining properties of a ``variety of languages''.

\noindent {\bf Related work.} We briefly summarize here some related work (see the Conclusions for a more detailed discussion.) There are various generalizations of Birkhoff's theorem  
for (finite) algebras such as \cite{awodeyh,bana,banab,chen}. In order to derive Eilenberg--type correspondences, in this paper we prove a categorical versions of Birkhoff's theorem for 
algebras and finite algebras, which are stated, under mild assumptions, as one--to--one correspondences between (pseudo)varieties of algebras and (pseudo)equational theories. The 
variety version is derived from \cite{bana} and the pseudovariety version is based on the observation that pseudovarieties of algebras are directed unions of of equational classes of 
finite algebras \cite[Proposition 4]{banab}. It is worth mentioning that the proof presented here for Birkhoff's theorem for finite $\mathsf{T}$--algebras does not involve the use of 
topology nor profinite techniques, in contrast to \cite{banab,chen,reit}.

\noindent Related work such as \cite{gehrke00,gehrke0,gehrke,gehrke1} 
have influenced and motivated the use of duality in language theory to characterize recognizable languages and to derive local versions of Eilenberg--type correspondences. 
In the present paper, the use of duality is a key aspect that helps us to understand and unveil Eilenberg--type correspondences. 

\noindent There are some works in which categorical approaches to derive Eilenberg--type correspondences are used, notably \cite{adamek,mikolaj,urbat,jste}. The work in 
\cite{urbat} subsumes 
the work made in \cite{adamek,mikolaj} and the present paper subsumes the work made in \cite{jste,mikolaj}. The kind of algebras considered in \cite{adamek} are algebras with a monoid 
structure which restricts the kind of algebras one can consider, e.g., Eilenberg's theorem \cite[Theorem 34s]{eilenberg} for pseudovarieties of semigroups cannot be derived 
from \cite{adamek}. A different approach to get a general Eilenberg--type theorem is the approach given in \cite{mikolaj} were the algebras considered are algebras for a monad 
$\mathsf{T}$ on $\set^S$, for a fixed set $S$. The fact that all the monads considered are on $\set^S$ was not general enough to cover cases such as 
\cite{pin,polak} in which the varieties of languages are not necessarily Boolean algebras. The approach in \cite{mikolaj} of considering algebras for a monad $\mathsf{T}$ is 
also considered and generalized in \cite{jste,urbat} 
as well as in the present paper. One of the main challenges in categorical approaches to Eilenberg--type correspondences is to define the right concept of a ``variety of languages''. 
The definition of a ``variety of languages'' given in \cite{urbat} depends of finding what they call a ``unary representation'', which is a set of unary operations on a free algebra 
satisfying certain properties, see \cite[Definition 3.7.]{urbat}. From this ``unary representation'' one can construct syntactic algebras and define the kind derivatives that define 
a ``variety of languages''. The definition of a ``variety of languages'' in the present paper is a categorical one which avoids the explicit definition of derivatives and existence of syntactic algebras. 
In the present paper, derivatives are captured coalgebraically and syntactic algebras are not used to prove the abstract Eilenberg--type correspondences 
theorems, but both of those concepts can be easily obtained via duality in each concrete case. 
Coalgebraic approaches, from which one can easily define the concept of a ``variety of languages'', are not used in \cite{mikolaj,urbat}. Another important related work is 
\cite{jan1}, in which an Eilenberg--type correspondence for varieties of monoids is shown, which is an Eilenberg--type correspondence that can be derived from 
the present paper but not from \cite{adamek,mikolaj,urbat}. The work made in \cite{jan1} motivates the study of Eilenberg--type correspondences for other classes of algebras different 
than pseudovarieties. It is worth mentioning that in \cite{jan1} the duality between equations and coequations is studied for the first time in the context of an Eilenberg--type correspondence.

\noindent All in all, the contributions of the present paper can be summarized as follows:
\begin{enumerate}[-]
	\item To unveil Eilenberg--type correspondences and show that:
		\[
			\text{Eilenberg--type correspondences = Birkhoff's theorem for (finite) algebras + duality.}
		\]
	\item To show and understand where ``varieties of languages'' come from, that is:
		\[
			\text{``varieties of languages'' = duals of (pseudo)equational theories.}
		\]
		This fact was conjectured by the author in \cite{jste}.
	\item To provide categorical versions of Birkhoff's theorem for (finite) $\mathsf{T}$--algebras as one--to--one correspondences between (pseudo)varieties of 
		$\mathsf{T}$--algebras and (pseudo)equational $\mathsf{T}$--theories, which are easily obtained from \cite{banab,bana}. 
		A categorical definition of a (pseudo)equa\-tional theory is given. 
	\item To show a categorical version of Birkhoff's theorem for finite algebras without the use of topology nor profinite techniques.
	\item To provide a general and abstract Eilenberg--type correspondence theorem that encompasses existing Eilenberg--type 
		correspondences from the literature. Not only for (local) pseudovarieties of algebras but also for (local) varieties of algebras.
	\item To derive Eilenberg--type correspondences without the use of syntactic algebras.
	\item To show that the notion of derivatives used to define the different kinds of ``varieties of languages'' in Eilenberg--type correspondences is exactly 
		the coalgebraic structure of an object, which is easily derived via duality, is most of the cases, from the notion of an algebra homomorphism.
\end{enumerate}

\noindent We present the content of our paper as follows: In Section 1, we fix some notation and some categorical facts we will use through the paper. In Section 2, we state  
an abstract Eilenberg--type correspondence for varieties of $\mathsf{T}$--algebras, which is derived from a categorical version of Birkhoff's theorem plus duality. 
As an application, we derive Eilenberg--type correspondences for some varieties of algebras. In Section 3, we redo Section 2, for the case of pseudovarieties of 
$\mathsf{T}$--algebras, i.e., we use Birkhoff's theorem for finite $\mathsf{T}$--algebras. In Section 4, we do a similar work for local (pseudo)varieties of $X$--generated 
$\mathsf{T}$--algebras whose proof easily follows from what is done in Section 2 and Section 3. Then we finish in Section 5 with the conclusions of this work.

\section{Preliminaries}\label{prel}

We introduce the notation and some facts that we will use in the paper. 
We assume that the reader is familiar with basic concepts from category theory and (co)algebra, see, e.g., \cite{awodey,januc}.

\noindent Given a category $\mathcal{D}$ and a monad $\mathsf{T}=(T,\eta,\mu)$ on $\mathcal{D}$, we denote 
the category of (Eilenberg--Moore) $\mathsf{T}$-algebras and their homomorphisms by $\Alg (\mathsf{T})$. Objects in $\Alg (\mathsf{T})$ are pairs $\mathbf{X}=(X,\alpha)$ 
where $X$ is an object in $\mathcal{D}$ and $\alpha\in \mathcal{D}(TX,X)$ is a morphism $\alpha:TX\to X$ in $\mathcal{D}$ that satisfies the identities 
$\alpha\circ \eta_X=id_X$ and $\alpha\circ T\alpha=\alpha\circ \mu_X$ . A homomorphism from a 
$\mathsf{T}$-algebra $\mathbf{X_1}=(X_1,\alpha_1)$ to a $\mathsf{T}$-algebra $\mathbf{X_2}=(X_2,\alpha_2)$ is a morphism $h\in \mathcal{D}(X_1,X_2)$ 
such that $h\circ \alpha_1=\alpha_2\circ Th$. 

\noindent Dually, given a category $\mathcal{C}$ and a comonad $\mathsf{B}=(B,\epsilon,\delta)$ on $\mathcal{C}$, 
$\Coalg (\mathsf{B})$ denotes the category of (Eilenberg--Moore) $\mathsf{B}$-coalgebras. Objects in $\Coalg (\mathsf{B})$ 
are pairs $\mathbf{Y}=(Y,\beta)$ where $Y$ is an object in $\mathcal{C}$ and $\beta\in \mathcal{C}(Y,BY)$ satisfies the identities 
$\epsilon_Y\circ \beta=id_Y$ and $B\beta\circ \beta=\delta_Y\circ \beta$. A homomorphism from a 
$\mathsf{B}$-coalgebra $\mathbf{Y_1}=(Y_1,\beta_1)$ to a $\mathsf{B}$-coalgebra $\mathbf{Y_2}=(Y_2,\beta_2)$ is a morphism $h\in \mathcal{C}(Y_1,Y_2)$ 
such that $\beta_2 \circ h=Bh\circ \beta_1$.

\noindent If $\mathcal{D}$ and $\mathcal{C}$ are dual categories and $\mathsf{T}$ is a monad on $\mathcal{D}$, then there is a canonical comonad $\mathsf{B}$ on $\mathcal{C}$ 
such that the duality between $\mathcal{D}$ and $\mathcal{C}$ lifts to their corresponding Eilenberg--Moore categories.

\begin{proposition}\label{monadtocomonad}
	\cite[Proposition 14]{sbr} 
	Let $F:\mathcal{C}\to \mathcal{D}$ and $G:\mathcal{D}\to \mathcal{C}$ be contravariant functors that form a duality 
	with natural isomorphisms $\eta^{GF}:Id_\mathcal{C}\Rightarrow GF$ and $\eta^{FG}:Id_\mathcal{D}\Rightarrow FG$. 
	Let $\mathsf{T}=(T,\eta, \mu)$ be a monad on $\mathcal{D}$. Then $\mathsf{B}=(B,\epsilon,\delta)$,
	where $B = GTF$ and $\epsilon, \delta$ are defined as:
	\begin{align*}
		\epsilon &= (GLF \xRightarrow{G\eta_{F}} GF \xRightarrow{(\eta^{GF})^{-1}} Id_\mathcal{C})  \\
		\delta &= (GLF \xRightarrow{G\mu_{F}} GLLF \xRightarrow{GL(\eta^{FG})^{-1}_{LF}} GLFGLF ),
	\end{align*}
	is a comonad on $\mathcal{C}$.
	Further, the duality between $F$ and $G$ lifts 
	to a duality between $\widehat{F}:\Coalg (\mathsf{B})\to \Alg (\mathsf{T})$ and 
	$\widehat{G}:\Alg (\mathsf{T})\to \Coalg (\mathsf{B})$.\qed
\end{proposition}

\noindent The following is a list of the categories we will use in the examples given in this paper:
\begin{center}
\begin{tabular}{|c|m{5cm}|m{7.5cm}|}
	\hline
	Category & Objects & Morphisms \\
	\hline
	$\set$ & Sets  & Functions \\
	\hline
	$\CABA$ & Complete atomic Boolean algebras & Complete Boolean algebra homomorphisms \\
	\hline 
	$\poset$ & Partially ordered sets & Order preserving functions \\
	\hline
	$\ACDL$ & Algebraic completely distributive lattices & Complete lattice homomorphisms \\
	\hline
	$\fvec$ & $\mathbb{K}$--vector spaces & Linear maps \\
	\hline
	$\tvec$ & Topological $\mathbb{K}$--vector spaces that are Stone spaces, i.e, they have compact, Hausdorff and zero dimensional topology & Linear continuous maps \\
	\hline 
	$\JSL$ & Join semilattices with zero& Join semilattice homomorphisms preserving zero \\
	\hline 
	$\TBJSL$ & Topological join semilattices with zero that are Stone spaces  & Continuous join semilattice homomorphisms preserving zero\\
	\hline	
\end{tabular}
\end{center}
We will use the facts that $\set$ is dual to $\CABA$, $\poset$ is dual to $\ACDL$, $\fvec$ is dual to $\tvec$ for a finite field $\mathbb{K}$ and that $\JSL$ is dual to $\TBJSL$. 
For a given concrete category $\mathcal{C}$, we denote the full subcategory of $\mathcal{C}$ consisting of its finite objects by $\mathcal{C}_f$. In the case of $\Alg(\mathsf{T})$ and 
$\Coalg(\mathsf{B})$,  we denote them by $\Alg_f(\mathsf{T})$ and $\Coalg_f(\mathsf{B})$, respectively.

\noindent Let $\mathcal{D}$ be a category and let $\mathscr{E}$ and $\mathscr{M}$ be classes of morphisms in $\mathcal{D}$. $\mathscr{E}/\mathscr{M}$ is called a 
{\it factorization system} on $\mathcal{D}$ if:
\begin{enumerate}[i)]
	\item Each of $\mathscr{E}$ and $\mathscr{M}$ is closed under composition with isomorphisms,
	\item Every morphism $f$ in $\mathcal{D}$ has a factorization $f=m\circ e$, with $e\in \mathscr{E}$ and $m\in \mathscr{M}$.
	\item Given any commutative diagram 
			\begin{center}	
				\begin{tikzpicture}[>=stealth,shorten >=3pt,
						node distance=2.5cm,on grid,auto,initial text=,
						accepting/.style={thick,double distance=2.5pt}]
						\node (a) at (0,0) {$\cdot$};
						\node (b) at (2,0) {$\cdot$};
						\node (c) at (0,1.5) {$\cdot$};
						\node (d) at (2,1.5) {$\cdot$};
						\path[right hook->] (a) edge [] node [below] {$m$} (b);
						\path[->>] (c) edge [] node [above] {$e$} (d);
						\path[->] (c) edge [] node [left] {$f$} (a)
								(d) edge [] node [right] {$g$} (b);
				\end{tikzpicture}
			\end{center}
			with $e\in \mathscr{E}$ and $m\in \mathscr{M}$, there is a unique {\it diagonal fill--in}, i.e., a unique morphism $d$ such that the following diagram commutes:
			\begin{center}	
				\begin{tikzpicture}[>=stealth,shorten >=3pt,
						node distance=2.5cm,on grid,auto,initial text=,
						accepting/.style={thick,double distance=2.5pt}]
						\node (a) at (0,0) {$\cdot$};
						\node (b) at (2,0) {$\cdot$};
						\node (c) at (0,1.5) {$\cdot$};
						\node (d) at (2,1.5) {$\cdot$};
						\path[right hook->] (a) edge [] node [below] {$m$} (b);
						\path[->>] (c) edge [] node [above] {$e$} (d);
						\path[->] (c) edge [] node [left] {$f$} (a)
								(d) edge [] node [above] {$d$} (a)
								(d) edge [] node [right] {$g$} (b);
				\end{tikzpicture}
			\end{center}
\end{enumerate}
A factorization system $\mathscr{E}/\mathscr{M}$ is {\it proper} if every morphism in $\mathscr{E}$ is epi and every morphism in $\mathscr{M}$ is mono. 
We will use the following facts about factorization systems \cite{ahs}.
\begin{lemma}\label{lemmafs}
	Let $\mathcal{D}$ be a category and $\mathscr{E}/\mathscr{M}$ be a factorization system on $\mathcal{D}$ such that every morphism in 
	$\mathscr{M}$ is mono. Then $f\circ g\in \mathscr{E}$ implies $f\in \mathscr{E}$.\qed
\end{lemma}

\begin{lemma}\label{lemmafs1}
	Let $\mathcal{D}$ be a category, $\mathsf{T}=(T,\eta,\mu)$ a monad on $\mathcal{D}$ and $\mathscr{E}/\mathscr{M}$ a proper factorization system on 
	$\mathcal{D}$. If $T$ preserves the morphisms in $\mathscr{E}$ then 
	$\Alg(\mathsf{T})$ inherits the same $\mathscr{E}/\mathscr{M}$ factorization system.\qed
\end{lemma}


\section{Eilenberg--type correspondences for varieties of $\mathsf{T}$--algebras}\label{secbirk}

Varieties of algebras have been studied in universal algebra and equational logic. In particular, Birkhoff's variety theorem (see, e.g., \cite{birk,bys}) states that 
a class of algebras of the same type is a variety, i.e., it is closed under homomorphic images, subalgebras and (not necessarily finite) products, if and only if 
it is definable by equations. As a consequence, for a fixed type of algebras, we get a one--to--one correspondence between varieties of algebras 
and equational theories. Birkhoff's theorem has been generalized to a categorical level, see, e.g., \cite{ahs,awodeyh,bana,barr}, to characterize subcategories of a given category that are, 
in some sense, equationally defined. In this section, we provide, under mild assumptions, a Birkhoff's theorem for varieties of $\mathsf{T}$--algebras, Theorem \ref{birkthm}, which, 
in order to derive Eilenberg--type correspondences, will be stated as a one--to--one correspondence between varieties of $\mathsf{T}$--algebras and equational $\mathsf{T}$--theories. 
A categorical definition of an equational $\mathsf{T}$--theory will be given. 

\noindent Next we dualize the categorical definition of an equational $\mathsf{T}$--theory to get that of a coequational $\mathsf{B}$--theory, where $\mathsf{B}$ is a 
comonad on a category $\mathcal{C}$. Our main contribution is to note that particular instances of coequational $\mathsf{B}$--theories have been already studied 
in the literature under the name of ``varieties of languages'' to establish Eilenberg--type correspondences, e.g., \cite[Theorem 39]{jan1}. Thus, if we assume that $\mathcal{D}$ 
and $\mathcal{C}$ are dual categories and that the comonad $\mathsf{B}$ is the dual of the monad $\mathsf{T}$, as in Proposition \ref{monadtocomonad}, then, by duality, 
we get a one--to--one correspondence between equational $\mathsf{T}$--theories and coequational $\mathsf{B}$--theories. All in all, we get a one--to--one correspondence 
between varieties of $\mathsf{T}$--algebras and coequational $\mathsf{B}$--theories, which is the abstract Eilenberg--type correspondence for 
varieties of $\mathsf{T}$--algebras, Proposition \ref{eilvar}. The main facts in this section can be summarized in the following picture:
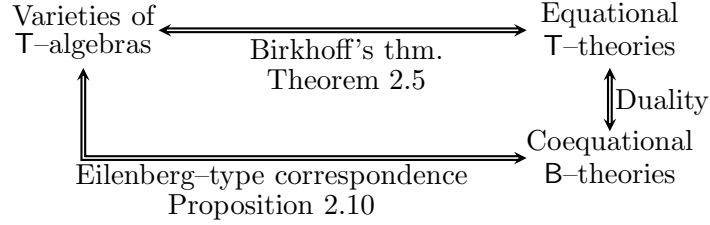
\begin{center}	
	\begin{tikzpicture}[>=stealth,shorten >=3pt,
		node distance=2.5cm,on grid,auto,initial text=,
		accepting/.style={thick,double distance=2.5pt}]
		\node () at (0,-0.05) {Varieties of};
		\node () at (0,-0.45) {$\mathsf{T}$--algebras};
		\draw[double,thick,<->] (1,-0.25) -- (6,-0.25);
		\node () at (3.5,-0.5) {Birkhoff's thm.};
		\node () at (3.5,-0.9) {Theorem \ref{birkthm}};

		\node () at (7,-0.05) {Equational};
		\node () at (7,-0.45) {$\mathsf{T}$--theories};

		\draw[double,thick,<->] (7,-0.75) -- (7,-1.7);
		\node () at (7.7,-1.2) {Duality};
		\node () at (7,-1.75) {Coequational};
		\node () at (7,-2.15) {$\mathsf{B}$--theories};
		\draw[double,thick,<->] (0,-0.75) -- (0,-1.95)--(6,-1.95);
		\node () at (2.5,-2.2) {Eilenberg--type correspondence};
		\node () at (2.5,-2.6) {Proposition \ref{eilvar}};
	\end{tikzpicture}
\end{center}

\noindent where each arrow symbolizes a one--to--one correspondence and $\mathsf{B}$ is the comonad that is the dual of the monad $\mathsf{T}$.

\subsection{Birkhoff's Theorem for $\mathsf{T}$--algebras}\label{seceqth}

The concept of a variety of algebras and an equationally defined class can be formulated in categorical terms to prove a categorical version of Birkhoff's theorem \cite{ahs,awodeyh,bana,barr}. 
We prove in this subsection that, under mild assumptions, there is a one--to--one correspondence between varieties of $\mathsf{T}$--algebras and equational $\mathsf{T}$--theories, 
Theorem \ref{birkthm}, where $\mathsf{T}$ is a monad on a category $\mathcal{D}$. The definition of variety of $\mathsf{T}$--algebras, which depends on the concept of 
homomorphic images and subalgebras, will be defined by using a factorization system $\mathscr{E}/\mathscr{M}$ on $\mathcal{D}$. In order to define the concept of 
equational $\mathsf{T}$--theories, we base our approach on \cite[Definition II.14.16]{bys}. After providing the assumptions and basic definitions needed to state Theorem \ref{birkthm}, 
its proof will easily follow from the assumptions needed by using the work by Banaschewski and Herrlich \cite{bana}.

\noindent We fix a complete category $\mathcal{D}$, a monad $\mathsf{T}=(T,\eta,\mu)$ on $\mathcal{D}$, a factorization system $\mathscr{E}/\mathscr{M}$ on $\mathcal{D}$ 
and a full subcategory $\mathcal{D}_0$ of $\mathcal{D}$. We will use the following assumptions:
\begin{enumerate}
	\item[(B1)] The factorization system $\mathscr{E}/\mathscr{M}$ is proper. That is, every map in $\mathscr{E}$ is an epimorphism and every map in $\mathscr{M}$ is a monomorphism.
	\item[(B2)] For every $X\in \mathcal{D}_0$, the free $\mathsf{T}$--algebra $\mathbf{TX}=(TX,\mu_X)$ is 
				{\it projective with respect to $\mathscr{E}$ in $\Alg(\mathsf{T})$}. That is, for every $h\in \Alg(\mathsf{T})(\mathbf{TX},\mathbf{B})$ with $X\in \mathcal{D}_0$ and 
				$e\in \Alg (\mathsf{T})(\mathbf{A},\mathbf{B})\cap \mathscr{E}$ there exists $g\in \Alg(\mathsf{T})(\mathbf{TX},\mathbf{A})$ such that the following diagram commutes:
				\begin{center}	
				\begin{tikzpicture}[>=stealth,shorten >=3pt,
						node distance=2.5cm,on grid,auto,initial text=,
						accepting/.style={thick,double distance=2.5pt}]
						\node (a) at (0,0) {$B$};
						\node (b) at (3,0) {$A$};
						\node (c) at (3,1.5) {$TX$};
						\path[<<-] (a) edge [] node [below] {$e$} (b);
						\path[<-] (a) edge [] node [above] {$h$} (c);
						\path[<-,dashed] (b) edge [] node [right] {$g$} (c);
				\end{tikzpicture}
				\end{center}
	\item [(B3)] For every $\mathbf{A}\in \Alg(\mathsf{T})$ there exists $X_A\in\mathcal{D}_0$ and $s_A\in \Alg(\mathsf{T})(\mathbf{TX_A},\mathbf{A})\cap \mathscr{E}$.
	\item[(B4)] $T$ preserves morphisms in $\mathscr{E}$.
	\item[(B5)] For every $X\in \mathcal{D}_0$, there is, up to isomorphism, only a set of $\mathsf{T}$--algebra morphisms in $\mathscr{E}$ with domain $\mathbf{TX}$.
\end{enumerate}

\noindent The notion of a variety of $\mathsf{T}$--algebras, which depends on the concept of homomorphic images and subalgebras, will be defined by using the factorization system 
$\mathscr{E}/\mathscr{M}$ on $\mathcal{D}$, which is lifted to $\Alg(\mathsf{T})$ using (B1) and (B4), Lemma \ref{lemmafs1}. The role of $\mathcal{D}_0$ is that the objects from which ``variables'' for the equations are considered are objects in $\mathcal{D}_0$. Assumption (B2) of $\mathbf{TX}$ being projective with respect to $\mathscr{E}$, 
$X\in \mathcal{D}_0$, will play a fundamental role in relating varieties of algebras with equational theories. 
Assumption (B3) guarantees that every algebra in $\Alg(\mathsf{T})$ is the homomorphic image of a free $\mathsf{T}$--algebra with object of generators from $\mathcal{D}_0$. 
Condition (B5) will allow us to define the equational theory for a given variety of algebras. 

\noindent For Birkhoff's classical variety theorem \cite{birk}, we can take $\mathcal{D}=\mathcal{D}_0=\set$, $\mathscr{E}=$ surjections, 
$\mathscr{M}=$ injections, and $\mathsf{T}$ to be the term monad for a given type of algebras $\T$, i.e., $TX=T_{\T}(X)$, the set of terms of type 
$\T$ on the set of variables $X$ (see Example \ref{excleqth} and Example \ref{exbirkhoff}). Another important example will be given by $\mathcal{D}=\poset$, with 
$\mathcal{D}_0=$ discrete posets (i.e., we do not want the ``variables'' to be ordered) to obtain a Birkhoff theorem for ordered algebras \cite{bloom}. 

\noindent We now give the necessary definitions to formulate our Birkhoff theorem for $\mathsf{T}$--algebras. We start by defining varieties of $\mathsf{T}$--algebras.

\begin{defi}
	Let $\mathcal{D}$ be a complete category, $\mathsf{T}$ a monad on $\mathcal{D}$ and $\mathscr{E}/\mathscr{M}$ a factorization system on $\mathcal{D}$. 
	Let $K$ be a class of algebras in $\Alg(\mathsf{T})$. We say that $K$ {\it is closed under $\mathscr{E}$--quotients} if $\mathbf{B}\in \Alg(\mathsf{T})$ for every  
	$e\in \Alg(\mathsf{T})(\mathbf{A},\mathbf{B})\cap \mathscr{E}$ with $\mathbf{A}\in K$. We say that $K$ {\it is closed under $\mathscr{M}$--subalgebras} if 
	$\mathbf{B}\in \Alg(\mathsf{T})$ for every $m\in \Alg(\mathsf{T})(\mathbf{B},\mathbf{A})\cap \mathscr{M}$ with $\mathbf{A}\in K$. We say that 
	$K$ {\it is closed under products} if $\prod_{i\in I} \mathbf{A}_i\in K$ for every set $I$ such that $\mathbf{A}_i\in K$, $i\in I$.
	A class $V$ of algebras in $\Alg (\mathsf{T})$ is called a {\it variety of $\mathsf{T}$--algebras } if it is closed under $\mathscr{E}$--quotients, 
	$\mathscr{M}$--subalgebras and products. 
\end{defi}

\noindent Now, we define the other main concept to state our Bikhoff's theorem for $\mathsf{T}$--algebras, namely, the concept of an equational $\mathsf{T}$--theory.

\begin{defi}\label{defeqth}
	Let $\mathcal{D}$ be a category, $\mathsf{T}$ a monad on $\mathcal{D}$, $\mathcal{D}_0$ a full subcategory of $\mathcal{D}$ and $\mathscr{E}/\mathscr{M}$ a 
	factorization system on $\mathcal{D}$. An {\it equational $\mathsf{T}$--theory on $\mathcal{D}_0$} is a family of $\mathsf{T}$--algebra morphisms  
	$\E=\{TX\overset{e_X}{\lepi} Q_X\}_{X\in \mathcal{D}_0}$ in $\mathscr{E}$ such that 	for any $X,Y\in \mathcal{D}_0$ and any $g\in \Alg(\mathsf{T})(\mathbf{TX},\mathbf{TY})$ 
	there exists $g'\in \Alg(\mathsf{T})(\mathbf{Q_X},\mathbf{Q_Y})$ such that the following diagram 
	commutes:
	\begin{center}	
		\begin{tikzpicture}[>=stealth,shorten >=3pt,
				node distance=2.5cm,on grid,auto,initial text=,
				accepting/.style={thick,double distance=2.5pt}]
				\node (a) at (2,1.5) {$TY$};
				\node (b) at (2,0) {$Q_Y$};
				\node (c) at (0,1.5) {$TX$};
				\node (d) at (0,0) {$Q_X$};
				\path[->>] (a) edge [] node [right] {$e_Y$} (b)
						(c) edge [] node [left] {$e_X$} (d);
				\path[->] (c) edge [] node [above] {$\forall g$} (a);
				\path[->,dashed] (d) edge [] node [below] {$g'$} (b);
		\end{tikzpicture}
	\end{center}	
\end{defi}

\noindent Intuitively, in the setting of Birkhoff's classical variety theorem, for every object $X\in \mathcal{D}_0$ (i.e., a set of variables) the morphism $e_X$, which we asssume to be a 
surjection, 
represents the set of equations $\ker(e_X)$, which is a congruence on $\mathbf{TX}$, i.e., it is an equivalence relation on $TX$ which is closed under the componentwise algebric operations. 
Commutativity of the diagram above means that the family of all equations $\{\ker(e_X)\}_{X\in \mathcal{D}_0}$ is closed under any substitution 
$g\in \Alg(\mathsf{T})(\mathbf{TX},\mathbf{TY})$. The previous definition generalizes the definition of an equational theory to a categorical level, cf. \cite[Definition II.14.16]{bys}.

\begin{example}\label{excleqth}
	Consider the case $\mathcal{D}=\mathcal{D}_0=\set$, $\mathscr{E}=$ surjections and $\mathscr{M}=$ injections. For a given type of algebras $\T$, consider the monad 
	$\mathsf{T}_{\T}=(T_{\T},\eta,\mu)$ such that $T_{\T}(X)$ is the set of terms for $\T$ on variables $X$, see \cite[Definition II.10.1]{bys}. The unit $\eta_X:X\to T_{\T}(X)$ is the inclusion 
	function and multiplication $\mu_X:T_{\T}(T_{\T}(X))\to T_{\T}(X)$ is the identity map. Now, $\Alg(\mathsf{T}_{\T})$ is the category of algebras $\mathbf{A}=(A,\alpha)$ of type $\T$, 
	where $\alpha: T_{\T}(A)\to A$ is the evaluation $\alpha(t)$ in $A$ of each term $t\in T_{\T}(A)$. An equational $\mathsf{T}_{\T}$--theory on $\mathcal{D}_0=\set$ is a family of 
	surjective homomorphisms  $\E=\{T_{\T}(X)\overset{e_X}{\lepi} Q_X\}_{X\in \set}$ in $\Alg(\mathsf{T}_{\T})$ such that every $\ker(e_x)$ is a congruence on $T_{\T}(X)$ 
	and the family $\{\ker(e_X)\}_{X\in \mathcal{D}_0}$ is closed under substitution, i.e., for $\left(p(x_1,\ldots,x_n),q(x_1,\ldots,x_n)\right)\in \ker(e_X)$ and 
	$r_x\in T_{\T}(Y)$, $x\in X$, we have that $(p(r_{x_1},\ldots,r_{x_n}),q(r_{x_1},\ldots,r_{x_n}))\in \ker(e_Y)$, where $t(r_{x_1},\ldots,r_{x_n})$ is the term in 
	$T_{\T}(Y)$ obtained from $t(x_1,\ldots,x_n)\in T_{\T}(X)$ by replacing each variable $x_i$ by $r_{x_i}$, $i=1,\ldots,n$.\qed
\end{example}

\begin{example}\label{exineqth}
	Consider the case $\mathcal{D}=\poset$, $\mathcal{D}_0=$ the full subcategory of discrete posets, $\mathscr{E}=$ surjections and $\mathscr{M}=$ embeddings. 
	Let $\T$ be a type of algebras. An {\it ordered algebra of type $\T$} is a triple $A=(A,\leq_A,\{f_A:A^{n_f}\to A\}_{f\in \T})$ such that $(A,\leq_A)\in \poset$ and all the 
	functions $f_A:A^{n_f}\to A$ are 	order preserving, where the order in $A^{n_f}$ is componentwise, $f\in \T$. We can define the monad $\mathsf{T}_{\T}=(T_{\T},\eta,\mu)$ 
	where $T_{\T}(X,\leq_X)$ is the poset $(T_{\T}(X),\leq_{\mathsf{T}_{\T}(X)})$ defined as: $x\leq_{\mathsf{T}_{\T}(X)} y$ for every $x,y\in X$ such that $x\leq_X y$, and 
	$f(t_1,\ldots,t_{n_f})\leq_{\mathsf{T}_{\T}(X)} f(q_1,\ldots,q_{n_f})$ for every $f\in \T$ and terms $t_i,q_i\in T_{\T}(X)$ such that $t_i\leq_{\mathsf{T}_{\T}(X)} q_i$, $i=1,\ldots,n_f$. 
	Algebras in $\Alg(\mathsf{T}_{\T})$ are ordered algebras of type $\T$.

	\noindent An equational $\mathsf{T}_{\T}$--theory is a family $\E=\{T_{\T}(X)\overset{e_X}{\lepi} Q_X\}_{X\in \mathcal{D}_0}$ of surjective homomorphisms, which are trivially 
	order preserving since $T_{\T}(X)$ is discrete for any $X\in \mathcal{D}_0$, such that $\overrightarrow{\ker}(e_X)$ is an admissible preorder on $T_{\T}(X)$ 
	\footnote{A preorder $\sqsubseteq$ on an ordered algebra $(A,\leq_A,\{f_A:A^{n_f}\to A\}_{f\in \T})$ of type $\T$ is {\it compatible} if for every 
	$f\in \T$ and $a_i,b_i\in A$ with $a_i\sqsubseteq b_i$, $i=1,\ldots, n_f$, 	we have that $f_A(a_1,\ldots,a_{n_f})\sqsubseteq f_A(b_1,\ldots,b_{n_f})$. A preorder $\sqsubseteq$ 
	is {\it admissible} if it is compatible and $a\sqsubseteq b$ whenever $a\leq_A b$. The congruence $\theta_\sqsubseteq$ on $A$ induced by the compatible preorder 
	$\sqsubseteq$ is the relation $\theta_\sqsubseteq$ on $A$ defined as $\theta_\sqsubseteq:=\sqsubseteq\cap \sqsubseteq^{-1}$. Then $(A/\theta_\sqsubseteq,\leq_\sqsubseteq)$ 
	is an ordered algebra with the order given by $[x]\leq_\sqsubseteq [y]$ iff $x\sqsubseteq y$. See \cite{bloom}.}, where $\overrightarrow{\ker}(e_X):=\{(u,v)\mid e_X(u)\leq e_X(v)\}$, 
	and the family $\{\overrightarrow{\ker}(e_x)\}_{X\in \mathcal{D}_0}$ is closed under substitution as in the previous example. 
	In this case, $\overrightarrow{\ker}(e_X)$ represents the equations and inequations of terms with variables in $X$ in the equational $\mathsf{T}_{\T}$--theory. 
	Note that if we take $\mathcal{D}_0=\poset$ then condition (B2) does not hold.\qed
\end{example}

\noindent Given an equational $\mathsf{T}$--theory $\E=\{TX\overset{e_X}{\lepi} Q_X\}_{X\in \mathcal{D}_0}$ and an algebra $\mathbf{A}\in \Alg(\mathsf{T})$, we say that 
$\mathbf{A}$ {\it satisfies} $\E$, denoted as $\mathbf{A}\models \E$, if $\mathbf{A}$ is {\it $\E$--injective}, that is, if for every $X\in \mathcal{D}_0$ and 
every $f\in \Alg(\mathsf{T}) (\mathbf{TX},\mathbf{A})$ there exists a $\mathsf{T}$--algebra morphism $g_f$ such that $f=g_f\circ e_X$.

\noindent Intuitively, $\mathbf{A}\models \E$ if for every assignment $f\in \Alg(\mathsf{T})(\mathbf{TX},\mathbf{A})$ of the 
variables $X$ to elements of the algebra $\mathbf{A}$, all the equations represented by $e_X\colon TX \lepi Q_X$ hold in $\mathbf{A}$. Given an equational $\mathsf{T}$--theory 
$\E$ we denote the {\it models of} $\E$ by $\md (\E)$, that is:
\[
	\md (\E):=\{\mathbf{A}\in \Alg(\mathsf{T}) \mid \mathbf{A}\models \E\}
\]
A class $K$ of $\mathsf{T}$--algebras is {\it defined} by $\E$ if $K=\md (\E)$.

\begin{thm}[Birkhoff's Theorem for $\mathsf{T}$--algebras]\label{birkthm}
	Let $\mathcal{D}$ be a complete category, $\mathsf{T}$ a monad on $\mathcal{D}$, $\mathcal{D}_0$ a full subcategory of $\mathcal{D}$ and $\mathscr{E}/\mathscr{M}$ 
	a factorization system on $\mathcal{D}$. Assume (B1) to (B5). Then a class $K$ of $\mathsf{T}$--algebras is a variety of $\mathsf{T}$--algebras if and only if 
	it is defined by an equational $\mathsf{T}$--theory on $\mathcal{D}_0$. Additionally, varieties of $\mathsf{T}$--algebras are in one--to--one correspondence with 
	equational $\mathsf{T}$--theories on $\mathcal{D}_0$. \qed
\end{thm}

\noindent From the previous theorem we have the following.

\begin{example}\label{exbirkhoff}
	By considering the monad and the categories given in Example \ref{excleqth} we obtain the classical Birkhoff variety theorem \cite{birk}.\qed
\end{example}

\begin{example}\label{exorderedbirkhoff}
	By considering the monad and the categories given in Example \ref{exineqth} we obtain the Birkhoff variety theorem for ordered algebras \cite{bloom}.\qed
\end{example}

\begin{example}[cf. \text{\cite[Theorem 39]{jan1}}]\label{exjan1}
	Consider the case $\mathcal{D}=\mathcal{D}_0=\set$, $\mathsf{T}$ the monad given by $TX=X^*$, where $X^*$ is the free monoid on $X$, $\mathscr{E}=$ surjections 
	and $\mathscr{M}=$ injections. We have that conditions (B1) to (B5) are fullfilled. Therefore we have a one--to--one correspondence between varieties of monoids and 
	equational $\mathsf{T}$--theories. Now, consider the category $\mathcal{C}=\mathcal{C}_0=\CABA$ which is dual to $\set$ and let $\mathsf{B}$ be the comonad on $\CABA$ 
	that is dual to the 
	monad $\mathsf{T}$ on $\set$, i.e, $B$ is defined, up to isomorphism, as $B(2^X)=2^{X^*}$. Then, by duality, we have a one--to--one correspondence between equational 
	$\mathsf{T}$--theories $\E$ and its dual $\E^\partial$, i.e., families of 
	monomorphisms $\{S_X\xhookrightarrow{m_X} BX \}_{X\in \mathcal{C}_0=\mathcal{C}}$ in $\Coalg(\mathsf{B})$ such that for any $X,Y\in \mathcal{C}_0$ and any 
	$g\in \Coalg(\mathsf{B})(\mathbf{BX},\mathbf{BY})$ there exists $g'\in \Coalg(\mathsf{B})(\mathbf{S_X},\mathbf{S_Y})$ such that $m_Y\circ g=g'\circ m_X$.

	\noindent This notion of the dual of an equational $\mathsf{T}$--theory is equivalent with the --more complicated-- notion of a variety of languages in \cite[Definition 35]{jan1} 
	(see Example \ref{exjan11} for more details). In this setting, we get a one--to--one correspondence between varieties of monoids and duals of equational 
	$\mathsf{T}$--theories, i.e., varieties of languages. This is exactly the Eilenberg--type theorem \cite[Theorem 39]{jan1}.\qed
\end{example}

\subsection{Abstract Eilenberg--type correspondence for varieties of $\mathsf{T}$--algebras}\label{seceilv}

As we saw in Example \ref{exjan1}, by Using Theorem \ref{birkthm} and duality, we can derive Eilenberg--type correspondences for varieties of $\mathsf{T}$--algebras. 
In this subsection we state the abstract Eilenberg--type correspondence for varieties of $\mathsf{T}$--algebras, i.e., a one--to--one correspondence between 
varieties of $\mathsf{T}$--algebras and duals of equational $\mathsf{T}$--theories, and then instantiate this result to some particular cases. 
We dualize the definition of an equational $\mathsf{T}$--theory as follows.

\begin{defi}\label{defcoeqth}
	Let $\mathcal{C}$ be a category, $\mathsf{B}=(B,\epsilon,\delta)$ a comonad on $\mathcal{C}$, $\mathscr{E}/\mathscr{M}$ a factorization system on $\mathcal{C}$ and 
	$\mathcal{C}_0$ a full subcategory of $\mathcal{C}$. A {\it coequational $\mathsf{B}$--theory on $\mathcal{C}_0$} is a family of $\mathsf{B}$--coalgebra morphisms 
	$\M=\{S_Y\xhookrightarrow{m_Y} BY\}_{Y\in \mathcal{C}_0}$ in $\mathscr{M}$ such that 
	for any $X,Y\in \mathcal{C}_0$ and any $g\in \Coalg(\mathsf{B})(\mathbf{BX},\mathbf{BY})$ there exists 	$g'\in \Coalg(\mathsf{B})(\mathbf{S_X},\mathbf{S_Y})$ such that 
	the following diagram 
	commutes:
	\begin{center}
		\begin{tikzpicture}[>=stealth,shorten >=3pt,
				node distance=2.5cm,on grid,auto,initial text=,
				accepting/.style={thick,double distance=2.5pt}]
				\node (a) at (2,0) {$S_Y$};
				\node (b) at (2,1.5) {$BY$};
				\node (c) at (0,0) {$S_X$};
				\node (d) at (0,1.5) {$BX$};
				\path[right hook->] (a) edge [] node [right] {$m_Y$} (b)
						(c) edge [] node [left] {$m_X$} (d);
				\path[->,dashed] (c) edge [] node [below] {$g'$} (a);
				\path[->] (d) edge [] node [above] {$\forall g$} (b);
		\end{tikzpicture}
	\end{center}
\end{defi}

\noindent Intuitively, every $\mathbf{S_Y}$ is a $\mathsf{B}$--subcoalgebra of the cofree coalgebra $\mathbf{BY}=(BY,\delta_Y)$, and the family $\{S_Y\}_{Y\in \mathcal{C}_0}$ is 
closed under any coalgebra morphism, i.e., for every morphism $g\in \Coalg(\mathsf{B})(\mathbf{BX},\mathbf{BY})$, $x\in S_X$ implies $g(x)\in S_Y$. As an example of coequational 
$\mathsf{B}$--theories we have the ``varieties of languages'' defined in \cite[Definition 35]{jan1} which we describe in a simpler way in Example \ref{exjan11}. 
Now, with the previous definition, Theorem \ref{birkthm} and duality, we have the following.

\begin{prop}[Abstract Eilenberg--type correspondence for varieties of $\mathsf{T}$--algebras]\label{eilvar}
	Let $\mathcal{D}$ be a complete category, $\mathsf{T}$ a monad on $\mathcal{D}$, $\mathscr{E}/\mathscr{M}$ a factorization system on $\mathcal{D}$ and $\mathcal{D}_0$ a 
	full subcategory of $\mathcal{D}$. Assume (B1) to (B5). Let $\mathcal{C}$ be a category that is dual to $\mathcal{D}$, $\mathcal{C}_0$ the  
	corresponding dual category of $\mathcal{D}_0$ and let $\mathsf{B}$ be the comonad on $\mathcal{C}$ that is dual to $\mathsf{T}$ which is defined as in 
	Proposition \ref{monadtocomonad}. 
	Then there is a one--to--one correspondence between varieties of $\mathsf{T}$--algebras and coequational $\mathsf{B}$--theories on $\mathcal{C}_0$.
\end{prop}

\noindent In the rest of this section we show some particular instances of Eilenberg--type correspondences derived from the the previous proposition. We start with the 
continuation of Example \ref{exjan1} by describing the defining properties 
of the coequational $\mathsf{B}$--theories that correspond to varieties of monoids. Then we do a similar work to derive Eilenberg--type correspondences for other kind of varieties 
such as semigroups, groups, monoid actions, ordered monoids, vector spaces and idempotent semirings. It is worth mentioning that these kind of results, except for the one in 
Example \ref{exjan11} which was shown in \cite[Theorem 39]{jan1}, seem to be new and cannot be derived with categorical approaches given in \cite{adamek,mikolaj,urbat}.

\begin{example}[Example \ref{exjan1} continued]\label{exjan11}
	From the setting of Example \ref{exjan1} we get a one--to--one correspondence between varieties of monoids and coequational $\mathsf{B}$--theories on $\CABA$. 
	The latter can be characterized as operators $\Lan$ on $\set$ such that for every $X\in \set$:
	\begin{enumerate}[i)]
		\item $\Lan(X)\in \CABA$ and it is a subalgebra of the complete atomic Boolean algebra $\set(X^*,2)$ of subsets of $X^*$, i.e., every element in $\Lan (X)$ is a language on $X$.
		\item $\Lan(X)$ is closed under left and right derivatives. That is, if $L\in \Lan(X)$ and $x\in X$ then $_xL,L_x\in \Lan(X)$, where $_xL(w)=L(wx)$ and $L_x(w)=L(xw)$, 
			$w\in X^*$.
		\item $\Lan$ is closed under morphic preimages. That is, for every $Y\in \set$, homomorphism of monoids $h:Y^*\to X^*$ and $L\in \Lan(X)$, we have that 
			$L\circ h\in \Lan(Y)$.
	\end{enumerate}

	\noindent The previous notion of a coequational $\mathsf{B}$--theory on $\CABA$ is equivalent with the --more complicated-- notion of a variety of languages in 
	\cite[Definition 35]{jan1}. This one--to--one correspondence is exactly the Eilenberg--type theorem \cite[Theorem 39]{jan1} (see Example \ref{exua} and Appendix for more details).\qed
\end{example}

\noindent The following example describes explicitely Eilenberg--type correspondences for varieties of algebras of any given type $\T$ where each function symbol in $\T$ has finite arity. 

\begin{example}\label{exua}
	Let $\T$ be a type of algebras where each function symbol $g\in \T$ has arity $n_{g}\in \mathbb{N}$ and let $K$ be a variety of algebras of type $\T$. Consider the case 
	$\mathcal{D}=\mathcal{D}_0=\set$, $\mathscr{E}=$ surjections, $\mathscr{M}=$ injections and let $\mathsf{T}_K$ be the monad such that for every $X\in \set$, 
	$T_KX$ is the underlying set of the free algebra in $K$ on $X$ generators (see \cite[Definition II.10.9]{bys} and \cite[VI.8]{maclane}). 
	We have that $\CABA$ is dual to $\set$, so we can consider 
	$\mathcal{C}=\mathcal{C}_0=\CABA$. By using the duality between $\CABA$ and $\set$, each coequational $\mathsf{B}$--theory can be indexed by $\mathcal{D}_0=\set$ and 
	can be presented, up to isomorphism, as a family $\{S_X\xhookrightarrow{m_X} 2^{T_KX}\}_{X\in \set}$. From this, we present coequational $\mathsf{B}$--theories as 
	operators $\Lan$ on $\set$ given by $\Lan(X):=\Imag(m_X)$. Then, we get a one--to--one correspondence between 
	varieties of algebras in $K$ and operators $\Lan$ on $\set$ such that for every $X\in \set$:
	\begin{enumerate}[i)]
		\item $\Lan(X)\in \CABA$ and it is a subalgebra of the complete atomic Boolean algebra $\set(T_KX,2)$ of subsets of $T_KX$.
		\item $\Lan(X)$ is closed under derivatives with respect to the type $\T$. That is, for every $g \in \T$ of arity $n_g$, every $1\leq i\leq n_g$, 
			every $t_j\in T_KX$, $1\leq j< n_g$, and every $L\in \Lan(X)$ we have that $L^{(i)}_{(g,t_1,\ldots ,t_{n_g-1})}\in \Lan(X)$ where 
			$L^{(i)}_{(g,t_1,\ldots ,t_{n_g-1})}\in \set(T_KX,2)$ 
			is defined as 
			\[
				L^{(i)}_{(g,t_1,\ldots ,t_{n_g-1})}(t)=L(g(t_1,\ldots,t_{i-1},t,t_{i},\ldots,t_{n_g-1}))
			\]
			$t\in T_KX$. That is, for every function symbol $g\in \T$ we get $n_g$ kinds of derivatives.
		\item $\Lan$ is closed under morphic preimages. That is, for every $Y\in \set$, homomorphism of $\mathsf{T}_K$--algebras $h:T_KY\to T_KX$ and $L\in \Lan(X)$, we have 
			that $L\circ h\in \Lan(Y)$.
	\end{enumerate}
	In fact, we have the following:
	\begin{enumerate}[a)]
		\item Condition i) above follows from the fact that $\Lan(X)=\Imag(m_X)\cong S_X\in \CABA$ and $\Imag(m_X)\subseteq \set(T_KX,2)$.
		\item Condition ii) above follows from lifting the duality between $\set$ and $\CABA$ to a duality between $\Alg(\mathsf{T_K})$ and $\Coalg(\mathsf{B})$. In fact, every 
			surjective $\mathsf{T}_K$--algebra morphism $e_X:T_KX\lepi Q_X$ defines the injective morphism $\set(e_X,2)$ in 
			$\Coalg(\mathsf{B})$ which is defined as $\set(e_X,2)(f)=f\circ e_X$, $f\in \set(Q_X,2)$, and from this we have:
			\[
				\Lan(X)=\Imag(\set(e_X,2))=\{f\circ e_X \mid f\in \set(Q_X,2)\}.
			\]
			Closure of $\Lan(X)$ under under derivatives with respect to the type $\T$ follows from the fact that $e_X$ is a $\mathsf{T}_K$--algebra morphism. In fact, 
			for every $f\in \set(Q_X,2)$ we have:
			\begin{align*}
				(f\circ e_X)^{(i)}_{(g,t_1,\ldots ,t_{n_g-1})}(t)&=(f\circ e_X)(g(t_1,\ldots,t_{i-1},t,t_{i},\ldots,t_{n_g-1}))\\
							&=f(e_X(g(t_1,\ldots,t_{i-1},t,t_{i},\ldots,t_{n_g-1})))\\
							&=f(g(e_X(t_1),\ldots,e_X(t_{i-1}),e_X(t),e_X(t_{i}),\ldots,e_X(t_{n_g-1}))))\\
							&=\left(f^{(i)}_{(g,e_X(t_1),\ldots ,e_X(t_{n_g-1}))}\circ e_X\right)(t)
			\end{align*}
			where the function $f^{(i)}_{(g,e_X(t_1),\ldots ,e_X(t_{n_g-1}))}\in \set(Q_X,2)$ is defined for every $q\in Q_X$ as 
			$f^{(i)}_{(g,e_X(t_1),\ldots ,e_X(t_{n_g-1}))}(q)=f(g(e_X(t_1),\ldots,e_X(t_{i-1}),q,e_X(t_{i}),\ldots,e_X(t_{n_g-1})))$. Therefore, 
			$(f\circ e_X)^{(i)}_{(g,t_1,\ldots ,t_{n_g-1})}=f^{(i)}_{(g,e_X(t_1),\ldots ,e_X(t_{n_g-1}))}\circ e_X \in \Lan(X)$, i.e., $\Lan(X)$ is closed under 
			derivatives with respect to the type $\T$.

			\noindent Conversely, any $S\in \CABA$ closed under derivatives with respect to the type $\T$ such that $S$ is a subalgebra of $\set(T_KX,2)\in \CABA$ will 
			define, by duality, the canonical surjective function  $e_S:T_KX\to T_KX/\theta_S$ where $\theta_S\subseteq T_KX\times T_KX$ is defined as:
			\[
				\theta_S:=\{(v,w)\in T_KX\times T_KX \mid \exists A\in \At(S)\text{ s.t. }A(w)=A(v)=1\}
			\]
			where $\At(S)$ is the set of atoms of $S$. Clearly, $\theta_S$ is an equivalence relation on $T_KX$ since $\At(S)$ is a partition of $T_KX$. We only need to show that 
			$\theta_S$ is an $\T$--congruence on $\mathbf{T_KX}$ \footnote{A $\T$--congruence on an algebra $\mathbf{A}$ of type
			$\T$ is an equivalence relation $\theta \subseteq A\times A$ on $A$ such that for every $g\in \T$ of arity $n_g$, every $1\leq i\leq n_g$ and $a_j\in A$, 
			$1\leq j< n_g$, the property $(u,v)\in \theta$ implies $(g(a_1,\ldots a_{i-1},u,a_{i},\ldots a_{n_g-1}),g(a_1,\ldots a_{i-1},v,a_{i},\ldots a_{n_g-1}))\in \theta$. 
			(cf. \cite[Definition II.5.1]{bys})}. In fact, let $g\in \T$ of arity $n_g$, let $1\leq i\leq n_g$, let $t_j\in T_KX$, $1\leq j< n_g$, and assume $(u,v)\in \theta$, 
			i.e., there exists $A\in \At(S)$ such that $A(u)=A(v)=1$, we have to show that $(g(t_1,\ldots t_{i-1},u,t_{i},\ldots t_{n_g-1}),g(t_1,\ldots t_{i-1},v,t_{i},\ldots t_{n_g-1}))\in \theta$. 
			If $(g(t_1,\ldots t_{i-1},u,t_{i},\ldots t_{n_g-1}),g(t_1,\ldots t_{i-1},v,t_{i},\ldots t_{n_g-1}))\notin \theta$ then there exists $B\in \At(S)$ such that 
			$B(g(t_1,\ldots t_{i-1},u,t_{i},\ldots t_{n_g-1}))\neq B(g(t_1,\ldots t_{i-1},v,t_{i},\ldots t_{n_g-1}))$ which means that 
			$B^{(i)}_{(g,t_1,\ldots ,t_{n_g-1})}(u)\neq B^{(i)}_{(g,t_1,\ldots ,t_{n_g-1})}(v)$ with $B^{(i)}_{(g,t_1,\ldots ,t_{n_g-1})}\in S$ by closure under derivatives with 
			respect to the type $\T$. Therefore $A\cap B^{(i)}_{(g,t_1,\ldots ,t_{n_g-1})}$ is an element in $S$ such that $0<A\cap B^{(i)}_{(g,t_1,\ldots ,t_{n_g-1})}<A$ 
			which contradicts the fact that $A$ is an atom. 
			This proves that $(g(t_1,\ldots t_{i-1},u,t_{i},\ldots t_{n_g-1}),g(t_1,\ldots t_{i-1},v,t_{i},\ldots t_{n_g-1}))\in \theta$, which means that $e_S$ is a surjective 
			$\mathsf{T}$--algebra morphism.
		\item Condition iii) above is the commutativity of the diagram in Definition \ref{defcoeqth}.
	\end{enumerate}
	Conversely, each operator $\Lan$ on $\set$ with the properties i), ii) and iii) above defines the coequational $\mathsf{B}$--theory $\{\Lan(X)\xhookrightarrow{i_X} 2^{T_KX}\}_{X\in \set}$ 
	where $i_X$ is the inclusion $\mathsf{B}$--coalgebra morphism. Note that conditions i) and ii) above are exactly the properties that $\Lan(X)$ is a $\mathsf{B}$--subcoalgebra 
	of $\set(T_KX,2)$.\qed
\end{example}

\begin{example}\label{exvv1}
	From the previous general example we can provide details for the properties i), ii) and iii) in Example \ref{exjan11}. In fact, for the case of monoids we have the type 
	$\T=\{e,\cdot\}$ where $e$ is a nullary function symbol and $\cdot$ is a binary function symbol. We write $x\cdot y$ for $\cdot(x,y)$. By considering 
	the variety $K$ of monoids, we get the monad $\mathsf{T}_K$ such that $T_KX=X^*$, where $X^*$ is the free monoid on $X$. Then, we have:
	\begin{enumerate}[1)]
		\item Properties i) and iii) in Example \ref{exua} trivially become properties i) and iii) in Example \ref{exjan11}.
		\item Property ii) in Example \ref{exua} does not give us any kind of derivatives for the nullary function symbol $e\in \T$, but will give us the derivatives 
			$L^{(1)}_{(\cdot,u)}$ and $L^{(2)}_{(\cdot,u)}$ for the binary function symbol $\cdot\in \T$, $u\in T_KX=X^*$, which are defined for every 
			$w\in X^*$ as
			\[
				L^{(1)}_{(\cdot,u)}(w)=L(w\cdot u)=L(wu) \text{\ \ and\ \ } L^{(2)}_{(\cdot,u)}(w)=L(u\cdot w)=L(uw)
			\]
			which are respectively the left and right derivatives of $L$ with respect to $u$.
	\end{enumerate}

	\noindent In a similar way, from Example \ref{exua}, we get the following Eilenberg--type correspondences

	\begin{enumerate}[(1)]
		\item A one--to--one correspondence between varieties of semigroups and operators $\Lan$ on $\set$ such that for every $X\in \set$:
			\begin{enumerate}[i)]
				\item $\Lan(X)\in \CABA$ and it is a subalgebra of the complete atomic Boolean algebra $\set(X^+,2)$ of subsets of $X^+$, i.e., every element in $\Lan (X)$ is a 
					language on $X$ not containing the empty word.
				\item $\Lan(X)$ is closed under left and right derivatives. That is, if $L\in \Lan(X)$ and $x\in X$ then $_xL,L_x\in \Lan(X)$, where $_xL(w)=L(wx)$ and 
					$L_x(w)=L(xw)$, $w\in X^+$.
				\item $\Lan$ is closed under morphic preimages. That is, for every $Y\in \set$, homomorphism of semigroups $h:Y^+\to X^+$ and $L\in \Lan(X)$, we have 
					that $L\circ h\in \Lan(Y)$.
			\end{enumerate}
		\item A one--to--one correspondence between varieties of groups and operators $\Lan$ on $\set$ such that for every $X\in \set$:
			\begin{enumerate}[i)]
				\item $\Lan(X)\in \CABA$ and it is a subalgebra of the complete Boolean algebra $\set(\mathfrak{F}_G(X),2)$ of subsets of the free group $\mathfrak{F}_G(X)$ 
					on $X$.
				\item $\Lan(X)$ is closed under left and right derivatives and inverses. That is, if $L\in \Lan(X)$ and $x\in X$ then $_xL,L_x,L^{-1}\in \Lan(X)$, where 
					$_xL(w)=L(wx)$, $L_x(w)=L(xw)$ and $L^{-1}(w)=L(w^{-1})$, $w\in \mathfrak{F}_G(X)$.
				\item $\Lan$ is closed under morphic preimages. That is, for every $Y\in \set$, homomorphism of groups $h:\mathfrak{F}_G(Y)\to \mathfrak{F}_G(X)$ and 
					$L\in \Lan(X)$, we have that $L\circ h\in \Lan(Y)$.
			\end{enumerate}
		\item For a fixed monoid $\mathbf{M}=(M,e,\cdot)$ a one--to--one correspondence between varieties of $\mathbf{M}$--actions, i.e., dynamical systems on $\mathbf{M}$, and 
			operators $\Lan$ on $\set$ such that for every $X\in \set$:
			\begin{enumerate}[i)]
				\item $\Lan(X)\in \CABA$ and it is a subalgebra of the complete atomic Boolean algebra $\set(M\times X,2)$ of subsets of $M\times X$.	
				\item $\Lan(X)$ is closed under translations. That is, if $L\in \Lan(X)$ and $m\in M$ then $mL\in \Lan(X)$, where 
					$mL(n,x)=L(m\cdot n, x)$, $(n,x)\in M\times X$.
				\item $\Lan$ is closed under morphic preimages. That is, for every $Y\in \set$, homomorphism of $\mathbf{M}$--actions 
					$h:M\times Y\to M\times X$ (i.e., $h(m\cdot(n,y))=m \cdot h(n,y)$) and $L\in \Lan(X)$, we have that $L\circ h\in \Lan(Y)$.
			\end{enumerate}
		\item Consider the type of algebras $\T=\{\cdot, (\underline{\ \ })^\omega\}$ where $\cdot$ is a binary operation and $(\underline{\ \ })^\omega$ is a unary operation. 
			Now, let $\mathsf{T}$ be the free monad on $\set$ for the algebras of type $\T$ that satisfy the following equations:
			\begin{center}
				\begin{tabular}{m{4cm}m{4cm}m{4cm}}
				$(x\cdot y)\cdot z=x\cdot (y\cdot z)$ &  $x^\omega\cdot y=x^\omega$ & $(y\cdot x^\omega)^\omega=y\cdot x^\omega$\\
				$(x^n)^\omega=x^\omega$, $n\geq 1$ & $(x\cdot y)^\omega= x\cdot (y\cdot x)^\omega$ & \ \\
				\end{tabular}
			\end{center}
			Here $x\cdot y$ is the product of $x$ and $y$, in that order, and $x^\omega$ represents the infinite product $x\cdot x\cdot \cdots$. 
			Hence, for every $X\in \set$ the algebra $\mathbf{TX}$ has as carrier set the set $X^+\cup X^{(\omega)}$, where $X^{(\omega)}$ represents 
			the set of all ultimately periodic sequences in $X^\omega$, i.e., every element in $X^{(\omega)}$ is of the form $uv^\omega$ for some $u,v\in X^+$, and 
			$X^+\cup X^{(\omega)}$ has the natural operations $\cdot$ of concatenation and $(\underline{\ \ })^\omega$ of ``infinite power''.\\
			\noindent In this case, we get a one--to--one correspondence between varieties of semigroups with infinite exponentiation and operators $\Lan$ on $\set$ 
			such that for every $X\in \set$:
			\begin{enumerate}[i)]
				\item $\Lan(X)\in \CABA$ and it is a subalgebra of the complete atomic Boolean algebra $\set(X^+\cup X^{(\infty)},2)$ of subsets of $X^+\cup X^{(\infty)}$.	
				\item $\Lan(X)$ is closed under left and right derivatives and infinite exponentiation. That is, if $L\in \Lan(X)$ and $u\in X^+\cup X^{(\infty)}$ then 
					$_uL,L_u,L^\omega \in \Lan(X)$, where $_uL(w)=L(wu)$, $L_u(w)=L(uw)$ and $L^\omega(w)=L(w^\omega)$, $w\in X^+\cup X^{(\omega)}$.
				\item $\Lan$ is closed under morphic preimages. That is, for every $Y\in \set$, homomorphism of $\mathsf{T}$--algebras
					$h:Y^+\cup Y^{(\infty)}\to X^+\cup X^{(\infty)}$ and $L\in \Lan(X)$, we have that $L\circ h\in \Lan(Y)$.\qed
			\end{enumerate}			
	\end{enumerate}
\end{example}

\noindent We can do a similar work as in Example \ref{exua} to get Eilenberg--type correspondences for varieties of ordered algebras for any given type.

\begin{example}\label{exuoa}
	Let $\T$ be a type of algebras where each function symbol $g\in \T$ has arity $n_{g}\in \mathbb{N}$ and let $K$ be a variety of ordered algebras of type $\T$. Consider the case 
	$\mathcal{D}=\poset$, $\mathcal{D}_0=$ discrete posets, $\mathscr{E}=$ surjections, $\mathscr{M}=$ embeddings and let $\mathsf{T}_K$ be the monad such that for every 
	$\mathbf{X}=(X,\leq)\in \poset$, $T_K\mathbf{X}:=(T_KX,\leq_{T_KX})$ is the underlying poset of the free ordered algebra in $K$ on $\mathbf{X}$ generators 
	(see, e.g., \cite[Proposition 1]{bloomp}). 
	We have that $\ACDL$ is dual to $\poset$, so we can consider $\mathcal{C}=\ACDL$, $\mathcal{C}_0=\CABA$. Similar to Example \ref{exua}, by using the duality between 
	$\poset$ and $\ACDL$, each coequational $\mathsf{B}$--theory can be indexed by $\set$ (i.e., we consider every object $X\in \set$ as the object $(X,=)\in\poset$, which is in 
	$\mathcal{D}_0$) 
	and can be presented, up to isomorphism, as a family 
	$\{S_X\xhookrightarrow{m_X} 2^{T_KX}\}_{X\in \set}$, then we present coequational $\mathsf{B}$--theories as operators $\Lan$ on $\set$ given by 
	$\Lan(X):=\Imag(m_X)$. Then, we get a one--to--one correspondence between varieties of ordered algebras in $K$ and operators $\Lan$ on $\set$ such that for every $X\in \set$:
	\begin{enumerate}[i)]
		\item $\Lan(X)\in \ACDL$ and it is a subalgebra of the algebraic completely distributive lattice $\poset(T_KX,\mathbf{2}_c)\cong \set(T_KX,2)$ of subsets of $T_KX$. Here 
			$\mathbf{2}_c\in \poset$ is the two--element chain.
		\item $\Lan(X)$ is closed under derivatives with respect to the type $\T$. That is, for every $g \in \T$ of arity $n_g$, every $1\leq i\leq n_g$, 
			every $t_j\in T_KX$, $1\leq j< n_g$, and every $L\in \Lan(X)$ we have that $L^{(i)}_{(g,t_1,\ldots ,t_{n_g-1})}\in \Lan(X)$ where 
			$L^{(i)}_{(g,t_1,\ldots ,t_{n_g-1})}\in \set(T_KX,2)$ 
			is defined as 
			\[
				L^{(i)}_{(g,t_1,\ldots ,t_{n_g-1})}(t)=L(g(t_1,\ldots,t_{i-1},t,t_{i},\ldots,t_{n_g-1}))
			\]
			$t\in T_KX$. That is, for every function symbol $g\in \T$ we get $n_g$ kinds of derivatives.
		\item $\Lan$ is closed under morphic preimages. That is, for every $Y\in \set$, homomorphism of $\mathsf{T}_K$--algebras $h:T_KY\to T_KX$ and $L\in \Lan(X)$, we have 
			that $L\circ h\in \Lan(Y)$.\qed
	\end{enumerate}
\end{example}

\begin{example}[cf. \text{\cite[Theorem 5.8]{pin}}]\label{exvv3}
	From the previous example we can obtain Eilenberg--type correspondences for varieties of ordered semigroups, varieties of ordered monoids, varieties of ordered groups, and so on. 
	For instance, for the case of varieties of ordered semigroups we can consider the type $\T=\{\cdot\}$ where $\cdot$ is a binary function symbol and $K$ is the variety of ordered 
	semigroups. Then we get a one--to--one correspondence between varieties of ordered semigroups and operators $\Lan$ on $\set$ such that for every $X\in \set$:
	\begin{enumerate}[i)]
		\item $\Lan(X)\in \ACDL$ and it is a subalgebra of the algebraic completely distributive lattice $\set(X^+,2)$ of subsets of $X^+$, i.e. every element in $\Lan (X)$ 
			is a language on $X$ not containing the empty word. In particular, $\Lan(X)$ is closed under unions and intersections.
		\item $\Lan(X)$ is closed under left and right derivatives. That is, if $L\in \Lan(X)$ and $x\in X$ then $_xL,L_x\in \Lan(X)$, where $_xL(w)=L(wx)$ and 
			$L_x(w)=L(xw)$, $w\in X^+$.
		\item $\Lan$ is closed under morphic preimages. That is, for every $Y\in \set$, homomorphism of semigroups $h:Y^+\to X^+$ and $L\in \Lan(X)$, we have 
			that $L\circ h\in \Lan(Y)$.\qed
	\end{enumerate}
\end{example}

\begin{example}[cf. \text{\cite[Th\'eor\`eme III.1.1.]{reut}}]\label{exvv4}
	Let $\mathbb{K}$ be a finite field. Consider the case $\mathcal{D}=\mathcal{D}_0=\fvec$, $\mathscr{E}=$ surjections, and $\mathscr{M}=$ injections. 
	We have that $\tvec$ is dual to $\fvec$, so we can consider $\mathcal{C}=\mathcal{C}_0=\tvec$. For every set $X$ denote by $\mathsf{V}(X)$ the $\mathbb{K}$--vector space 
	with basis $X$. Consider the monad $T(\mathsf{V}(X))=\mathsf{V}(X^*)$, where $X^*$ is the free monoid on $X$. Then we get a one--to--one correspondence between 
	varieties of $\mathbb{K}$--algebras and operators $\Lan$ on $\set$ such that for every $X\in \set$:
	\begin{enumerate}[i)]
		\item $\Lan(X)\in \tvec$ and it is a subspace of the space $\fvec(\mathsf{V}(X^*),\mathbb{K})$ where the topology on $\fvec(\mathsf{V}(X^*),\mathbb{K})$ is the subspace 
			topology of the product $\mathbb{K}^{\mathsf{V}(X^*)}$.
		\item $\Lan(X)$ is closed under left and right derivatives. That is, if $L\in \Lan(X)$ and $v\in \mathsf{V}(X^*)$ then $_vL,L_v\in \Lan(X)$, where $_vL(w)=L(wv)$ and 
			$L_v(w)=L(vw)$, $w\in \mathsf{V}(X^*)$.
		\item $\Lan$ is closed under morphic preimages. That is, for every $Y\in \set$, $\mathbb{K}$--linear map $h:\mathsf{V}(Y^*)\to \mathsf{V}(X^*)$ and 
			$L\in \Lan(X)$, we have that $L\circ h\in \Lan(Y)$.\qed
	\end{enumerate}
\end{example}

\begin{example}[cf. \text{\cite[Theorem 5 (iii)]{polak}}]\label{exvv5}
	Consider the case $\mathcal{D}=\JSL$, $\mathcal{D}_0=$ free join semilattices, i.e., $\mathcal{D}_0=\{(\mathcal{P}_f(X),\cup)\mid X\in \set\}$, where $\mathcal{P}_f(X)$ is the 
	set of all finite subsets of $X$, $\mathscr{E}=$ surjections and $\mathscr{M}=$ injections. We have that $\TBJSL$ is dual to $\JSL$, so we can consider 
	$\mathcal{C}=\TBJSL$ and $\mathcal{C}_0=\{ \JSL((\mathcal{P}_f(X),\cup),2)\mid X\in \set\}$. Let $\mathsf{T}$ be the monad on $\JSL$ such that $T(S,\lor)$ is the 
	free idempotent semiring on $(S,\lor)\in \JSL$. Then we get a one--to--one correspondence between varieties of idempotent semirings and operators $\Lan$ on $\set$ such that 
	for every $X\in \set$:
	\begin{enumerate}[i)]
		\item $\Lan(X)\in \TBJSL$ and it is a subspace of $\set(X^*,2)$ where the topology given on $\set(X^*,2)$ is the subspace 
			topology of the product $2^{X^*}$. In particular, $\Lan(X)$ is closed under union. 
		\item $\Lan(X)$ is closed under left and right derivatives. That is, if $L\in \Lan(X)$ and $x\in X$ then $_xL,L_x\in \Lan(X)$, where $_xL(w)=L(wx)$ and 
			$L_x(w)=L(xw)$, $w\in X^*$.
		\item $\Lan$ is closed under morphic preimages. That is, for every $Y\in \set$, semiring homomorphism $h:\mathcal{P}_f(Y^*)\to \mathcal{P}_f(X^*)$ and 
			$L\in \Lan(X)$, we have that $L^\sharp\circ h\circ \eta_{Y^*}\in \Lan(Y)$, where $\eta_{Y^*}\in \set(Y^*,\mathcal{P}_f(Y^*))$ and 
			$L^\sharp\in \JSL(\mathcal{P}_f(X^*),2)$ are defined as $\eta_{Y^*}(w)=\{w\}$ and $L^\sharp (\{w_1,\ldots, w_n\})=\bigvee_{i=1}^n L(w_i)$. 
			Note that the composite $L^\sharp\circ h\circ \eta_{Y^*}$ is the same as $h^{(-1)}(L)$ defined in \cite{polak}. The reason of the exponent $^\sharp$ and 
			the use of $\eta_{Y^*}$ is that we are using the isomorphism:
			\begin{align*}
				\JSL(\mathcal{P}_f(X^*),&2)\cong \set(X^*,2) \\
				& \ f \mapsto f\circ \eta_{X^*}\\
				& \ L^\sharp \mapsfrom L
			\end{align*}
	\end{enumerate}
	See the Appendix for more details.\qed
\end{example}

\begin{remark}
	Note that Eilenberg--type correspondences for varieties of $\mathbb{K}$--algebras and idempotent semirings can also be obtained from Example \ref{exua}.
\end{remark}

\section{Eilenberg--type correspondences for pseudovarieties of $\mathsf{T}$--algebras}\label{secreit}

This section is similar to the previous one with the restriction that all the algebras considered are finite. We state a categorical version of Birkhoff's theorem for finite 
$\mathsf{T}$--algebras and an abstract Eilenberg--type correspondence for pseudovarieties of $\mathsf{T}$--algebras. We use the prefix `pseudo' to indicate that 
all the algebras considered are finite. That is, a pseudovariety of $\mathsf{T}$--algebras is a variety of finite $\mathsf{T}$--algebras, which is a class of finite $\mathsf{T}$--algebras 
closed under homomorphic images, subalgebras and finite producs. The Birkhoff variety theorem for finite algebras has been previously proved to prove that a class 
of finite algebras of the same type is a pseudovariety, i.e., it is closed under subalgebras, homomorphic images and finite products, if and only if it is defined by `extended equations' 
\cite{reit,banab}. An `extended equation' is a concept that generalizes the concept of an equation and can be defined by using topological techniques or, alternatively, by implicit 
operations \cite{reit,banab}. Reiterman's proof for the Birkhoff theorem for finite algebras involves topological methods in which the set of $n$--ary implicit operations is the 
completion of the set of $n$--ary terms \cite{reit}. A topological approach was also explored by Banaschewski by using uniformities \cite{banab}. Recently, in \cite{chen}, profinite 
techniques were used to define the concept of profinite equations which are the kind of equations that define pseudovarieties of $\mathsf{T}$--algebras.

\noindent We provide a categorical version of the Birkhoff theorem for finite algebras, Theorem \ref{reitthm}, which, under mild assumptions, establishes a one--to--one 
correspondence between pseudovarieties of $\mathsf{T}$--algebras and pseudoequational $\mathsf{T}$--theories. Different versions of this theorem such as \cite{reit,banab,chen} 
use topological approaches and/or profinite techniques. In the present paper, topological approaches and profinite techniques are not used, thus avoiding constructions of certain limits 
and profinite completions, which gives us a better and basic understanding on how pseudovarieties are characterized. The main strategy we follow to state and 
prove our theorem is that pseudovarieties of algebras are exactly directed unions of equational classes of finite algebras, which is a fact that was proved 
in \cite{baldwin,banab,eilsch}. The definition of pseudoequational $\mathsf{T}$--theories is based on the previous observation and the categorical dual of ``varieties of languages'' that was 
used by the author to derive an Eilenberg--type correspondence for $\mathsf{T}$--algebras \cite{jste}.

\noindent As in the previous section, the main purpose of this approach is to derive Eilenberg--type correspondences for pseudovarieties of $\mathsf{T}$--algebras. This is summarized in the 
following picture:

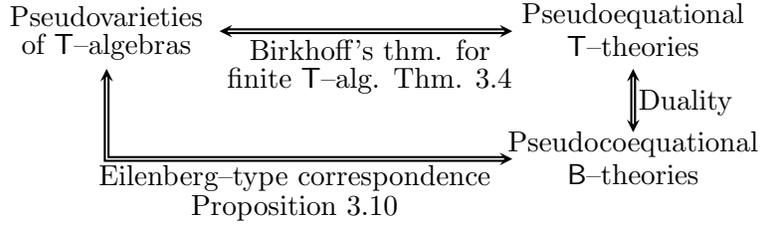
\begin{center}	
	\begin{tikzpicture}[>=stealth,shorten >=3pt,
		node distance=2.5cm,on grid,auto,initial text=,
		accepting/.style={thick,double distance=2.5pt}]
		\node () at (0,-0.05) {Pseudovarieties};
		\node () at (0,-0.45) {of $\mathsf{T}$--algebras};
		\draw[double,thick,<->] (1.5,-0.25) -- (5.5,-0.25);
		\node () at (3.5,-0.5) {Birkhoff's thm. for};
		\node () at (3.5,-0.9) {finite $\mathsf{T}$--alg. Thm. \ref{reitthm}};

		\node () at (7,-0.05) {Pseudoequational};
		\node () at (7,-0.45) {$\mathsf{T}$--theories};

		\draw[double,thick,<->] (7,-0.75) -- (7,-1.7);
		\node () at (7.7,-1.2) {Duality};
		\node () at (7,-1.75) {Pseudocoequational};
		\node () at (7,-2.15) {$\mathsf{B}$--theories};
		\draw[double,thick,<->] (0,-0.75) -- (0,-1.95)--(5.5,-1.95);
		\node () at (2.5,-2.2) {Eilenberg--type correspondence};
		\node () at (2.5,-2.6) {Proposition \ref{eilpvar}};
	\end{tikzpicture}
\end{center}

\subsection{The Birkhoff theorem for finite $\mathsf{T}$--algebras}

\noindent Throughout this section, we fix a complete concrete category $\mathcal{D}$ such that its forgetful functor preserves epis, monos and products, a monad 
$\mathsf{T}=(T,\eta,\mu)$ on 
$\mathcal{D}$, a full subcategory $\mathcal{D}_0$ of $\mathcal{D}$ and a factorization system $\mathscr{E}/\mathscr{M}$ on $\mathcal{D}$. We make the following assumptions:
\begin{enumerate}
	\item[(B$_f$1)] The factorization system $\mathscr{E}/\mathscr{M}$ is proper.
	\item[(B$_f$2)] For every $X\in \mathcal{D}_0$, the free $\mathsf{T}$--algebra 
				$\mathbf{TX}=(TX,\mu_X)$ is {\it projective with respect to $\mathscr{E}$ in $\Alg(\mathsf{T})$}. That is, for every $h\in \Alg(\mathsf{T})(\mathbf{TX},\mathbf{B})$ 
				with $X\in \mathcal{D}_0$ and $e\in \Alg (\mathsf{T})(\mathbf{A},\mathbf{B})\cap \mathscr{E}$ there exists $g\in \Alg(\mathsf{T})(\mathbf{TX},\mathbf{A})$ 
				such that $e\circ g=h$.
	\item [(B$_f$3)] For every finite $\mathbf{A}\in \Alg(\mathsf{T})$ there exists $X_A\in\mathcal{D}_0$ and 
			$s_A\in \Alg(\mathsf{T})(\mathbf{TX_A},\mathbf{A})\cap \mathscr{E}$.
	\item[(B$_f$4)] $T$ preserves morphisms in $\mathscr{E}$.
\end{enumerate}

\noindent In order to talk about finite algebras, we assume that the category $\mathcal{D}$ is a concrete category. That is, if $U:\mathcal{D}\to \set$ is the forgetful functor 
for the concrete category $\mathcal{D}$, then an object $X\in \mathcal{D}$ is {\it finite} if $U(X)$ is a finite set. Similarly, an algebra $\mathbf{A}\in \Alg(\mathsf{T})$ is {\it finite} 
if its carrier object $A\in \mathcal{D}$ is finite. The algebras of interest will be the objects $\Alg_f(\mathsf{T})$ of finite algebras in $\Alg(\mathsf{T})$. The factorization system 
$\mathscr{E}/\mathscr{M}$ on $\mathcal{D}$, which is lifted to $\Alg(\mathsf{T})$ by using (B$_f$1) and (B$_f$4), allows us to define the concept of homomorphic image and subalgebra. 
In this case, the requirement of the forgetful functor $U$ preserving epis, monos and products, will give us the property that subalgebras, homomorphic images and finite products 
of finite algebras are also finite. The purpose of the subcategory $\mathcal{D}_0$ is that the objects from which ``variables'' for the 
equations are considered are objects in $\mathcal{D}_0$. Assumption (B$_f$3) guarantees that every algebra is the homomorphic image of a free one with object of generators in 
$\mathcal{D}_0$.

\noindent To obtain Birkhoff's theorem for finite algebras we can consider $\mathcal{D}=\set$, $\mathcal{D}_0=$ finite sets, $\mathscr{E}=$ surjections, $\mathscr{M}=$ injections, 
and $\mathsf{T}$ to be the term monad for a given type of algebras $\T$, i.e., $TX=T_{\T}(X)$, the set of terms of type $\T$ on the set of variables $X$ (see Example \ref{excleqth}). 
Another important example will be given by $\mathcal{D}=\poset$ and $\mathcal{D}_0$ to be the full subcategory of finite discrete posets (as before, we do not want the 
``variables'' to be ordered). 

\noindent Now, we will define the main concepts needed to state our categorical Birkhoff's theorem for finite $\mathsf{T}$--algebras.

\begin{defi}\label{defpet}
	Let $\mathcal{D}$ be a complete concrete category such that its forgetful functor preserves epis, $\mathsf{T}$ a monad on $\mathcal{D}$, $\mathcal{D}_0$ a full 
	subcategory of $\mathcal{D}_0$ and $\mathscr{E}/\mathscr{M}$ a factorization system on $\mathcal{D}$. Assume (B$_f$1) and (B$_f$4). A {\it pseudoequational 
	$\mathsf{T}$--theory on $\mathcal{D}_0$} is an operator $\Ps$ on $\mathcal{D}_0$ such that for every $X\in \mathcal{D}_0$, $\Ps(X)$ is a nonempty collection of 
	$\mathsf{T}$--algebra morphisms in $\mathscr{E}$ with domain $\mathbf{TX}$ and finite codomain and:
	\begin{enumerate}[i)]
		\item For every finite set $I$ and $f_i\in \Ps(X)$, $i\in I$, there exists $f\in \Ps(X)$ such that every $f_i$ factors through $f$, $i\in I$.
		\item For every $e\in \Ps(X)$ with codomain $\mathbf{A}$ and every $\mathsf{T}$--algebra morphism $e'\in \mathscr{E}$ with domain $\mathbf{A}$ 
			we have that $e'\circ e\in \Ps(X)$.
		\item For every $Y\in \mathcal{D}_0$, $f\in \Ps(X)$ and $h\in \Alg(\mathsf{T})(\mathbf{TY},\mathbf{TX})$ we have that $e_{f\circ h}\in \Ps(Y)$ where 
			$f\circ h=m_{f\circ h}\circ e_{f\circ h}$ is the factorization of $f\circ h$.
	\end{enumerate}
\end{defi}

\noindent Pseudovarieties of algebras are exactly directed unions of equational classes of finite algebras \cite{baldwin,banab,eilsch}. With this in mind, we can give an 
intuition of the previous definition. In fact, for each object $X\in \mathcal{D}_0$ of variables every morphism $f\in \Ps(X)$ represents a set of equations on $X$, namely $\ker(f)$, which 
can be equivalently given by a $\mathsf{T}$--algebra morphism in $\mathscr{E}$ with domain $\mathbf{TX}$. Condition i) says that the set of all the equations on a fixed $X$ 
is a directed set, i.e., for every set of equations $f_i\in \Ps(X)$, $i\in I$, with $I$ finite, there is an upper bound $f\in \Ps(X)$. Here $f$ is an upper bound of $\{f_i\mid i\in I\}$ 
if every $f_i$ factors through $f$. Condition iii) says that all the equations considered are preserved under any substitution $h\in \Alg(\mathsf{T})(\mathbf{TY},\mathbf{TX})$ 
of variables in $Y$ by terms in $TX$, this condition is related to the commutativity of the diagram given in Definition \ref{defeqth}. Condition ii) is needed for uniqueness of 
the pseudoequational theory defining a given pseudovariety of algebras. In fact, two directed unions of equational classes of finite algebras can give us the same pseudovariety, but if 
we put the requirement of being downward closed, which is the requirement in condition ii), then we get uniqueness.

\noindent Given an algebra $\mathbf{A}\in \Alg_f(\mathsf{T})$, we say that $\mathbf{A}$ {\it satisfies} $\Ps$, denoted as $\mathbf{A}\models \Ps$, if for every 
$X\in \mathcal{D}_0$ and $f\in\Alg(\mathsf{T})(\mathbf{TX},\mathbf{A})$ we have that $f$ factors through some morphism in $\Ps(X)$. We denote by 
$\md_f (\Ps)$ the {\it finite models} of $\Ps$, that is:
\[
	\md_f (\Ps):=\{\mathbf{A}\in \Alg_f(\mathsf{T}) \mid \mathbf{A}\models \Ps\}
\]
A class $K$ of finite $\mathsf{T}$--algebras is {\it defined} by $\Ps$ if $K=\md_f(\Ps)$.

\noindent Let $K$ be a class of algebras in $\Alg_f(\mathsf{T})$. We say that $K$ {\it is closed under $\mathscr{E}$--quotients} if $\mathbf{B}\in K$ for every  
$e\in \Alg(\mathsf{T})(\mathbf{A},\mathbf{B})\cap \mathscr{E}$ with $\mathbf{A}\in K$. We say that $K$ {\it is closed under $\mathscr{M}$--subalgebras} if 
$\mathbf{B}\in K$ for every $m\in \Alg(\mathsf{T})(\mathbf{B},\mathbf{A})\cap \mathscr{M}$ with $\mathbf{A}\in K$. We say that $K$ 
{\it is closed under finite products} if $\prod_{i\in I} \mathbf{A}_i\in K$ for every finite set $I$ such that $\mathbf{A}_i\in K$, $i\in I$.

\begin{defi}
	Let $\mathcal{D}$ be a complete concrete category, $\mathsf{T}$ a monad on $\mathcal{D}$ and $\mathscr{E}/\mathscr{M}$ a factorization system on $\mathcal{D}$. 
	A class $K$ of finite algebras in $\Alg (\mathsf{T})$ is called a {\it pseudovariety of $\mathsf{T}$--algebras } if it is closed under $\mathscr{E}$--quotients, 
	$\mathscr{M}$--subalgebras and finite products. 
\end{defi}

\begin{example}
	Consider the setting $\mathcal{D}=\set$, $\mathcal{D}_0=$ finite sets, $\mathscr{E}=$ surjections, $\mathscr{M}=$ injections, and $\mathsf{T}$ to be the term monad for a given 
	type of algebras $\T$. Then we have that equational classes of finite algebras are examples of pseudovarieties of $\mathsf{T}$--algebras. For example, finite semigroups, 
	finite monoids, finite groups, finite vector spaces, finite Boolean algebras, finite lattices, and so on. In \cite{banab}, some non--equational examples of pseudovarieties are shown 
	such as:
	\begin{enumerate}[(1)]
		\item the finite commutative monoids satisfying some identity $x^n=x^{n+1}$, $n=1,2,\ldots$,
		\item the finite cancellation monoids,
		\item the finite abelian $p$--groups, for a given prime number $p$, and
		\item the finite products of finite fields of a given prime characteristic.
	\end{enumerate}
	Each of those pseudovarieties is not equational. In fact, every equation satisfied in the given pseudovariety is also satisfied in the larger pseudovariety, i.e., 
	the pseudovariety of all commutative monoids for (1), the pseudovariety of all monoids for (2), the pseudovariety of all abelian groups for (3), and the pseudovariety of all 
	commutative rings with unit of a given prime characteristic for (4).\qed
\end{example}

\noindent Now we can formulate our categorical Birkhoff's theorem for finite $\mathsf{T}$--algebras as follows.

\begin{thm}[Birkhoff's Theorem for finite $\mathsf{T}$--algebras]\label{reitthm}
	Let $\mathcal{D}$ be a complete concrete category such that its forgetful functor preserves epis, monos and products, $\mathsf{T}$ a monad on $\mathcal{D}$, $\mathcal{D}_0$ 
	a full subcategory of $\mathcal{D}_0$ and $\mathscr{E}/\mathscr{M}$ a factorization system on $\mathcal{D}$. Assume (B$_f$1) to (B$_f$4). Then a class $K$ of finite 
	$\mathsf{T}$--algebras is a pseudovariety of $\mathsf{T}$--algebras if and only if is defined by a pseudoequational $\mathsf{T}$--theory on $\mathcal{D}_0$. Additionally, 
	pseudovarieties of $\mathsf{T}$--algebras are in one--to--one correspondence with pseudoequational $\mathsf{T}$--theories on $\mathcal{D}_0$. \qed
\end{thm}

\noindent Now we derive Birkhoff's theorem for pseudovarieties of (ordered) algebras for a given type, then show an 
example of a particular pseudovariety of algebras with its defining pseudoequational $\mathsf{T}$--theory and finish this subsection by deriving Eilenberg's theorem 
\cite[Theorem 34]{eilenberg} to show a one--to--one correspondence between pseudovarieties of monoids and pseudovarieties of languages.

\begin{example}
	Consider the case $\mathcal{D}=\set$, $\mathcal{D}_0=$ finite sets, $\mathscr{E}=$ surjections, $\mathscr{M}=$ injections and, for a given type of algebras 
	$\T$, let $\mathsf{T}_{\T}$ be the term monad for $\T$. Then, by Theorem \ref{reitthm}, a class of algebras of type 
	$\T$ is a pseudovariety if and only if it is defined by a pseudoequational $\mathsf{T}_{\T}$--theory.\qed
\end{example}	

\begin{example}
	Consider the case $\mathcal{D}=\poset$, $\mathcal{D}_0=$ finite discrete posets, $\mathscr{E}=$ surjections, $\mathscr{M}=$ embeddings and, 
	for a given type of algebras $\T$, let $\mathsf{T}_{\T}$ be the monad on $\poset$ defined in Example \ref{exineqth}. Then, by Theorem \ref{reitthm}, a 
	class of ordered algebras of type $\T$ is a pseudovariety if and only if it is defined by a pseudoequational $\mathsf{T}_{\T}$--theory.\qed
\end{example}	

\begin{example}
	Consider the case $\mathcal{D}=\set$, $\mathcal{D}_0=$ finite sets, $\mathsf{T}$ the monad given by $TX=X^*$, where $X^*$ is the free 
	monoid on $X$, $\mathscr{E}=$ surjections, and $\mathscr{M}=$ injections. We have that conditions (B$_f$1) to (B$_f$4) are fullfilled. In this case, we have that 
	$\Alg(\mathsf{T})$ is the category of monoids. To describe the pseudovariety of all commutative monoids satisfying some identity $x^n=x^{n+1}$, $n=1,2,\ldots$, 
	we define $\Ps$ on $\mathcal{D}_0$ as follows:
	\begin{enumerate}[-]
		\item For every $X\in \mathcal{D}_0$, and $n=1,2\ldots$, we define the surjective homomorphism of monoids $e_n:X^*\lepi \mathfrak{F}_n(X)$, where 
			$\mathfrak{F}_n(X)$ is the free commutative monoid on $X$ that satisfies the identity $x^n=x^{n+1}$. That is, $\mathfrak{F}_n(X)=(\set(X,\mathbb{N}),\cdot, 0)$ 
			where $0\in \set(X,\mathbb{N})$ is the zero function, i.e., $0(x)=0$ for every $x\in X$, and $\cdot$ is defined on $\set(X,\mathbb{N})$ as 
			$(f\cdot g)(x)=\min\{n,f(x)+g(x)\}$. $e_n$ is defined on the set of generators $X$ as $e_n(x)=\chi_x$, where $\chi_x(x)=1$ and $\chi_x(y)=0$ for $x\neq y$. 
			Define $\Ps(X)$ as:
			\[
				\Ps(X)=\{e'\circ e_n\mid n\in \mathbb{N}^+\text{ and $e'$ is a $\mathsf{T}$--algebra morphism in $\mathscr{E}$ with domain $\mathfrak{F}_n(X)$}\}
			\]
	\end{enumerate}
	We have then that $\Ps$ is a pseudoequational $\mathsf{T}$--theory and $\md_f(\Ps)$ is the pseudovariety of all finite commutative monoids that satisfy some identity 
	$x^n=x^{n+1}$, $n=1,2,\ldots$.\qed
\end{example}

\noindent In the next example we derive Eilenberg's variety theorem \cite[Theorem 3.4.]{eilenberg}. Given a finite set $\Sigma$, i.e., an {\it alphabet}, a {\it language} $L$ on $\Sigma$ 
is a subset $L$ of $\Sigma^*$, i.e., a collection of words with letters in $\Sigma$. We identify a language $L$ on $\Sigma$ by its characteristic function $L:\Sigma^*\to 2$. A language 
$L$ on $\Sigma$ is {\it recognizable} if there exists a finite monoid $\mathbf{A}$, a homomorphism of monoids $h:\Sigma^*\to A$ and a function $L':A\to 2$ such that $L'\circ h=L$. We denote 
by $\rec(\Sigma)$ the Boolean algebra of all recognizable languages on $\Sigma$. A {\it pseudovariety of languages} is an operator $\Lan$ such that for every finite set $\Sigma$ we have:
\begin{enumerate}[i)]
	\item $\Lan(\Sigma)$ is a subalgebra of the Boolean algebra $\rec(\Sigma)$,
	\item $\Lan(\Sigma)$ is closed under left and right derivatives. That is, $_aL,L_a\in \Lan(\Sigma)$ for every $L\in \Lan(\Sigma)$ and $a\in \Sigma$, and
	\item $\Lan$ is closed under morphic preimages. That is, for every alphabet $\Gamma$, homomorphism of monoids $h:\Gamma^*\to \Sigma^*$ and $L\in \Lan(\Sigma)$, 
		we have that $L\circ h\in \Lan(\Gamma)$.
\end{enumerate}
Eilenberg's variety theorem \cite[Theorem 34]{eilenberg} says that there is a one--to--one correspondence between pseudovarieties of monoids and 
pseudovarieties of languages. This theorem is derived from Theorem \ref{reitthm} as follows.

\begin{example}[Eilenberg's variety theorem]\label{exeil}
	Consider the setting as in the previous example, i.e., $\mathcal{D}=\set$, $\mathcal{D}_0=$ finite sets, $\mathsf{T}$ the monad given by $TX=X^*$, 
	where $X^*$ is the free monoid on $X$, $\mathscr{E}=$ surjections, and $\mathscr{M}=$ injections. Then, we have a one--to--one correspondence between pseudovarieties of 
	monoids, i.e., pseudovarieties of $\mathsf{T}$--algebras, and pseudoequational $\mathsf{T}$--theories 
	on $\mathcal{D}_0$. Now, we have that pseudoequational $\mathsf{T}$--theories on $\mathcal{D}_0$ are in one--to--one correspondence with pseudovarieties of languages. 
	In fact, every pseudoequational $\mathsf{T}$--theory $\Ps$ on $\mathcal{D}_0$ defines the pseudovariety of languages $\Lan_{\Ps}$ defined as 
	$\Lan_{\Ps}(X):=\bigcup_{e\in \Ps(X)}\Imag(\set(e,2))$, and every pseudovariety of languages $\Lan$ defines the pseudoequational $\mathsf{T}$--theory $\Ps_{\Lan}$ on 
	$\mathcal{D}_0$ such that $\Ps_{\Lan}(X)$ is the collection of all $\mathsf{T}$--algebra morphisms $e\in \mathscr{E}$ with domain $\mathbf{TX}$ and finite codomain such that 
	$\Imag(\set(e,2))\subseteq \Lan(X)$, $X\in \mathcal{D}_0$. Furthermore, this correspondence is bijective, that is, for every pseudoequational $\mathsf{T}$--theory $\Ps$ on 
	$\mathcal{D}_0$ and every pseudovariety of languages $\Lan$ we have that $\Ps=\Ps_{\Lan_{\Ps}}$ and $\Lan=\Lan_{\Ps_{\Lan}}$ (see Example \ref{exuaf} for more details). \qed
\end{example}

\subsection{Abstract Eilenberg--type correspondence for pseudovarieties of $\mathsf{T}$--algebras}\label{secreiteil}

As we saw in Example \ref{exeil}, we can derive Eilenberg--type correspondences for pseudovarieties of $\mathsf{T}$--algebras from Birkhoff's theorem for finite $\mathsf{T}$--algebras,  
Theorem \ref{reitthm}. Eilenberg--type correspondences for pseudovariaties of $\mathsf{T}$--algebras are exactly one--to--one correspondences between pseudovarieties of 
$\mathsf{T}$--algebras and duals of pseudoequational $\mathsf{T}$--theories. By dualizing the definition of a pseudoequational $\mathsf{T}$--theory we get the following. 

\begin{defi}\label{defpcoet}
	Let $\mathcal{C}$ be a concrete category such that its forgetful functor preserves monos, $\mathsf{B}=(\mathsf{B},\epsilon,\delta)$ a comonad on $\mathcal{C}$, 
	$\mathscr{E}/\mathscr{M}$ a factorization system on $\mathcal{C}$ and $\mathcal{C}_0$ a full subcategory of $\mathcal{C}$. Assume (B$_f$1) and that $B$ preserves 
	the morphisms in $\mathscr{M}$. A {\it pseudocoequational $\mathsf{B}$--theory on $\mathcal{C}_0$} is an operator $\Psc$ on $\mathcal{C}_0$ such that for every 
	$X\in \mathcal{C}_0$, $\Psc(X)$ is a nonempty collection of $\mathsf{B}$--coalgebra morphisms in $\mathscr{M}$ with codomain $\mathbf{BX}$ and finite domain and:
	\begin{enumerate}[i)]
		\item For every finite set $I$ and $f_i\in \Psc(X)$, $i\in I$, there exists $f\in \Psc(X)$ such that every $f_i$ factors through $f$, $i\in I$.
		\item For every $m\in \Psc(X)$ with domain $\mathbf{A}$ and every $\mathsf{B}$--coalgebra morphism $m'\in \mathscr{M}$ with codomain $\mathbf{A}$ 
			we have that $m\circ m'\in \Psc(X)$.
		\item For every $Y\in \mathcal{C}_0$, $f\in \Psc(X)$ and $h\in \Coalg(\mathsf{B})(\mathbf{BX},\mathbf{BY})$ we have that $m_{h\circ f}\in \Psc(Y)$ where 
			$h\circ f=m_{h\circ f}\circ e_{h\circ f}$ is the factorization of $h\circ f$.
	\end{enumerate}
\end{defi}

\noindent With the previous definition, Theorem \ref{reitthm} and duality, we have the following:

\begin{prop}[Abstract Eilenberg--type correspondence for pseudovarieties of $\mathsf{T}$--al\-gebras]\label{eilpvar}
	Let $\mathcal{D}$ be a complete concrete category such that its forgetful functor preserves epis, monos and products, $\mathsf{T}$ a monad on $\mathcal{D}$, $\mathcal{D}_0$ 
	a full subcategory of $\mathcal{D}_0$ and $\mathscr{E}/\mathscr{M}$ a factorization system on $\mathcal{D}$. Assume (B$_f$1) to (B$_f$4). Let $\mathcal{C}$ be a category 
	that is dual to $\mathcal{D}$, let $\mathcal{C}_0$ be dual of $\mathcal{D}_0$ and $\mathsf{B}$ be the comonad on $\mathcal{C}$ that is dual to the monad $\mathsf{T}$ on 
	$\mathcal{D}$ which is defined as in Proposition \ref{monadtocomonad}. Then there is a one--to--one correspondence between pseudovarieties of $\mathsf{T}$--algebras 
	and pseudocoequational $\mathsf{B}$--theories on $\mathcal{C}_0$.
\end{prop}

\noindent We will consider the same settings given in the examples in subsection \ref{seceilv} to obtain the following Eilenberg--type correspondences for pseudovarieties of 
$\mathsf{T}$--algebras. In all of them, we only need to change the category $\mathcal{D}_0$ and condition i) for the operators $\Lan$ (see Appendix for their details).

\begin{example}[cf. Example \ref{exua}]\label{exuaf}
	In this example we obtain an Eilenberg--type correspondence for any pseudovariety of algebras of any given type $\T$, where each of the function symbols in $\T$ has finite arity. 
	Let $\T$ be a type of algebras where each function symbol $g\in \T$ has arity $n_{g}\in \mathbb{N}$ and let $K$ be a variety of algebras for of type $\T$. Consider the case 
	$\mathcal{D}=\set$, $\mathcal{D}_0=\set_f$, $\mathscr{E}=$ surjections, $\mathscr{M}=$ injections and let $\mathsf{T}_K$ be the monad such that for every $X\in \set$, 
	$T_KX$ is the underlying set of the free algebra in $K$ on $X$ generators (see \cite[Definition II.10.9]{bys} and \cite[VI.8]{maclane}). 
	We have that $\CABA$ is dual to $\set$, so we can consider $\mathcal{C}=\CABA$ and $\mathcal{C}_0=\CABA_f$. In this case, we get a one--to--one correspondence 
	between pseudovarieties of algebras in $K$ and operators $\Lan$ on $\set_f$ such that for every $X\in \set_f$:
	\begin{enumerate}[i)]
		\item $\Lan(X)$ is a Boolean algebra and it is a subalgebra of the complete atomic Boolean algebra $\set(T_KX,2)$ of subsets of $T_KX$ such that 
			for every $L\in \Lan(X)$ there exists a finite algebra $\mathbf{A}$ in $K$, a morphism $h\in \Alg(\mathsf{T}_K)(\mathbf{T_KX},\mathbf{A})$ and $L'\in \set(A,2)$ 
			such that $L=L'\circ h$.
		\item $\Lan(X)$ is closed under derivatives with respect to the type $\T$. That is, for every $g \in \T$ of arity $n_g$, every $1\leq i\leq n_g$, 
			every $t_j\in T_KX$, $1\leq j< n_g$, and every $L\in \Lan(X)$ we have that $L^{(i)}_{(g,t_1,\ldots ,t_{n_g-1})}\in \Lan(X)$ where 
			$L^{(i)}_{(g,t_1,\ldots ,t_{n_g-1})}\in \set(T_KX,2)$ 
			is defined as 
			\[
				L^{(i)}_{(g,t_1,\ldots ,t_{n_g-1})}(t)=L(g(t_1,\ldots,t_{i-1},t,t_{i},\ldots,t_{n_g-1}))
			\]
			$t\in T_KX$.
		\item $\Lan$ is closed under morphic preimages. That is, for every $Y\in \set_f$, homomorphism of $\mathsf{T}_K$--algebras $h:T_KY\to T_KX$ and $L\in \Lan(X)$, we have 
			that $L\circ h\in \Lan(Y)$.
	\end{enumerate}
	
	\noindent In fact, let $\Ps$ be a pseudoequational $\mathsf{T}_K$--theory on $\set_f$ and let $\Lan$ be an operator on $\set_f$ satisfying the three properties above. Then:
	\begin{enumerate}[a)]
		\item Define the operator $\Lan_{\Ps}$ on $\set_f$ as $\Lan_{\Ps}(X):=\bigcup_{e\in \Ps(X)}\Imag(\set(e,2))$. 
			We claim that $\Lan_{\Ps}$ satisfies the three properties above. In fact, as the family $\Ps(X)$ is directed in the sense of Definition \ref{defpet} i), then the union 
			$\bigcup_{e\in \Ps(X)}\Imag(\set(e,2))\subseteq \set(T_KX,2)$ is a directed union of finite objects in $\CABA$ which is a Boolean subalgebra of $\set(T_KX,2)$. 
			As each $e\in \Ps(X)$ has as codomain a finite algebra in $K$ then $\Imag(\set(e,2))$ is a subset of $\set(T_KX,2)$ which is closed under derivatives 
			with respect to the type $\T$ (see Example \ref{exua}). The previous argument shows that $\Lan_{\Ps}$ satisfies properties i) and ii) above. 
			Now, closure under morphic preimages follows from property iii) in Definition \ref{defpet}. Therefore, $\Lan_{\Ps}$ satisfies the three properties above.
		\item Define the operator $\Ps_{\Lan}$ on $\set_f$ such that $\Ps_{\Lan}(X)$ is the collection of all $\mathsf{T}_K$--algebra morphisms 
			$e\in \mathscr{E}$ with domain $\mathbf{T_KX}$ and finite codomain such that $\Imag(\set(e,2))\subseteq \Lan(X)$. We claim that $\Ps_{\Lan}$ is a pseudoequational 
			$\mathsf{T}_K$--theory. In fact, we have that $\Ps_{\Lan}(X)$ is nonempty since 
			$e:\mathbf{T_KX}\lepi \mathbf{1}\in \Ps_{\Lan}(X)$, where $\mathbf{1}$ is the one--element $\mathsf{T}_K$--algebra. 
			By definition, we have that $\Ps_{\Lan}(X)$ satisfies property ii) in Definition \ref{defpet}, and, it also satisfies property iii) in Definition \ref{defpet} since $\Lan$ is closed under 
			morphic preimages. Now, consider a family $\{T_KX\overset{e_i}{\lepi} A_i\}_{i\in I}$ in $\Ps_{\Lan}(X)$ with $I$ finite such that $\Imag(\set(e_i,2))\subseteq \Lan(X)$, we need 
			to find a morphism $e\in \Ps_{\Lan}(X)$ such that every $e_i$ factors through $e$. In fact, let $\mathbf{A}$ be the product of $\prod_{i\in I}\mathbf{A_i}$ with projections 
			$\pi_i:A\to A_i$, then, by the universal property of $\mathbf{A}$ there exists a $\mathsf{T}_K$--algebra morphism $f:T_KX\to A$ such that $\pi_i\circ f=e_i$, for every $i\in I$. 
			Let $f=m_f\circ e_f$ be the factorization of $f$ in $\Alg(\mathsf{T}_K)$. We claim that $e=e_f$ is a morphism in $\Ps_{\Lan}(X)$ such that 
			every $e_i$ factors through $e$. Clearly, from the construction above, each $e_i$ factors through $e=e_f$. Now, let's prove that $\Imag(\set(e,2))\subseteq \Lan(X)$. In fact, 
			let $\mathbf{S}$ be the codomain of $e=e_f$ and let $g\in \set(S,2)$. We have to prove that $g\circ e\in \Lan(X)$ which follows from the following straightforward identity:
			\[
				g\circ e=\bigcup_{s\in g}\left(\bigcap_{i\in I} h_{i,s}\circ e_i\right)
			\]
			where $h_{i,s}\in \set(A_i,2)$ is the set $\{\pi_i(m_f(s))\}$ (i.e., we express the subset $g$ of $S$ as the union of its points and each point $s\in S$ is represented as 
			$\bigcap_{i\in I} h_{i,s}\circ \pi_i\circ m_f$). Now, for every $s\in S$ and $i\in I$ the composition $h_{i,s}\circ e_i$ belongs to $\Lan(X)$ since 
			$h_{i,s}\circ e_i\in \Imag(\set(e_i,2))\subseteq \Lan(X)$. As $S$ and $I$ are finite then $g\circ e\in \Lan(X)$ because $\Lan(X)$ is a Boolean algebra.
		\item We have that $\Ps=\Ps_{\Lan_{\Ps}}$. In fact, for every $X\in \set_f$ the inclusion $\Ps(X)\subseteq \Ps_{\Lan_{\Ps}}(X)$ is obvious. Now, to prove that 
			$\Ps_{\Lan_{\Ps}}(X)\subseteq \Ps(X)$, let $e'\in \Alg(\mathsf{T})(\mathbf{T_KX},\mathbf{A})\cap\mathscr{E}$ with finite codomain such that $e'\in \Ps_{\Lan_{\Ps}}(X)$, i.e., 
			$\Imag(\set(e',2))\subseteq \bigcup_{e\in \Ps(X)}\Imag(\set(e,2))$. Then the previous inclusion means that for every $f\in \set(A,2)$ there exists $e_f\in \Ps(X)$ and $g_f$ such 
			that $f\circ e'=g_f\circ e_f$. As $\{e_f\mid f\in \set(A,2)\}$ is finite, then there exists $e\in \Ps(X)$ such that each $e_f$ factors through $e$. We will prove that $e'$ 
			factors through $e\in \Ps(X)$ which will imply that $e'\in \Ps(X)$, since $\Ps$ is a pseudoequational $\mathsf{T}_K$--theory. It is enough to show that 
			$\ker(e)\subseteq \ker (e')$. 
			In fact, assume that $(u,v)\in \ker(e)$ and define $f'\in\set(A,2)$ as $f'(x)=1$ iff $x=e'(u)$. Then, as $e_{f'}$ factors through $e$ we have that $\ker(e)\subseteq \ker(e_{f'})$ 
			which implies $(u,v)\in \ker(e_{f'})$. But $\ker(e_{f'})\subseteq \ker(g_{f'}\circ e_{f'})=\ker(f'\circ e')$, which implies that $(u,v)\in \ker(f'\circ e')$, i.e., 
			$1=f'(e'(u))=f'(e'(v))$, but the later equality means that $e'(u)=e'(v)$ by definition of $f'$, i.e., $(u,v)\in \ker(e')$ as desired.
		\item We have that $\Lan=\Lan_{\Ps_{\Lan}}$. In fact, for every $X\in \set_f$ the inclusion $\Lan_{\Ps_{\Lan}}(X)\subseteq \Lan(X)$ is obvious. Now, 
			to prove $\Lan(X)\subseteq \Lan_{\Ps_{\Lan}}(X)$ we need to find for every $L\in \Lan(X)$ a surjective homomorphism $e:T_KX\to A$ with $\mathbf{A}\in K$ such that 
			$L\in \Imag(\set(e,2))\subseteq \Lan(X)$. In fact, for $L\in \Lan(X)$ let $e':T_KX\to B$ be a  homomorphism with $\mathbf{B}\in K$ and $g\in \set(B,2)$ such that 
			$L=g\circ e'$, this can be done by property i) above. Let $\langle\mkern-3mu\langle L\rangle\mkern-3mu\rangle$ be the subset of $\set(T_KX,2)$ 
			obtained from $\{L\}$ which is closed under Boolean combinations and derivatives with respect to the type $\T$. We show that 
			$\langle\mkern-3mu\langle L\rangle\mkern-3mu\rangle \in \Coalg_f(\mathsf{B})$, that is, we show that $\langle\mkern-3mu\langle L\rangle\mkern-3mu\rangle$ is a 
			finite object in $\CABA$ that is closed under derivatives with respect to the type $\T$. In fact, 
			$\Imag(\set(e',2))\in \Coalg_f(\mathsf{B})$ is such that $\langle\mkern-3mu\langle L\rangle\mkern-3mu\rangle\subseteq \Imag(\set(e',2))$, which implies that 
			$\langle\mkern-3mu\langle L\rangle\mkern-3mu\rangle$ is a finite Boolean algebra, i.e., an object in $\Coalg_f(\mathsf{B})$. By construction of 
			$\langle\mkern-3mu\langle L\rangle\mkern-3mu\rangle$ we have that $L\in \langle\mkern-3mu\langle L\rangle\mkern-3mu\rangle\subseteq \Lan(X)$ 
			since $\Lan$ satisfies properties i) and ii) above. Now, let $i\in \Coalg(\mathsf{B})(\langle\mkern-3mu\langle L\rangle\mkern-3mu\rangle,\set(T_KX,2))$ 
			be the inclusion morphism, then by duality we have that the dual morphism $e$ in $\Alg(\mathsf{T}_K)$ of $i$ is such that $L\in \Imag(\set(e,2))\subseteq \Lan(X)$ 
			(in fact, $\Imag(\set(e,2))=\langle\mkern-3mu\langle L\rangle\mkern-3mu\rangle$). 
			Note that the codomain of $e$ is in $K$ since it is an $\mathscr{E}$--quotient of $\mathbf{B}\in K$. 
	\end{enumerate}
	\begin{remark}
		Note that, for every ``language'' $L\in \set(T_KX,2)$, the object $\langle\mkern-3mu\langle L\rangle\mkern-3mu\rangle$ in d) above is the $\mathsf{B}$--subcoalgebra of 
		$\set(T_KX,2)$ generated by $L$ which implies, by duality, that its dual is the syntactic algebra $\mathbf{S}_L$ of $L$. Additionally, by using duality and the construction of 
		$\langle\mkern-3mu\langle L\rangle\mkern-3mu\rangle$, we have that every ``language'' in $\Imag(\set(e,2))$ (i.e., recognized by the syntactic algebra of $L$) is a Boolean 
		combination of derivatives of $L$, where $e$ is the dual of the inclusion $i\in \Coalg(\mathsf{B})(\langle\mkern-3mu\langle L\rangle\mkern-3mu\rangle,\set(T_KX,2))$.\qed
	\end{remark}
\end{example}

\begin{example}\label{expvv1}
	From the previous example, we get the following Eilenberg--type correspondences:

	\begin{enumerate}[(1)]
		\item \cite[Theorem 34]{eilenberg} A one--to--one correspondence between pseudovarieties of monoids and operators $\Lan$ on $\set_f$ such that for every $X\in \set_f$:
			\begin{enumerate}[i)]
				\item $\Lan(X)$ is a Boolean subalgebra of $\set(X^*,2)$ such that for every $L\in \Lan(X)$ there exists a finite monoid $\mathbf{M}$, a homomorphism 
					$h\in \Alg(\mathsf{T})(\mathbf{TX},\mathbf{M})$ and $L'\in \set(M,2)$ such that $L'\circ h=L$, i.e., $L$ is a recognizable language on $X$.
				\item $\Lan(X)$ is closed under left and right derivatives. That is, if $L\in \Lan(X)$ and $x\in X$ then $_xL,L_x\in \Lan(X)$, where $_xL(w)=L(wx)$ and 
					$L_x(w)=L(xw)$, $w\in X^*$.
				\item $\Lan$ is closed under morphic preimages. That is, for every $Y\in \set$, homomorphism of monoids $h:Y^*\to X^*$ and $L\in \Lan(X)$, we have 
					that $L\circ h\in \Lan(Y)$.
			\end{enumerate}
		\item \cite[Theorem 34s]{eilenberg} A one--to--one correspondence between pseudovarieties of semigroups and operators $\Lan$ on $\set_f$ such that for every $X\in \set_f$:
			\begin{enumerate}[i)]
				\item $\Lan(X)$ is a Boolean subalgebra of $\set(X^+,2)$ such that for every $L\in \Lan(X)$ there exists a finite semigroup $\mathbf{S}$, a homomorphism 
					$h\in \Alg(\mathsf{T})(\mathbf{TX},\mathbf{S})$ and $L'\in \set(S,2)$ such that $L'\circ h=L$, i.e., $L$ is a recognizable language on $X$ not containing 
					the empty word.
				\item $\Lan(X)$ is closed under left and right derivatives. That is, if $L\in \Lan(X)$ and $x\in X$ then $_xL,L_x\in \Lan(X)$, where $_xL(w)=L(wx)$ and 
					$L_x(w)=L(xw)$, $w\in X^+$.
				\item $\Lan$ is closed under morphic preimages. That is, for every $Y\in \set$, homomorphism of semigroups $h:Y^+\to X^+$ and $L\in \Lan(X)$, we have 
					that $L\circ h\in \Lan(Y)$.
			\end{enumerate}
		\item A one--to--one correspondence between pseudovarieties of groups and operators $\Lan$ on $\set_f$ such that for every $X\in \set_f$:
			\begin{enumerate}[i)]
				\item $\Lan(X)$ is a Boolean subalgebra of $\set(\mathfrak{F}_G(X),2)$ such that for every $L\in \Lan(X)$ there exists a finite group $\mathbf{G}$, a 
					homomorphism $h\in \Alg(\mathsf{T})(\mathbf{TX},\mathbf{G})$ and $L'\in \set(G,2)$ such that $L'\circ h=L$.
				\item $\Lan(X)$ is closed under left and right derivatives and inverses. That is, if $L\in \Lan(X)$ and $x\in X$ then $_xL,L_x,L^{-1}\in \Lan(X)$, where 
					$_xL(w)=L(wx)$, $L_x(w)=L(xw)$ and $L^{-1}(w)=L(w^{-1})$, $w\in \mathfrak{F}_G(X)$.
				\item $\Lan$ is closed under morphic preimages. That is, for every $Y\in \set$, homomorphism of groups $h:\mathfrak{F}_G(Y)\to \mathfrak{F}_G(X)$ and 
					$L\in \Lan(X)$, we have that $L\circ h\in \Lan(Y)$.
			\end{enumerate}
		\item For a fixed monoid $\mathbf{M}=(M,e,\cdot)$, a one--to--one correspondence between pseudovarieties of $\mathbf{M}$--actions, i.e., dynamical systems on 
			$\mathbf{M}$, and operators $\Lan$ on $\set_f$ such that for every $X\in \set_f$:
			\begin{enumerate}[i)]
				\item $\Lan(X)$ is a Boolean subalgebra of $\set(M\times X,2)$ such that for every $L\in \Lan(X)$ there exists a finite $\mathbf{M}$--action $\mathbf{S}$, 
					a homomorphism $h\in \Alg(\mathsf{T})(\mathbf{TX},\mathbf{S})$ and $L'\in \set(S,2)$ such that $L'\circ h=L$.
				\item $\Lan(X)$ is closed under translations. That is, if $L\in \Lan(X)$ and $m\in M$ then $mL\in \Lan(X)$, where 
					$mL(n,x)=L(m\cdot n, x)$, $(n,x)\in M\times X$.
				\item $\Lan$ is closed under morphic preimages. That is, for every $Y\in \set$, homomorphism of $\mathbf{M}$--actions 
					$h:M\times Y\to M\times X$ (i.e., $h(m\cdot(n,y))=m\cdot h(n,y)$) and $L\in \Lan(X)$, we have that $L\circ h\in \Lan(Y)$.
			\end{enumerate}
		\item (cf. \cite{wilke}) A one--to--one correspondence between pseudovarieties of semigroups with infinite exponentiation and 
			operators $\Lan$ on $\set_f$ such that for every $X\in \set_f$:
			\begin{enumerate}[i)]
				\item $\Lan(X)$ is a Boolean subalgebra of $\set(X^+\cup X^{(\infty)},2)$ such that for every $L\in \Lan(X)$ there exists a finite semigroup with infinite exponentiation 
					$\mathbf{S}$, a homomorphism $h\in \Alg(\mathsf{T})(\mathbf{TX},\mathbf{S})$ and $L'\in \set(S,2)$ such that $L'\circ h=L$.
				\item $\Lan(X)$ is closed under left and right derivatives and infinite exponentiation. That is, if $L\in \Lan(X)$ and $u\in X^+\cup X^{(\infty)}$ then 
					$_uL,L_u,L^\omega \in \Lan(X)$, where $_uL(w)=L(wu)$, $L_u(w)=L(uw)$ and $L^\omega(w)=L(w^\omega)$, $w\in X^+\cup X^{(\omega)}$.
				\item $\Lan$ is closed under morphic preimages. That is, for every $Y\in \set$, homomorphism of $\mathsf{T}$--algebras
					$h:Y^+\cup Y^{(\infty)}\to X^+\cup X^{(\infty)}$ and $L\in \Lan(X)$, we have that $L\circ h\in \Lan(Y)$.\qed
			\end{enumerate}
	\end{enumerate}
\end{example}

\noindent In the next example we obtain an Eilenberg--type correspondence for any variety of ordered algebras for a given type $\T$ such that each function symbol in 
$\T$ has a finite arity.

\begin{example}[cf. Example \ref{exuoa}]\label{exuoaf}
	Let $\T$ be a type of algebras where each function symbol $g\in \T$ has arity $n_{g}\in \mathbb{N}$ and let $K$ be a variety of ordered algebras of type $\T$. Consider the case 
	$\mathcal{D}=\poset$, $\mathcal{D}_0=$ finite discrete posets, $\mathscr{E}=$ surjections, $\mathscr{M}=$ embeddings and let $\mathsf{T}_K$ be the monad such that for every 
	$\mathbf{X}=(X,\leq)\in \poset$, $T_K\mathbf{X}:=(T_KX,\leq_{T_KX})$ is the underlying poset of the free ordered algebra in $K$ on $\mathbf{X}$ generators 
	(see \cite[Proposition 1]{bloomp}). We have that $\ACDL$ is dual to $\poset$, so we can consider $\mathcal{C}=\ACDL$, $\mathcal{C}_0=\CABA_f$. In this case, 
	we get a one--to--one correspondence between pseudovarieties of ordered algebras in $K$ and operators $\Lan$ on $\set_f$ such that for every $X\in \set_f$:
	\begin{enumerate}[i)]
		\item $\Lan(X)$ is a distributive sublattice of $\poset(T_KX,\mathbf{2}_c)\cong \set(T_KX,2)$ of subsets of $T_KX$ such that 
			for every $L\in \Lan(X)$ there exists a finite ordered algebra $\mathbf{A}$ in $K$, a morphism $h\in \Alg(\mathsf{T}_K)(\mathbf{T_KX},\mathbf{A})$ and 
			$L'\in \poset(A,\mathbf{2}_c)$ such that $L=L'\circ h$. Here $\mathbf{2}_c\in \poset$ is the two--element chain.
		\item $\Lan(X)$ is closed under derivatives with respect to the type $\T$. That is, for every $g \in \T$ of arity $n_g$, every $1\leq i\leq n_g$, 
			every $t_j\in T_KX$, $1\leq j< n_g$, and every $L\in \Lan(X)$ we have that $L^{(i)}_{(g,t_1,\ldots ,t_{n_g-1})}\in \Lan(X)$ where 
			$L^{(i)}_{(g,t_1,\ldots ,t_{n_g-1})}\in \set(T_KX,2)$ 
			is defined as 
			\[
				L^{(i)}_{(g,t_1,\ldots ,t_{n_g-1})}(t)=L(g(t_1,\ldots,t_{i-1},t,t_{i},\ldots,t_{n_g-1}))
			\]
			$t\in T_KX$. That is, for every function symbol $g\in \T$ we get $n_g$ kinds of derivatives.
		\item $\Lan$ is closed under morphic preimages. That is, for every $Y\in \set$, homomorphism of $\mathsf{T}_K$--algebras $h:T_KY\to T_KX$ and $L\in \Lan(X)$, we have 
			that $L\circ h\in \Lan(Y)$.\qed
	\end{enumerate}	
\end{example}

\begin{example}[\text{\cite[Theorem 5.8]{pin}} cf. Example \ref{exvv3}]\label{expvv3}
	From the previous example we can obtain Eilenberg--type correspondences for pseudovarieties of ordered semigroups, pseudovarieties of ordered monoids, pseudovarieties of 
	ordered groups, and so on. For instance, for the case of pseudovarieties of ordered monoids we can consider the type $\T=\{e,\cdot\}$ where $e$ is a nullary function symbol, 
	$\cdot$ is a binary function symbol and $K$ is the variety of ordered monoids. Then we get a one--to--one correspondence between pseudovarieties of ordered monoids and 
	operators $\Lan$ on $\set_f$ such that for every $X\in \set_f$:
	\begin{enumerate}[i)]
		\item $\Lan(X)$ is a distributive sublattice of the distributive lattice $\set(X^*,2)$ of subsets of $X^*$, i.e., every element in $\Lan (X)$ 
			is a language on $X$, such that every $L\in \Lan(X)$ is a regular language.
		\item $\Lan(X)$ is closed under left and right derivatives. That is, if $L\in \Lan(X)$ and $x\in X$ then $_xL,L_x\in \Lan(X)$, where $_xL(w)=L(wx)$ and 
			$L_x(w)=L(xw)$, $w\in X^*$.
		\item $\Lan$ is closed under morphic preimages. That is, for every $Y\in \set$, homomorphism of monoids $h:Y^*\to X^*$ and $L\in \Lan(X)$, we have 
			that $L\circ h\in \Lan(Y)$.\qed
	\end{enumerate}
\end{example}

\begin{example}[cf. \text{\cite[Th\'eor\`eme III.1.1.]{reut}} and Example \ref{exvv4}]\label{expvv4}
	Let $\mathbb{K}$ be a finite field. Consider the case $\mathcal{D}=\fvec$, $\mathcal{D}_0=$ finite $\mathbb{K}$--vector spaces, $\mathscr{E}=$ surjections and 	
	$\mathscr{M}=$ injections. We have that $\tvec$ is dual to $\fvec$, so we can consider $\mathcal{C}=\tvec$ and $\mathcal{C}_0=$ finite $\mathbb{K}$--vector spaces. 
	For every set $X$ denote by $\mathsf{V}(X)$ the $\mathbb{K}$--vector space with basis $X$. Consider the monad $T(\mathsf{V}(X))=\mathsf{V}(X^*)$, where $X^*$ is the 
	free monoid on $X$. Then we get a one--to--one correspondence between pseudovarieties of $\mathbb{K}$--algebras and operators $\Lan$ on $\set_f$ such that for every 
	$X\in \set_f$:
	\begin{enumerate}[i)]
		\item $\Lan(X)$ is a $\mathbb{K}$--vector space which is a subspace of $\fvec(\mathsf{V}(X^*),\mathbb{K})$ such that every element $S$ in $\Lan(X)$ 
			is a {\it recognizable series} on $X$, i.e., there exists a $\mathbb{K}$--algebra morphism $h: \mathbf{TX}\to \mathbf{A}$, with $\mathbf{A}$ finite, and 
			$S'\in \fvec(A,\mathbb{K})$ such that $S'\circ h=S$.
		\item $\Lan(X)$ is closed under left and right derivatives. That is, if $L\in \Lan(X)$ and $v\in \mathsf{V}(X^*)$ then $_vL,L_v\in \Lan(X)$, where $_vL(w)=L(wv)$ and 
			$L_v(w)=L(vw)$, $w\in \mathsf{V}(X^*)$.
		\item $\Lan$ is closed under morphic preimages. That is, for every $Y\in \set$, $\mathbb{K}$--linear map $h:\mathsf{V}(Y^*)\to \mathsf{V}(X^*)$ and 
			$L\in \Lan(X)$, we have that $L\circ h\in \Lan(Y)$.\qed
	\end{enumerate}
\end{example}

\begin{example}[\text{\cite[Theorem 5 (iii)]{polak}} cf. Example \ref{exvv5}]\label{expvv5}
	Consider the case $\mathcal{D}=\JSL$, $\mathcal{D}_0=$ finite free join semilattices, i.e., $\mathcal{D}_0=\{(\mathcal{P}(X),\cup)\mid X\in \set_f\}$, where $\mathcal{P}$ is the 
	powerset operator, $\mathscr{E}=$ surjections and $\mathscr{M}=$ injections. We have that $\TBJSL$ is dual to $\JSL$, so we can consider 
	$\mathcal{C}=\TBJSL$ and $\mathcal{C}_0=\{ \JSL((\mathcal{P}(X),\cup),2)\mid X\in \set_f\}$. Let $\mathsf{T}$ be the monad on $\JSL$ such that $T(S,\lor)$ is the 
	free idempotent semiring on $(S,\lor)\in \JSL$. Then we get a one--to--one correspondence between pseudovarieties of idempotent semirings and operators $\Lan$ on $\set_f$ 
	such that for every $X\in \set_f$:
	\begin{enumerate}[i)]
		\item $\Lan(X)$ is a join subsemilattice of $\set(X^*,2)$ such that every $L\in \Lan(X)$ is a regular language. In particular, $\Lan(X)$ is closed under union. 
		\item $\Lan(X)$ is closed under left and right derivatives. That is, if $L\in \Lan(X)$ and $x\in X$ then $_xL,L_x\in \Lan(X)$, where $_xL(w)=L(wx)$ and 
			$L_x(w)=L(xw)$, $w\in X^*$.
		\item $\Lan$ is closed under morphic preimages. That is, for every $Y\in \set$, semiring homomorphism $h:\mathcal{P}_f(Y^*)\to \mathcal{P}_f(X^*)$ and 
			$L\in \Lan(X)$, we have that $L^\sharp\circ h\circ \eta_{Y^*}\in \Lan(Y)$ (see Example \ref{exvv5}).\qed
	\end{enumerate}
\end{example}
\begin{remark}
	Note that Eilenberg--type correspondences for pseudovarieties of $\mathbb{K}$--algebras and idempotent semirings can also be obtained from Example \ref{exuaf}.
\end{remark}

\section{Local Eilenberg--type correspondences}\label{seclocal}

In this section, we provide abstract versions of local Eilenberg--type correspondences for local (pseudo)varieties of $\mathsf{T}$--algebras. Local Eilenberg--type correspondences 
have been studied in \cite{adamekl,gehrke}. The main idea of local Eilenberg--type correspondences is to work with a fixed alphabet, which in our notation reduces to consider 
the case in which the category $\mathcal{D}_0$ has only one object, say $X$. In order to do this, the kind of algebras considered in this local version are algebras that are generated 
by the object $X$ in the following sense. 

\begin{defi}
	Let $\mathcal{D}$ be a category, $\mathsf{T}$ a monad on $\mathcal{D}$, $\mathscr{E}/\mathscr{M}$ a factorization system on $\mathcal{D}$ and $X\in \mathcal{D}$. 
	An algebra $\mathbf{A}\in \Alg(\mathsf{T})$ is $X$--generated if $\Alg(\mathsf{T})(\mathbf{TX},\mathbf{A})\cap \mathscr{E}$ is nonempty.
\end{defi}

\noindent We have that $\mathscr{E}$--quotients of $X$--generated $\mathsf{T}$--algebras are $X$--generated, but this property does not hold in general for 
$\mathscr{M}$--subalgebras and products. Thus, we will restrict our attention to $X$--generated $\mathscr{M}$--subalgebras, i.e., $\mathscr{M}$--subalgebras that are 
$X$--generated, and subdirect products. The latter are defined as follows.

\begin{defi}
	Let $\mathcal{D}$ be a complete category, $\mathsf{T}$ a monad on $\mathcal{D}$, $\mathscr{E}/\mathscr{M}$ a proper factorization system on $\mathcal{D}$ 
	such that $T$ preserves the morphisms in $\mathscr{E}$. Let $X\in \mathcal{D}$ and  
	let $\mathbf{A}_i$ be an $X$--generated $\mathsf{T}$--algebra with $e_i\in \Alg(\mathsf{T})(\mathbf{TX},\mathbf{A}_i)\cap \mathscr{E}$, $i\in I$. We define the 
	{\it subdirect product} of the family $\{(\mathbf{A}_i,e_i)\}_{i\in I}$ as the $X$--generated $\mathscr{M}$--subalgebra $\mathbf{S}$ of $\prod_{i\in I}\mathbf{A}_i$ described 
	in the following commutative diagram:
	\begin{center}	
		\begin{tikzpicture}[>=stealth,shorten >=3pt,
				node distance=2.5cm,on grid,auto,initial text=,
				accepting/.style={thick,double distance=2.5pt}]
				\node (a) at (0,1.5) {$TX$};
				\node (b) at (0,0) {$\prod_{i\in I}A_i$};
				\node (c) at (2.5,0) {$A_j$};
				\node (d) at (-2.5,0) {$S$};
				\path[->>] (a) edge [] node [right] {$e_j$} (c)
							edge [] node [left] {$e_e\ $} (d);
				\path[->] (b) edge [] node [below] {$\pi_j$} (c)
						(a) edge [] node [right] {$e$} (b);
				\path[right hook->] (d) edge [] node [below] {$m_e$} (b);
		\end{tikzpicture}
	\end{center}
	where $e$ is obtained from the morphisms $e_j$, $j\in I$, and the universal property of the product $\prod_{i\in I}\mathbf{A}_i$ and 
	$e=m_e\circ e_e$ is the factorization of $e$. We say that the subdirect product $\mathbf{S}$ defined above is {\it finite} if $I$ is a finite set.
\end{defi}

\noindent To obtain local versions of Eilenberg--type correspondences, the concept of (pseudo)variety used is: classes of (finite) $X$--generated $\mathsf{T}$--algebras 
closed under $\mathscr{E}$--quotients, $X$--generated $\mathscr{M}$--subalgebras and (finite) subdirect products. We state the two corresponding local versions in the 
rest of this section. Proofs are made in a similar way by using local versions of Birkhoff's theorem for (finite) $\mathsf{T}$--algebras.

\subsection{Eilenberg--type correspondence for local varieties of $\mathsf{T}$--algebras}

In this subsection, we provide Eilenberg--type correspondences for local varieties of $\mathsf{T}$--algebras. For this purpose, as in Section \ref{secbirk}, we first 
provide a local version of Birkhoff's theorem.

\noindent We fix a complete category $\mathcal{D}$, a monad $\mathsf{T}=(T,\eta,\mu)$ on $\mathcal{D}$, $\mathscr{E}/\mathscr{M}$ a factorization system 
on $\mathcal{D}$ and $X\in \mathcal{D}$. We will use the following assumptions:

\begin{enumerate}
	\item[(b1)] The factorization system $\mathscr{E}/\mathscr{M}$ is proper.
	\item[(b2)] The free $\mathsf{T}$--algebra $\mathbf{TX}=(TX,\mu_X)$ is {\it projective with respect to $\mathscr{E}$ in $\Alg(\mathsf{T})$}.
	\item[(b3)] $T$ preserves morphisms in $\mathscr{E}$.
	\item[(b4)] There is, up to isomorphism, only a set of $\mathsf{T}$--algebra morphisms in $\mathscr{E}$ with domain $\mathbf{TX}$.
\end{enumerate}

\begin{defi}
	Let $\mathcal{D}$ be a complete category, $\mathsf{T}$ a monad on $\mathcal{D}$, and $\mathscr{E}/\mathscr{M}$ a factorization system on $\mathcal{D}$. 
	Assume (b1) and (b3). Let $X\in \mathcal{D}$. A class $K$ of $X$--generated $\mathsf{T}$--algebras is a {\it local variety of $X$--generated $\mathsf{T}$--algebras} 
	if it is closed under $\mathscr{E}$--quotients, $X$--generated $\mathscr{M}$--subalgebras and subdirect products. 
\end{defi}

\begin{defi}\label{defleqth}
	Let $\mathcal{D}$ be a category, $\mathsf{T}$ a monad on $\mathcal{D}$, $X\in \mathcal{D}$ and $\mathscr{E}/\mathscr{M}$ a 
	factorization system on $\mathcal{D}$. A {\it local equational $\mathsf{T}$--theory on $X$} is a $\mathsf{T}$--algebra morphism $TX\overset{e_X}{\lepi} Q_X$ in $\mathscr{E}$ 
	such that for any $g\in \Alg(\mathsf{T})(\mathbf{TX},\mathbf{TX})$ there exists $g'\in \Alg(\mathsf{T})(\mathbf{Q_X},\mathbf{Q_X})$ such that the 
	following diagram commutes:
	\begin{center}	
		\begin{tikzpicture}[>=stealth,shorten >=3pt,
				node distance=2.5cm,on grid,auto,initial text=,
				accepting/.style={thick,double distance=2.5pt}]
				\node (a) at (2,1.5) {$TX$};
				\node (b) at (2,0) {$Q_X$};
				\node (c) at (0,1.5) {$TX$};
				\node (d) at (0,0) {$Q_X$};
				\path[->>] (a) edge [] node [right] {$e_X$} (b)
						(c) edge [] node [left] {$e_X$} (d);
				\path[->] (c) edge [] node [above] {$\forall g$} (a);
				\path[->,dashed] (d) edge [] node [below] {$g'$} (b);
		\end{tikzpicture}
	\end{center}	
\end{defi}

\noindent Note that, in the setting of Example \ref{excleqth}, a local equational $\mathsf{T}$--theory $TX\overset{e_X}{\lepi} Q_X$ on $X$ is characterized, up to isomorphism, by 
its kernel $\ker (e_X)$. In this case, the property of $e_X$ being a local equational $\mathsf{T}$--theory is exactly the property of $\ker(e_X)$ being a fully invariant congruence of 
$\mathbf{TX}$ \cite[Definition II.14.1]{bys}. This generalizes the definition of an equational theory over $X$ in \cite[Definition II.14.9]{bys} to a categorical level.

\noindent Given a local equational $\mathsf{T}$--theory $TX\overset{e_X}{\lepi} Q_X$ on $X$ and an $X$--generated $\mathsf{T}$--algebra $\mathbf{A}$, we say that 
$\mathbf{A}$ {\it satisfies} $e_X$, denoted as $\mathbf{A}\models e_X$, if every morphism $f\in \Alg(\mathsf{T})(\mathbf{TX},\mathbf{A})$ factors through $e_X$. 
We denote by $\md (e_X)$ the {\it $X$--generated models of} $e_X$, that is:
\[
	\md(e_X)=\{\mathbf{A}\in \Alg(\mathsf{T})\mid \text{ $\mathbf{A}$ is $X$--generated and $\mathbf{A}\models e_X$}\}.
\]
A class $K$ of $X$--generated $\mathsf{T}$--algebras is {\it defined} by $e_X$ if $K=\md(e_X)$.

\begin{thm}[Local Birkhoff's theorem for $\mathsf{T}$--algebras]\label{lbirkthm}
	Let $\mathcal{D}$ be a complete category, $\mathsf{T}$ a monad on $\mathcal{D}$, $X\in \mathcal{D}$ and $\mathscr{E}/\mathscr{M}$ a factorization system on 
	$\mathcal{D}$. Assume (b1) to (b4). Then a class $K$ of $X$--generated $\mathsf{T}$--algebras is a local variety of $X$--generated $\mathsf{T}$--algebras 
	if and only if is defined by a local equational $\mathsf{T}$--theory on $X$. Additionally, local varieties of $X$--generated $\mathsf{T}$--algebras are in one--to--one 
	correspondence with local equational $\mathsf{T}$--theories on $X$.\qed
\end{thm}

\begin{defi}\label{deflcoeqth}
	Let $\mathcal{C}$ be a category, $\mathsf{B}$ a comonad on $\mathcal{C}$, $Y\in \mathcal{C}$ and $\mathscr{E}/\mathscr{M}$ a factorization system on $\mathcal{C}$. 
	A {\it local coequational $\mathsf{B}$--theory on $Y$} is a $\mathsf{B}$--coalgebra morphism $S_Y\xhookrightarrow{m_Y} BY$ in $\mathscr{M}$ such that 
	for any $g\in \Coalg(\mathsf{B})(\mathbf{BY},\mathbf{BY})$ there exists $g'\in \Coalg(\mathsf{B})(\mathbf{S_Y},\mathbf{S_Y})$ such that 
	the following diagram 
	commutes:
	\begin{center}
		\begin{tikzpicture}[>=stealth,shorten >=3pt,
				node distance=2.5cm,on grid,auto,initial text=,
				accepting/.style={thick,double distance=2.5pt}]
				\node (a) at (2,0) {$S_Y$};
				\node (b) at (2,1.5) {$BY$};
				\node (c) at (0,0) {$S_Y$};
				\node (d) at (0,1.5) {$BY$};
				\path[right hook->] (a) edge [] node [right] {$m_Y$} (b)
						(c) edge [] node [left] {$m_Y$} (d);
				\path[->] (d) edge [] node [above] {$\forall g$} (b);
				\path[->,dashed] (c) edge [] node [below] {$g'$} (a);
		\end{tikzpicture}
	\end{center}
\end{defi}

\noindent With the previous definition, Theorem \ref{lbirkthm} and duality, we have the following.

\begin{prop}[Abstract Eilenberg--type correspondence for varieties of $X$--generated $\mathsf{T}$--algebras]\label{leilvar}
	Let $\mathcal{D}$ be a complete category, $\mathsf{T}$ a monad on $\mathcal{D}$, $\mathscr{E}/\mathscr{M}$ a factorization system on $\mathcal{D}$ and $X\in \mathcal{D}$. 
	Assume (b1) to (b4). Let $\mathcal{C}$ be a category that is dual to $\mathcal{D}$, $Y$ the  
	corresponding dual object of $X$ and let $\mathsf{B}$ be the comonad on $\mathcal{C}$ that is dual to $\mathsf{T}$ which is defined as in 
	Proposition \ref{monadtocomonad}. Then there is a one--to--one correspondence between local varieties of $X$--generated $\mathsf{T}$--algebras and local coequational 
	$\mathsf{B}$--theories on $Y$.
\end{prop}

\begin{example}\label{exlvv}
	By fixing an object $X\in \mathcal{D}_0$, we can get corresponding local versions of the Eilenberg--type correspondences showed in subsection \ref{seceilv} for varieties 
	of $X$--generated $\mathsf{T}$--algebras. For example, the local version of Example \ref{exjan11} reads as follows: There is a one--to--one correspondence between 
	varieties of $X$--generated monoids and subalgebras $\mathbf{S}\in \CABA$ of the complete atomic Boolean algebra $\set(X^*,2)$ such that:
	\begin{enumerate}[i)]
		\item $S$ is closed under left and right derivatives. That is, for every $L\in S$ and $x\in X$, $_xL,L_x\in S$.
		\item $S$ is closed under morphic preimages. That is, for every homomorphism of monoids $h:X^*\to X^*$ and $L\in S$, we have that $L\circ h\in S$.\qed
	\end{enumerate}
\end{example}

\subsection{Eilenberg--type correspondences for local pseudovarieties of $\mathsf{T}$--algebras}

In this subsection, we provide an abstract Eilenberg--type correspondence for local pseudovarieties of $\mathsf{T}$--algebras. In order to do this, we first provide a local 
version of Birkhoff's theorem for finite $\mathsf{T}$--algebras. 

\noindent We fix a complete concrete category $\mathcal{D}$ such that its forgetful functor preserves epis, monos and products, a monad $\mathsf{T}=(T,\eta,\mu)$ on $\mathcal{D}$, 
$X\in \mathcal{D}$ and a factorization system $\mathscr{E}/\mathscr{M}$ on $\mathcal{D}$. We will need the following assumptions:

\begin{enumerate}
	\item[(b$_f$1)] The factorization system $\mathscr{E}/\mathscr{M}$ is proper.
	\item[(b$_f$2)] The free $\mathsf{T}$--algebra $\mathbf{TX}=(TX,\mu_X)$ is {\it projective with respect to $\mathscr{E}$ in $\Alg(\mathsf{T})$}.
	\item[(b$_f$3)] $T$ preserves morphisms in $\mathscr{E}$.
\end{enumerate}

\begin{defi}\label{deflpet}
 	Let $\mathcal{D}$ be a concrete category such that its forgetful functor preserves epis and monos, $\mathsf{T}$ a monad on $\mathcal{D}$, $X\in \mathcal{D}$ and 
	$\mathscr{E}/\mathscr{M}$ a factorization system on $\mathcal{D}$. Assume (b$_f$1) and (b$_f$3). A {\it local pseudoequational $\mathsf{T}$--theory on $X$} is a nonempty 
	collection $\mathtt{P}_X$ of $\mathsf{T}$--algebra morphisms in $\mathscr{E}$ with domain $\mathbf{TX}$ and finite codomain such that:
	\begin{enumerate}[i)]
		\item For every finite set $I$ and $f_i\in \mathtt{P}_X$, $i\in I$, there exists $f\in \mathtt{P}_X$ such that $f_i$ factors through $f$, $i\in I$.
		\item For every $e\in \mathtt{P}_X$ with codomain $\mathbf{A}$ and every $\mathsf{T}$--algebra morphism $e'\in \mathscr{E}$ with domain $\mathbf{A}$ 
			we have that $e'\circ e\in \mathtt{P}_X$.
		\item For every $f\in \mathtt{P}_X$ and $h\in \Alg(\mathsf{T})(\mathbf{TX},\mathbf{TX})$ we have that $e_{f\circ h}\in \mathtt{P}_X$ where 
			$f\circ h=m_{f\circ h}\circ e_{f\circ h}$ is the factorization of $f\circ h$.
	\end{enumerate}
\end{defi}

\noindent Given an $X$--generated algebra $\mathbf{A}\in \Alg_f(\mathsf{T})$, we say that $\mathbf{A}$ {\it satisfies} $\mathtt{P}_X$, denoted as $\mathbf{A}\models \mathtt{P}_X$, 
if every $f\in\Alg(\mathsf{T})(\mathbf{TX},\mathbf{A})$ factors through some morphism in $\mathtt{P}_X$. We denote by $\md_f (\mathtt{P}_X)$ the {\it finite $X$--generated 
models of} $\mathtt{P}_X$, that is:
\[
	\md_f (\mathtt{P}_X):=\{\mathbf{A}\in \Alg_f(\mathsf{T}) \mid \text{$\mathbf{A}$ is $X$--generated and }\mathbf{A}\models \mathtt{P}_X\}.
\]
A class $K$ of finite $X$--generated $\mathsf{T}$--algebras is {\it defined} by $\mathtt{P}_X$ if $K=\md_f(\mathtt{P}_X)$.

\begin{defi}
	Let $\mathcal{D}$ be a complete concrete category, $\mathsf{T}$ a monad on $\mathcal{D}$, $\mathscr{E}/\mathscr{M}$ a factorization system on $\mathcal{D}$ and 
	$X\in \mathcal{D}$ a finite object. A class $K$ of finite $X$--generated algebras in $\Alg (\mathsf{T})$ is called a {\it local pseudovariety of $X$--generated 
	$\mathsf{T}$--algebras } if it is closed under $\mathscr{E}$--quotients, $X$--generated $\mathscr{M}$--subalgebras and finite subdirect products. 
\end{defi}

\begin{thm}[Local Birkhoff's Theorem for finite $\mathsf{T}$--algebras]\label{lreitthm}
	Let $\mathcal{D}$ be a concrete complete category such that its forgetful functor preserves epis, monos and products, $\mathsf{T}$ a monad on $\mathcal{D}$, 
	$X\in \mathcal{D}$ and $\mathscr{E}/\mathscr{M}$ a factorization system on $\mathcal{D}$. Assume (b$_f$1) to (b$_f$3). Then a class $K$ of finite $X$--generated 
	$\mathsf{T}$--algebras is a local pseudovariety of $X$--generated $\mathsf{T}$--algebras if and only if is defined by a local pseudoequational $\mathsf{T}$--theory on $X$. 
	Additionally, local pseudovarieties of $X$--generated $\mathsf{T}$--algebras are in one--to--one correspondence with local pseudoequational $\mathsf{T}$--theories on $X$.\qed
\end{thm}

\begin{defi}\label{deflpcet}
 	Let $\mathcal{C}$ be a concrete category such that its forgetful functor preserves monos, $\mathsf{B}$ a comonad on $\mathcal{C}$, $Y\in \mathcal{C}$ and 
	$\mathscr{E}/\mathscr{M}$ a factorization system on $\mathcal{C}$. Assume (b$_f$1) and that $B$ preserves the morphisms in $\mathscr{M}$. 
	A {\it local pseudocoequational $\mathsf{B}$--theory on $Y$} is a nonempty collection $\mathtt{R}_Y$ of 
	$\mathsf{B}$--coalgebra morphisms in $\mathscr{M}$ with codomain $\mathbf{BY}$ and finite codomain such that:
	\begin{enumerate}[i)]
		\item For every finite set $I$ and $f_i\in \mathtt{R}_Y$, $i\in I$, there exists $f\in \mathtt{R}_Y$ such that $f_i$ factors through $f$, $i\in I$.
		\item For every $m\in \mathtt{R}_Y$ with domain $\mathbf{A}$ and every $\mathsf{B}$--coalgebra morphism $m'\in \mathscr{M}$ with codomain $\mathbf{A}$ 
			we have that $m\circ m'\in \mathtt{R}_Y$.
		\item For every $f\in \mathtt{R}_Y$ and $h\in \Coalg(\mathsf{T})(\mathbf{BY},\mathbf{BY})$ we have that $m_{h\circ f}\in \mathtt{R}_Y$ where 
			$h\circ f=m_{h\circ f}\circ e_{h\circ f}$ is the factorization of $h\circ f$.
	\end{enumerate}
\end{defi}

\noindent With the previous definition, Theorem \ref{lreitthm} and duality, we have the following.

\begin{prop}[Abstract Eilenberg--type correspondence for pseudovarieties of $X$--gen\-erated $\mathsf{T}$--algebras]\label{leilpvar}
	Let $\mathcal{D}$ be a complete category, $\mathsf{T}$ a monad on $\mathcal{D}$, $\mathscr{E}/\mathscr{M}$ a factorization system on $\mathcal{D}$ and $X\in \mathcal{D}$. 
	Assume (b$_f$1) to (b$_f$3). Let $\mathcal{C}$ be a category that is dual to $\mathcal{D}$, $Y$ the corresponding dual object of $X$ and let $\mathsf{B}$ be the comonad on 
	$\mathcal{C}$ that is dual to $\mathsf{T}$ which is defined as in Proposition \ref{monadtocomonad}. 
	Then there is a one--to--one correspondence between local varieties of $X$--generated $\mathsf{T}$--algebras and local coequational $\mathsf{B}$--theories on $Y$.
\end{prop}

\begin{example}\label{exlpvv}
	By fixing an object $X\in \mathcal{D}_0$, we can get corresponding local versions of the Eilenberg--type correspondences showed in subsection \ref{secreiteil} for pseudovarieties 
	of $X$--generated $\mathsf{T}$--algebras. For example, the local version of (1) in Example \ref{expvv1} reads as follows: There is a one--to--one correspondence between 
	pseudovarieties of $X$--generated monoids and Boolean algebras $\mathbf{S}$ that are subalgebras of the complete atomic Boolean algebra $\set(X^*,2)$ such that:
	\begin{enumerate}[i)]
		\item Every element in $S$ is a recognizable language on $X$.
		\item $S$ is closed under left and right derivatives. That is, for every $L\in S$ and $x\in X$, $_xL,L_x\in S$.
		\item $S$ is closed under morphic preimages. That is, for every homomorphism of monoids $h:X^*\to X^*$ and $L\in S$, we have that $L\circ h\in S$.\qed
	\end{enumerate}
\end{example}

\section{Conclusions}

We proved that Eilenberg--type correspondences = Birkhoff's theorem for (finite) algebras + duality. The main contribution of the present paper was to realize that the concept 
of a ``variety of languages'' that is used in Eilenberg--type correspondences corresponds to the dual of the (pseudo)equational theory that defines the (pseudo)variety of algebras, which was 
conjectured by the author in \cite{jste}. This not only allows us to understand where ``varieties of languages'' come from but also to get an abstract and general result that 
encompasses both existing and new Eilenberg--type correspondences.

\noindent Our algebras of interest are $\mathsf{T}$--algebras, where $\mathsf{T}$ is a monad on a category $\mathcal{D}$. We 
stated and proved, under mild assumptions, categorical versions of Birkhoff's theorem for $\mathsf{T}$--algebras and Birkhoff's theorem for finite $\mathsf{T}$--algebras, the latter also 
known as Reiterman's theorem for $\mathsf{T}$--algebras. In order to get Eilenberg--type correspondences we stated Birkhoff's theorem for (finite) $\mathsf{T}$--algebras as a 
one--to--one correspondence between (pseudo)varieties of $\mathsf{T}$--algebras and (pseudo)equational $\mathsf{T}$--theories. The previous observation led us to define the 
notion of a (pseudo)equational $\mathsf{T}$--theory, i.e., collections of ``(pseudo)equations'' that are 
deductively closed. The proof of Birkhoff's theorem for $\mathsf{T}$--algebras is obtained from \cite{bana}. On the other hand, the proof of Birkhoff's theorem for finite 
$\mathsf{T}$--algebras follows the general idea used by the author in \cite{jste} to prove a general Eilenberg--type correspondence and the idea that pseudovarieties of algebras are 
exactly directed unions of equational classes of finite algebras \cite[Proposition 4]{banab}. It is worth mentioning that the proof of Birkhoff's theorem for finite algebras in the present paper 
has the advantage of avoiding topological and profinite techniques, which are usually used to prove this theorem. This help us to understand the proof of Birkhoff's theorem for finite algebras 
in a more basic setting.

\noindent Once we stated our categorical versions of Birkhoff's theorem and Birkhoff's theorem for finite algebras, with the use of duality, we stated our abstract Eilenberg--type 
correspondence theorems for varieties and pseudovarieties of $\mathsf{T}$--algebras. Then, in a similar way, in Section \ref{seclocal}, we derived corresponding local versions of 
Birkhoff's theorem for (finite) $X$--generated $\mathsf{T}$--algebras and their corresponding Eilenberg--type correspondences. These local versions can be seen as particular 
instances of the previous work in which the category $\mathcal{D}_0$ has only one object, say $X$. Proofs are made in a similar way by restricting the kind of algebras to 
$X$--generated $\mathsf{T}$--algebras. Thus the closure properties for local (pseudo)varieties are closure under $\mathscr{E}$--quotients, $X$--generated 
$\mathscr{M}$--subalgebras and (finite) subdirect products. 

\noindent From our abstract Eilenberg--type correspondence theorems we derived both existing and new Eilenberg--type correspondences, including the ones in the following table:

\begin{center}
\begin{tabular}{|m{3.7cm}|m{2.5cm}|m{2.4cm}|m{2.1cm}|m{2cm}|}
	\hline
	Eilenberg--type correspondence for... & Pseudovariety version & Local pseudovariety version & Variety version & Local variety version \\
	\hline
	Semigroups & \cite[Thm. 34s]{eilenberg} (Ex. \ref{expvv1} (2)) & \cite{gehrke} (Ex. \ref{exlpvv}) & Ex. \ref{exvv1} (1) & Ex. \ref{exlvv} \\
	\hline
	Ordered semigroups & \cite{pin} (Ex. \ref{expvv3}) & \cite{gehrke} (Ex. \ref{exlpvv}) & Ex. \ref{exvv3} & Ex. \ref{exlvv} \\
	\hline
	Monoids & \cite[Thm. 34]{eilenberg} (Ex. \ref{expvv1} (1)) & \cite{gehrke} (Ex. \ref{exlpvv}) & \cite[Thm. 39]{jan1} (Ex. \ref{exjan11}) & Ex. \ref{exlvv} \\
	\hline
	Ordered monoids & \cite{pin} (Ex. \ref{expvv3}) & \cite{gehrke} (Ex. \ref{exlpvv}) & Ex. \ref{exvv3} & Ex. \ref{exlvv} \\
	\hline
	Groups & Ex. \ref{expvv1} (3) & Ex. \ref{exlpvv} & Ex. \ref{exvv1} (2) & Ex. \ref{exlvv} \\
	\hline
	Ordered groups & Ex. \ref{expvv3} & Ex. \ref{exlpvv} & Ex. \ref{exvv3} & Ex. \ref{exlvv} \\
	\hline
	Monoid actions (dynamical systems) & Ex. \ref{expvv1} (4) & Ex. \ref{exlpvv} & Ex. \ref{exvv1} (3) & Ex. \ref{exlvv} \\
	\hline
	Semigroups with infinite exponentiation & cf. \cite{wilke} (Ex. \ref{expvv1} (5)) & cf. \cite[Thm. 6.3]{urbat} (Ex. \ref{exlpvv}) & Ex. \ref{exvv1} (4) & Ex. \ref{exlvv}\\
	\hline
	$\mathbb{K}$--algebras for a finite field $\mathbb{K}$ & \cite{reut} (Ex. \ref{expvv4}) & \cite{adamekl} for $\mathbb{K}=\mathbb{Z}_2$ (Ex. \ref{exlpvv}) & Ex. \ref{exvv4} 
			& Ex. \ref{exlvv}\\
	\hline
	Idempotent semirings & \cite{polak} (Ex. \ref{expvv5}) & \cite{adamekl} (Ex. \ref{exlpvv}) & Ex. \ref{exvv5} & Ex. \ref{exlvv} \\
	\hline
	Algebras of type $\T$ in a variety & Ex. \ref{exuaf} & Ex. \ref{exlpvv} & Ex. \ref{exua} & Ex. \ref{exlvv} \\
	\hline
	Ordered algebras of type $\T$ in a variety & Ex. \ref{exuoaf} & Ex. \ref{exlpvv} & Ex. \ref{exuoa} & Ex. \ref{exlvv} \\
	\hline
\end{tabular}
\end{center}

\subsection{Related work:} We will discuss some of the related work of this paper.

\noindent \textsc{Birkhoff's theorem:} We discuss categorical approaches for Birkhoff's theorem such as \cite{awodeyh,bana,barr}. The main purpose of the present 
paper was to prove abstract Eilenberg--type correspondences which, in the case of varieties of $\mathsf{T}$--algebras, led us to state a Birkhoff's theorem for $\mathsf{T}$--algebras 
in which every variety is defined by a unique collection of equations and vice versa. This version is obtained from \cite{bana} after defining the right notion of an equational 
$\mathsf{T}$--theory. In \cite{awodeyh}, the defining properties for a variety are taken with respect to the factorization system $\mathscr{E}/\mathscr{M}$ where $\mathscr{E}=$ 
regular epi and $\mathscr{M}=$ mono, \cite[Definition 2.1 and 2.2]{awodeyh}. In approaches such as \cite{awodeyh,bana} the morphisms that represent equations are epimorphisms 
with projective domain, which is stated in condition (B2), and there is also the requirement of having enough projectives, which is implied by conditions (B2) and (B3). 
In \cite{barr}, they work also with $\mathsf{T}$--algebras and a factorization system $\mathscr{E}/\mathscr{M}$ on $\Alg(\mathsf{T})$. In the present paper, the factorization system 
$\mathscr{E}/\mathscr{M}$ is on $\mathcal{D}$ which, under conditions (B1) and (B4), is lifted to $\Alg(\mathsf{T})$, but not every factorization system on $\Alg(\mathsf{T})$ is induced 
by one on the base category $\mathcal{D}$ (e.g., on $\set$ there is no factorization system that corresponds in any way to epimorphisms in the categories of monoids or of rings). 
In \cite{barr}, the defining properties of a variety are closure under $U$--split quotients, $\mathscr{M}$--subalgebras and products, where $U:\Alg(\mathsf{T})\to \mathcal{D}$ is the 
forgetful functor.

\noindent  \textsc{Birkhoff's theorem for finite algebras:} Birkhoff's theorem for finite algebras, which is also known as Reiterman's theorem \cite{reit}, has been generalized in \cite{banab,chen}. The approach in \cite{reit} was to consider implicit operations. Implicit operations generalize the notion of terms. Equations given by implicit operations are the kind of equations that define 
pseudovarieties. The proof given in \cite{reit} involves the use of topology in which the set of $n$--ary implicit operations is the completion of the set of $n$--ary terms. 
In \cite{banab}, a topological approach is also considered by using uniformities, and it is also shown that pseudovarieties are exactly directed unions of equational classes of finite algebras 
\cite[Proposition 4]{banab}. In \cite{chen}, a categorical approach is considered to prove a Reiterman's theorem for $\mathsf{T}$--algebras, this is done by using profinite techniques 
to define the notion of profinite equation which are the kind of equations that allow to define and characterize pseudovarieties. In \cite{chen}, for a given monad $\mathsf{T}$ on 
$\mathcal{D}$ they define the profinite monad $\widehat{\mathsf{T}}$ on the profinite completion $\widehat{\mathcal{D}}$ of the category $\mathcal{D}_f$ which is done by using 
limits (in fact, right Kan extensions). The approach in the present paper do not use topological nor profinite techniques, and it is based in the fact that pseudovarieties are 
exactly directed unions of equational classes of finite algebras, see \cite[Proposition 4]{banab} and \cite{baldwin,eilsch}. Nevertheless, profinite and topological techniques can be easily 
brought to the scene in the present paper if we identify the family of morphisms $\Ps(X)$ by its limit, where $\Ps$ is a pseudoequational $\mathsf{T}$--theory. This would have led us to 
deal with profinite completions and topological spaces, in particular, profinite monoids, Stone spaces and Stone duality. We prefer to avoid this approach for the following reasons:
\begin{enumerate}[a)]
	\item Make the present work more accessible to some readers.
	\item To present a different approach without using topology and profinite techniques.
	\item Eilenberg--type correspondences deal with pseudocoequational theories rather than its dual, i.e., pseudoequational theories. 
\end{enumerate}

\noindent  \textsc{Categorical Eilenberg--type correspondences:} There are some categorical approaches for Eilenberg--type correspondences in the literature such as \cite{adamek,mikolaj,jste,urbat}. In \cite{adamek,mikolaj,urbat} only 
pseudovarieties are considered, i.e., all the algebras are finite, while in \cite{jste} as well as in the present paper we can also consider varieties of algebras. In this respect, 
Eilenberg--type correspondences such as \cite[Theorem 39]{jan1} or the ones derived in examples of subsection \ref{seceilv} cannot be derived from \cite{adamek,mikolaj,urbat}. 
The work in \cite{urbat} subsumes the work made in \cite{adamek,mikolaj}. The main setting in \cite{adamek,urbat} is to consider predual categories, i.e., categories that are dual on finite objects. The main purpose of this preduality is to define pseudovarieties of algebras on one category and varieties of languages on the other one. In \cite{mikolaj}, no duality is involved. 
The definition of varieties of languages given in \cite{mikolaj} is restricted in the sense that it is always a Boolean algebra, which in our present paper reduces to consider 
$\mathcal{D}=\set$, $\mathcal{D}_0=\set_f^S$, $\mathcal{C}=\CABA$ and $\mathcal{C}_0=\CABA_f^S$, where $S$ is a fixed set. Eilenbeg--type correspondences such as 
\cite{pin,reut} cannot be derived 
from \cite{mikolaj}. In \cite{adamek}, all the algebras considered have a monoid structure which restricts the kind of algebras one can consider, e.g., a semigroup version of Eilenberg's theorem \cite[Theorem 34s]{eilenberg} cannot be derived from \cite{adamek}. Those previous two limitations in the kind of varieties of languages and the kind of algebras one can 
consider are overcome in \cite{jste,urbat} as well as in the present paper by considering algebras for a monad $\mathsf{T}$, which is the main idea in \cite{mikolaj}. 
In \cite{urbat}, the preduality considered as its main setting is lifted, under mild assumptions, to a full duality between one of the categories and the profinite completion of the other one. 
From this profinite completion, the concept of profinite equations is defined which are the kind of equations that define pseudovarieties of algebras. On the other hand, the definition of 
``varieties of languages'' given in \cite{urbat} depends on finding a ``unary representation'' which is a set of unary operations on a free algebra satisfying certain properties 
\cite[Definition 37]{urbat} and requires non--trivial work. From this unary representation they construct syntactic algebras and define the kind of derivatives that define a variety of languages. 
In \cite{jste} as well as in the present paper, the use of derivatives is not explicitely made which is captured in a more transparent and categorical way by using coalgebras, from this, 
the righ notion of derivatives easily follows by using duality and the defining properties of a $\mathsf{T}$--algebra (epi)morphism  
(see, e.g., Example \ref{exjan11}). The coalgebraic approach used in the present paper gave us the advantage to obtain what we called an abstract Eilenberg--type correspondence 
in which the concept of ``variety of languages'' is the one of being a (pseudo)coequational $\mathsf{B}$--theory, whose definition does not depend on finding the right notion of 
derivatives, contrary to \cite{adamek,mikolaj,urbat}, and does not depend on the existence of syntactic algebras, contrary to \cite{mikolaj,urbat}. Also, requirements considered in \cite{urbat} 
such as existence of ``unary representations'' or the fact that $\mathcal{D}$ and $\mathcal{C}$ are dual on finite objects are not needed to state our abstract Eilenberg--type correspondences. 
Nevertheless, in specific applications such as, e.g., Example \ref{exuaf}, the fact that the functor $\set(\underline{\ \ },2)$, which is part of the duality between $\set$ and $\CABA$, 
preserves finite objects allowed us to identify the dual of the family $\Ps(X)$ with the Boolean algebra $\Lan_{\Ps}(X):=\bigcup_{e\in \Ps(X)}\Imag(\set(e,2))$. But again, the fact that 
$\mathcal{C}$ and $\mathcal{D}$ are dual on finite objects does not play any role in our abstract Eilenberg--type correspondence.

\noindent The work in the present paper subsumes the work made by the author in \cite{jste} from which most of the ideas presented in this paper were obtained. The main idea 
of considering ``varieties of languages'' as coequations was initially made in \cite[Proposition 12]{jste} which was the starting point to suspect that ``varieties of languages'' are 
exactly duals of equational theories. With this in mind, the main work was focused on finding a proper definition of a (pseudo)equational theory and state categorical versions of 
Birkhoff's theorem and Birkhoff's theorem for finite $\mathsf{T}$--algebras to get one--to--one correspondences between (pseudo)varieties of $\mathsf{T}$--algebras and 
(pseudo)equational $\mathsf{T}$--theories. As a consequence, we now clearly understand where ``varieties of languages'' come from and how to derive and find their defining 
properties in each particular case, e.g., derivatives come from the properties that characterize a $\mathsf{T}$--algebra (epi)morphism and closure under morphic 
preimages, which is a property that it is always present, comes from the substitution property in a (pseudo)equational theory.

\noindent  \textsc{The use of duality:} Related work such as \cite{gehrke0,gehrke,gehrke1} have influenced and motivated the use of duality in language theory to characterize 
recognizable languages and to derive local versions of Eilenberg--type correspondences. In fact, in this paper we show that duality is an ingredient to obtain (abstract) 
Eilenberg--type correspondences. The most important aspect of this is in each concrete case of an Eilenberg--type correspondence, in which, by using the interaction between 
algebra and coalgebra and equations and coequations, one can easily find the right notion of derivatives as shown in the examples. Previous categorical approaches for 
Eilenberg--type correspondences such as \cite{adamek,mikolaj,jste,urbat} have used duality, either explicitely or implicitely, but the fact that the dual of a pseudovariety of languages 
is exactly the pseudoequational theory that defines the given pseudovariety of algebras has not been brought to light to derive and understand Eilenberg--type correspondences. 
For instance:
\begin{enumerate}[i)]
	\item None of the other categorical approaches to derive Eilenberg--type correspondences relates the dual of a (pseudo)variety of languages with a (pseudo)equational theory 
		or equations to obtain the defining properties of a (pseudo)variety of languages.
	\item Derivatives are not directly obtained via duality in \cite{adamek,mikolaj,urbat}, no interaction between algebra and coalgebra or equations and coequations.
	\item Approaches such as \cite{adamek,urbat} require that the categories $\mathcal{C}_f$ and $\mathcal{D}_f$ are dual in order to obtain an Eilenberg--type correspondence, 
		which in our abstract Eilenberg--type correspondence theorem is not necessary.
\end{enumerate}
It's is worth mentioning that some of the aspects of this paper have been previously studied, either implicitely or explicitely, in relation with Eilenberg--type correspondences. In fact:
\begin{enumerate}[i)]
	\item In \cite[Section 4]{gehrke00}, it is mentioned, for the concrete case of monoids, that the dual of a ``local variety of languages'' will induce a set of (in)equations. 
		In the present paper, this is seen in Example \ref{exlpvv} for the free (ordered) monoid monad on $\set$ ($\poset$).
	\item In \cite{gehrke0}, recognizable subsets of an algebra with a single binary operation are studied, which in the present paper is the case of the local version of Example \ref{exuaf} 
		for the type of algebras $\T'$ which consists of a single binary operation. In fact, ``closure under residuals w.r.t. singleton denominators'' in \cite{gehrke0} is the same 
		as closure under derivatives with respect to $\T'$ in the sense of Example \ref{exuaf}.
	\item The duality between equations and coequations in the context of an Eilenberg--type correspondence is studied for the first time in \cite{jan1}. There, the 
		``varieties of languages'' considered have a closure property which is defined in terms of coequations \cite[Definition 40]{jan1}.
	\item The treatment of ``varieties of languages'' as sets of coequations was considered in \cite[Proposition 2]{jste}. There, in the conclusions, was also conjectured that 
		``varieties of languages'' are exactly (pseudo)coequational theories and their dual are the defining (pseudo)equational theories for the (pseudo)varieties of algebras.
\end{enumerate}

\noindent  \textsc{Syntactic algebras:} Another important observation and conclusion of the present paper is about the use of syntactic algebras. In Eilenberg's original proof \cite[Theorem 34]{eilenberg} the use of 
syntactic monoids (semigroups) \cite[VII.1]{eilenberg} helped to prove his theorem. As in Eilenberg's proof, the use of syntactic algebras was also used in 
\cite{mikolaj,pin,polak,reut,urbat,jste} for establishing Eilenberg--type correspondences. Categorical approaches such as \cite{mikolaj,jste,urbat} generalized the concept of syntactic 
algebra. In \cite{mikolaj,urbat} syntactic algebras are obtained, under mild assumptions, by means of a congruence, while in \cite{jste} are obtained by using generalized pushouts, 
under the condition that $T$ preserves weak generalized pushouts. As we saw in the present paper, the use of syntactic algebras is not needed in order to establish abstract 
Eilenberg--type correspondences. Nonetheless, the study of syntactic algebras has their own importance in language theory and categorical generalizations of them such 
as \cite{mikolaj,jste,urbat} might deserve a further study.

\subsection{Future work:} Some of the future work that can be done based on the work presented in this paper include the following:

\begin{enumerate}[i)]
	\item The relation of equational $\mathsf{T}$--theories with monad morphisms $\alpha:\mathsf{T}\to \mathsf{S}$ as in \cite{barr}. 
	\item To find new Eilenberg--type correspondences as an application of our abstract (local) Eilenberg--type theorem for (pseudo)varieties of $\mathsf{T}$--algebras.
	\item To study Eilenberg--type correspondences for other classes of algebras, e.g., Eilenberg--type correspondences for quasivarieties. Which can be made by modifying the notion 
		of an equational theory given in the present paper by allowing $\mathsf{T}$--algebra morphisms with arbitrary domain (not necessarily a free one) see, e.g., \cite{bana}.		
	\item To find applications of the Eilenberg--type correspondences we can derive from the present paper. One example of this could be 
		to characterize the (pseudo)equational $\mathsf{B}$--theory that defines a particular (pseudo)variety of $\mathsf{T}$--algebras. This kind of problem 
		has been studied before in which the pseudovariety of aperiodic monoids is defined by the variety of languages in which every language is star--free 
		\cite{schut}.
	\item The study and applications of the dual theorems in this paper. That is, coBirkhoff's theorem \cite{awodeyh,kurz,kurzr}, coReiterman's theorem and a new subject that we can 
		call coEilenberg--type correspondences which are naturally defined as one--to--one correspondences between (pseudo)covarieties of $\mathsf{B}$--coalgebras 
		and (pseudo)equational 	$\mathsf{T}$--theories.
	\item To develop and study a general theory for syntactic algebras. As we mentioned, syntactic algebras are not used in the present paper to establish abstract 
		Eilenberg--type correspondences. In a previous research made by the author in \cite{jste}, to prove a general Eilenberg--type theorem, syntactic algebras were also 
		considered and constructed abstractly as a generalized pushout \cite[Proposition 10]{jste}. General syntactic algebras were also considered and constructed in \cite{mikolaj,urbat}.
\end{enumerate}

\section*{Acknowledgements}
I would like to thank Marcello Bonsangue, Jan Rutten and Alexander Kurz for their support, comments, and suggestions during the writing of this paper. 
I thank Henning Urbat, Ji\v{r}\'i Ad\'amek, Liang--Ting Chen and Stefan Milius for earlier discussions we had on this subject. I also thank anonymous referees for their valuable comments 
and suggestions on an earlier version of this paper.


\newpage

\appendix
\section{}

\subsection{Details for Section \ref{prel}}

\subsubsection{Proof of Lemma \ref{lemmafs}}

\ \\ {\bf Lemma \ref{lemmafs}.} {\it 
	Let $\mathcal{D}$ be a category and $\mathscr{E}/\mathscr{M}$ be a factorization system on $\mathcal{D}$ such that every morphism in 
	$\mathscr{M}$ is mono. Then $f\circ g\in \mathscr{E}$ implies $f\in \mathscr{E}$.
}

\begin{proof}
	Put $f=m\circ e$ with $e\in \mathscr{E}$ and $m\in \mathscr{M}$. Then $m\circ e\circ g=f\circ g\in \mathscr{E}$. From the diagram
	\begin{center}	
				\begin{tikzpicture}[>=stealth,shorten >=3pt,
						node distance=2.5cm,on grid,auto,initial text=,
						accepting/.style={thick,double distance=2.5pt}]
						\node (a) at (0,0) {$\cdot$};
						\node (b) at (2,0) {$\cdot$};
						\node (c) at (0,1.5) {$\cdot$};
						\node (d) at (2,1.5) {$\cdot$};
						\path[right hook->] (a) edge [] node [below] {$m$} (b);
						\path[->>] (c) edge [] node [above] {$m\circ e\circ g$} (d);
						\path[->] (c) edge [] node [left] {$e\circ g$} (a)
								(d) edge [] node [right] {$id$} (b);
				\end{tikzpicture}
	\end{center}
	using the diagonal fill--in we get a morphism $d$ such that $m\circ d=id$. Now, from $m\circ d\circ m=m$, by using the fact that $m$ is mono, we have 
	$d\circ m=id$, i.e., $m$ is iso. Therefore $f=m\circ e\in \mathscr{E}$ since $e\in \mathscr{E}$ and $m$ is iso.
\end{proof}

\subsubsection{Proof of Lemma \ref{lemmafs1}}

\ \\ {\bf Lemma \ref{lemmafs1}.} {\it
	Let $\mathcal{D}$ be a category, $\mathsf{T}=(T,\eta,\mu)$ a monad on $\mathcal{D}$, and $\mathscr{E}/\mathscr{M}$ a proper factorization system on 
	$\mathcal{D}$. If $T$ preserves the morphisms in $\mathscr{E}$ then $\Alg(\mathsf{T})$ inherits the same $\mathscr{E}/\mathscr{M}$ factorization system. That is: 
	\begin{enumerate}[A)]
		\item Given $\mathbf{A}=(A,\alpha)$ and $\mathbf{B}=(B,\beta)$ such that $\mathbf{A},\mathbf{B}\in \Alg(\mathsf{T})$, if 
			$f\in \Alg(\mathsf{T})(\mathbf{A},\mathbf{B})$ is factored as $f=m\circ e$ in $\mathcal{D}$ where $e\in\mathscr{E}$, $m\in \mathscr{M}$ and $C\in \mathcal{D}$ is the 
			domain of $m$, then there 
			exists a unique $\gamma\in \mathcal{D}(TC,C)$ such that $\mathbf{C}=(C,\gamma)\in \Alg(\mathsf{T})$ and $m$ and $e$ are $\mathsf{T}$--algebra morphisms.
		\item Given any commutative diagram 
			\begin{center}	
				\begin{tikzpicture}[>=stealth,shorten >=3pt,
						node distance=2.5cm,on grid,auto,initial text=,
						accepting/.style={thick,double distance=2.5pt}]
						\node (a) at (0,0) {$C$};
						\node (b) at (2,0) {$D$};
						\node (c) at (0,1.5) {$A$};
						\node (d) at (2,1.5) {$B$};
						\path[right hook->] (a) edge [] node [below] {$m$} (b);
						\path[->>] (c) edge [] node [above] {$e$} (d);
						\path[->] (c) edge [] node [left] {$f$} (a)
								(d) edge [] node [right] {$g$} (b);
				\end{tikzpicture}
			\end{center}
			in $\Alg(\mathsf{T})$ with $e\in \mathscr{E}$ and $m\in \mathscr{M}$, the unique {\it diagonal fill--in} morphism $d$ such that the diagram 
			\begin{center}	
				\begin{tikzpicture}[>=stealth,shorten >=3pt,
						node distance=2.5cm,on grid,auto,initial text=,
						accepting/.style={thick,double distance=2.5pt}]
						\node (a) at (0,0) {$C$};
						\node (b) at (2,0) {$D$};
						\node (c) at (0,1.5) {$A$};
						\node (d) at (2,1.5) {$B$};
						\path[right hook->] (a) edge [] node [below] {$m$} (b);
						\path[->>] (c) edge [] node [above] {$e$} (d);
						\path[->] (c) edge [] node [left] {$f$} (a)
								(d) edge [] node [above] {$d$} (a)
								(d) edge [] node [right] {$g$} (b);
				\end{tikzpicture}
			\end{center}
			commutes is a morphism in $\Alg(\mathsf{T})$.
	\end{enumerate}
}

\begin{proof}\hfill
	\begin{enumerate}[A)]
		\item As $f\in \Alg(\mathsf{T})(A,B)$ and $f=m\circ e$, we get the following commutative diagram:
			\begin{center}	
			\begin{tikzpicture}[>=stealth,shorten >=3pt,
				node distance=2.5cm,on grid,auto,initial text=,
				accepting/.style={thick,double distance=2.5pt}]
				\node (a) at (0,0) {$A$};
				\node (b) at (2,0) {$C$};
				\node (c) at (4,0) {$B$};
				\node (d) at (0,2) {$TA$};
				\node (e) at (2,2) {$TC$};
				\node (f) at (4,2) {$TB$};
				\path[->>] (a) edge [] node [above] {$e$} (b)
						(d) edge [] node [below] {$Te$} (e);
				\path[right hook->] (b) edge [] node [below] {$m$} (c);
				\path[->] (d) edge [] node [left] {$\alpha$} (a)
						(e) edge [] node [below] {$Tm$} (f)
						(f) edge [] node [right] {$\beta$} (c);
				\path[->,dashed] (e) edge [] node [right] {$\gamma$} (b);
				\path[->,bend angle=30] (a) edge [bend right,looseness=0.7] node [below] {$f$} (c)
						(d) edge [bend left,looseness=0.7] node [above] {$Tf$} (f);
			\end{tikzpicture}
			\end{center}
			where $\gamma\in \mathcal{D}(TC,C)$ is obtained by the diagonal fill--in property since $Te\in \mathscr{E}$. Now, we prove that $C=(C,\gamma)\in \Alg(\mathsf{T})$. 
			In fact,
			\begin{enumerate}[i)]
				\item We prove that $\gamma\circ \eta_C=id_C$. In fact, we have the following commutative diagram:
					\begin{center}	
					\begin{tikzpicture}[>=stealth,shorten >=3pt,
						node distance=2.5cm,on grid,auto,initial text=,
						accepting/.style={thick,double distance=2.5pt}]
						\node (a) at (5.7,2) {$A$};
						\node (b) at (2,0) {$C$};
						\node (c) at (3.8,0) {$B$};
						\node (d) at (2,4) {$TA$};
						\node (e) at (2,2) {$TC$};
						\node (f) at (3.8,2) {$TB$};
						\node (g) at (0,4) {$A$};
						\node (h) at (0,2) {$C$};
						\node (i) at (5.7,0) {$C$};
						\path[->>] (a) edge [] node [right] {$e$} (i)
								(d) edge [] node [left] {$Te$} (e)
								(g) edge [] node [left] {$e$} (h);
						\path[right hook->] (b) edge [] node [below] {$m$} (c);
						\path[left hook->]	(i) edge [] node [below] {$m$} (c);
						\path[->] (d) edge [] node [above] {$\alpha$} (a)
								(e) edge [] node [below] {$Tm$} (f)
								(f) edge [] node [right] {$\beta$} (c)
								(e) edge [] node [right] {$\gamma$} (b)
								(a) edge [] node [below] {$f$} (c)
								(d) edge [] node [right] {$Tf$} (f)
								(g) edge [] node [above] {$\eta_A$} (d)
								(h) edge [] node [below] {$\eta_C$} (e);
					\end{tikzpicture}
					\end{center}
					from this, starting from $A$ at the top left corner and finishing at $B$ we have that $m\circ \gamma \circ \eta_C \circ e=m\circ e\circ \alpha\circ \eta_A$. 
					From that equation, using the fact that $\alpha\circ \eta_A=id_A$, since $A\in \Alg(\mathsf{T})$, and the fact that $m$ is mono and $e$ is epi, we have that 
					$\gamma\circ \eta_C=id_C$.
				\item We prove that $\gamma\circ \mu_C=\gamma\circ T\gamma$. In fact, we have the following commutative diagram:
					\begin{center}	
					\begin{tikzpicture}[>=stealth,shorten >=3pt,
						node distance=2.5cm,on grid,auto,initial text=,
						accepting/.style={thick,double distance=2.5pt}]
						\node (a) at (0,0) {$TC$};
						\node (b) at (3,0) {$C$};
						\node (c) at (6,0) {$B$};
						\node (d) at (3,1) {$TB$};
						\node (e) at (-1,3) {$TTC$};
						\node (f) at (1,3) {$TA$};
						\node (g) at (5,3) {$TB$};
						\node (h) at (7,3) {$C$};
						\node (i) at (3,3) {$TTB$};
						\node (j) at (3,5) {$TA$};
						\node (k) at (0,6) {$TTA$};
						\node (l) at (3,6) {$TTC$};
						\node (m) at (6,6) {$TC$};
						\path[->>] (f) edge [] node [right] {$Te$} (a)
								(k) edge [] node [left] {$TTe$} (e)
									edge [] node [above] {$TTe$} (l)
								(j) edge [] node [above] {$Te$} (m);
						\path[right hook->] (b) edge [] node [below] {$m$} (c);
						\path[left hook->]	(h) edge [] node [right] {$m$} (c);
						\path[->] (a) edge [] node [below] {$\gamma$} (b)
									edge [] node [above] {$Tm$} (d)
								(g) edge [] node [left] {$\beta$} (c)
								(m) edge [] node [left] {$Tm$} (g)
									edge [] node [right] {$\gamma$} (h)
								(d) edge [] node [above] {$\beta$} (c)
								(l) edge [] node [above] {$\mu_C$} (m)
								(k) edge [] node [above] {$\mu_A$} (j)
									edge [] node [right] {$T\alpha$} (f)
									edge [] node [right] {$TTf$} (i)
								(e) edge [] node [left] {$T\gamma$} (a)
								(f) edge [] node [above] {$\ Tf$} (d)
								(i) edge [] node [right] {$T\beta$} (d)
									edge [] node [above] {$\mu_B$} (g)
								(j) edge [] node [right] {$Tf$} (g);
					\end{tikzpicture}
					\end{center}
					Then by following the external arrows we get that $m\circ \gamma\circ \mu_C\circ TTe=m\circ \gamma\circ T\gamma\circ TTe$. Then, from that equation, since 
					$m$ is mono and $TTe$ is epi, we have that $\gamma\circ \mu_C=\gamma\circ T\gamma$.
			\end{enumerate}
		\item Put $\mathbf{B}=(B,\beta)$ and $\mathbf{C}=(C,\gamma)$. Then we have that $m\circ \gamma\circ Td\circ Te=m\circ d\circ \beta\circ Te$. From that equality, 
			since $m$ is mono and $Te$ is epi, we get $\gamma\circ Td=d\circ \beta$, i.e., $d\in \Alg(\mathsf{T})(\mathbf{B},\mathbf{C})$.
	\end{enumerate}
\end{proof}

\subsection{Details for Section \ref{secbirk}}

\subsubsection{Proof of Birkhoff's Theorem for $\mathsf{T}$--algebras}

The proof for Birkhoff's theorem for $\mathsf{T}$--algebras can be made by following the same ideas for standard proofs of Birkhoff's theorem, see, e.g., \cite{bys}. In our case 
we deal with equational $\mathsf{T}$--theories and we have a fixed the subcategory $\mathcal{D}_0$ of ``variables'', which is the main difference with respect to some other versions 
such as \cite{awodeyh,bana,barr}. We will derive Theorem \ref{birkthm} from the following basic facts and some facts from \cite{ahs,bana}.

\begin{lemma}\label{lemmafree}
	Let $\mathcal{D}$ be a category, $\mathscr{E}/\mathscr{M}$ a factorization system on $\mathcal{D}$, $\mathsf{T}=(T,\eta,\mu)$ a monad on $\mathcal{D}$ and 
	$\mathcal{D}_0$ a full subcategory of $\mathcal{D}$. Assume (B2). Let $\E=\{TX\overset{e_X}{\lepi} Q_X\}_{X\in \mathcal{D}_0}$ be an 
	equational $\mathsf{T}$--theory on $\mathcal{D}_0$. Then $\mathbf{Q_X}\in \md(\E)$ for every $X\in \mathcal{D}_0$.
\end{lemma}

\begin{proof}
	Let $Y\in \mathcal{D}_0$ and let $f\in \Alg(\mathsf{T}) (\mathbf{TY},\mathbf{Q_X})$. Then we have the following commutative diagram:
	\begin{center}	
		\begin{tikzpicture}[>=stealth,shorten >=3pt,
			node distance=2.5cm,on grid,auto,initial text=,
			accepting/.style={thick,double distance=2.5pt}]
			\node (b) at (4,2) {$Q_X$};
			\node (c) at (0,2) {$TY$};
			\node (d) at (2,2.4) {$Q_Y$};
			\node (e) at (6,2) {$TX$};
			\path[->>] (c) edge [] node [above] {$e_Y$} (d);
			\path[->>] (e) edge [] node [above] {$e_X$} (b);
			\path[->,dashed] (d) edge [] node [above] {$g'$} (b);
			\path[->,bend angle=10] (c) edge [bend right,looseness=0.7] 
						node [below] {$f$} (b);
			\path[->,bend angle=45,dashed] (c) edge [bend left,looseness=0.7] 
						node [above] {$g$} (e);			
		\end{tikzpicture}
	\end{center}
	where $g\in \Alg(\mathsf{T})(\mathbf{TY},\mathbf{TX})$ is obtained from $f$ and $e_X$ using assumption (B2) and $g'$ is obtained from the fact that $\E$ is an 
	equational $\mathsf{T}$--theory. Therefore, $f$ factors through $e_Y$ and hence $\mathbf{Q_X}\in \md(\E)$.
\end{proof}

\begin{proposition}\label{injth}
	Let $\mathcal{D}$ be a category, $\mathscr{E}/\mathscr{M}$ a factorization system on $\mathcal{D}$, $\mathsf{T}=(T,\eta,\mu)$ a monad on $\mathcal{D}$ and 
	$\mathcal{D}_0$ a full subcategory of $\mathcal{D}$. Assume (B1) and (B2). For $i=1,2$, let $\E_i=\{TX\overset{(e_i)_X}{\lepi} (Q_i)_X\}_{X\in \mathcal{D}_0}$ be an 
	equational $\mathsf{T}$--theory on $\mathcal{D}_0$. If $\E_1\neq \E_2$ then $\md (\E_1)\neq \md(\E_2)$.
\end{proposition}

\begin{proof} 
	As $\E_1\neq \E_2$, there exists $X\in \mathcal{D}_0$ such that $(e_1)_X\neq (e_2)_X$, i.e., there is no isomorphism 
	$\phi\in \Alg(\mathsf{T})(\mathbf{(Q_1)_X},\mathbf{(Q_2)_X})$ such that $\phi\circ (e_1)_X=(e_2)_X$. We have that $\mathbf{(Q_1)_X}\notin \md(\E_2)$ or 
	$\mathbf{(Q_2)_X}\notin \md(\E_1)$. In fact, if we assume by contradiction that $\mathbf{(Q_1)_X}\in \md(\E_2)$ and 
	$\mathbf{(Q_2)_X}\in \md(\E_1)$ then, from the fact that $\mathbf{(Q_1)_X}\in \md(\E_2)$, we get the commutative diagram:
	\begin{center}	
		\begin{tikzpicture}[>=stealth,shorten >=3pt,
				node distance=2.5cm,on grid,auto,initial text=,
				accepting/.style={thick,double distance=2.5pt}]
				\node (b) at (5,2) {$(Q_1)_X$};
				\node (c) at (0,2) {$TX$};
				\node (d) at (2.5,2) {$(Q_2)_X$};
				\path[->>] (c) edge [] node [above] {$(e_2)_X$} (d);
				\path[->,dashed] (d) edge [] node [above] {$g_{21}$} (b);
				\path[->>,bend angle=40] (c) edge [bend left,looseness=0.7] 
							node [above] {$(e_1)_X$} (b);
		\end{tikzpicture}
	\end{center}
	i.e., there exists $g_{21}\in \Alg(\mathsf{T})(\mathbf{(Q_2)_X},\mathbf{(Q_1)_X})$ such that $g_{21}\circ (e_2)_X=(e_1)_X$. Similarly, from the fact that 
	$\mathbf{(Q_2)_X}\in \md(\E_1)$, we get that there exists $g_{12}\in \Alg(\mathsf{T})(\mathbf{(Q_1)_X},\mathbf{(Q_2)_X})$ such that $g_{12}\circ (e_1)_X=(e_2)_X$. 
	Hence we have that:
	\[
		(e_2)_X=g_{12}\circ (e_1)_X=g_{12}\circ g_{21}\circ (e_2)_X
	\]
	which implies that $g_{12}\circ g_{21}=id_{(Q_2)_X}$ since $(e_2)_X$ is epi by (B1). Similarly, $g_{21}\circ g_{12}=id_{(Q_1)_X}$, which implies that 
	$g_{12}$ is an isomorphism such that $g_{12}\circ (e_1)_X=(e_2)_X$ which is a contradiction. Hence $\mathbf{(Q_1)_X}\notin \md(\E_2)$ or 
	$\mathbf{(Q_2)_X}\notin \md(\E_1)$ and, by the previous lemma, $\mathbf{(Q_i)_X}\in \md(\E_i)$, which implies that $\md (\E_1)\neq \md(\E_2)$.
\end{proof}

\noindent The next proposition shows that, under conditions (B1), (B2) and (B4), every class defined by an equational $\mathsf{T}$--theory is a variety of $\mathsf{T}$--algebras.

\begin{proposition}\label{hspclosureb}
	Let $\mathcal{D}$ be a complete category, $\mathsf{T}=(T,\eta,\mu)$ a monad on $\mathcal{D}$, $\mathscr{E}/\mathscr{M}$ a factorization system on $\mathcal{D}$ and 
	$\mathcal{D}_0$ a full subcategory of $\mathcal{D}$. Assume (B1), (B2) and (B4). Let $\E$ be an equational $\mathsf{T}$--theory on $\mathcal{D}_0$. Then 
	$\md(\E)$ is a variety of $\mathsf{T}$--algebras.
\end{proposition}

\begin{proof}
	$\md (\E)$ is nonempty by Lemma \ref{lemmafree}. Put $\E=\{TX\overset{e_X}{\lepi} Q_X\}_{X\in \mathcal{D}_0}$, then:
	\begin{enumerate}[i)]
		\item $\md(\E)$ is closed under $\mathscr{E}$--quotients: Let $\mathbf{A},\mathbf{B}\in \Alg(\mathsf{T})$ with $\mathbf{A}\in \md(\E)$ and let 
			$e\in \Alg(\mathsf{T})(\mathbf{A},\mathbf{B})\cap \mathscr{E}$. Let $f\in \Alg(\mathsf{T}) (\mathbf{TX},\mathbf{B})$ such that $X\in \mathcal{D}_0$, then we 
			have the following commutative diagram:
			\begin{center}	
				\begin{tikzpicture}[>=stealth,shorten >=3pt,
					node distance=2.5cm,on grid,auto,initial text=,
					accepting/.style={thick,double distance=2.5pt}]
					\node (b) at (4,2) {$A$};
					\node (c) at (0,2) {$TX$};
					\node (d) at (2,2) {$Q_X$};
					\node (e) at (6,2) {$B$};
					\path[->>] (c) edge [] node [above] {$e_X$} (d)
							(b) edge [] node [above] {$e$} (e);
					\path[->,dashed] (d) edge [] node [above] {$g_k$} (b);
					\path[->,bend angle=35] (c) edge [bend left,looseness=0.6] 
								node [above] {$k$} (b);
					\path[->,bend angle=50] (c) edge [bend left,looseness=0.7] 
								node [above] {$f$} (e);
				\end{tikzpicture}
			\end{center}
			where $k\in \Alg(\mathsf{T})(\mathbf{TX},\mathbf{A})$ was obtained from $f$ using (B2) and $g_k\in \Alg(\mathsf{T})(\mathbf{Q_X},\mathbf{A})$ from the fact that 
			$\mathbf{A}\in \md (\E)$. Therefore $f$ factors through $e_X$, i.e., $\mathbf{B}\in \md (\E)$.
		\item $\md(\E)$ is closed under $\mathscr{M}$--subalgebras: Let $\mathbf{A},\mathbf{B}\in \Alg(\mathsf{T})$ with $\mathbf{A}\in \md(\E)$ and let 
			$m\in \Alg(\mathsf{T})(\mathbf{B},\mathbf{A})\cap\mathscr{M}$. Let $f\in \Alg(\mathsf{T}) (\mathbf{TX},\mathbf{B})$ such that $X\in\mathcal{D}_0$, 
			then we have the following commutative diagram:
			\begin{center}	
				\begin{tikzpicture}[>=stealth,shorten >=3pt,
					node distance=2.5cm,on grid,auto,initial text=,
					accepting/.style={thick,double distance=2.5pt}]
					\node (b) at (4,2) {$B$};
					\node (c) at (0,2) {$TX$};
					\node (d) at (2,2.4) {$Q_X$};
					\node (e) at (6,2) {$A$};
					\path[->>] (c) edge [] node [above] {$e_X$} (d);
					\path[right hook->] (b) edge [] node [above] {$m$} (e);
					\path[->,dashed] (d) edge [] node [above] {$k$} (b);
					\path[->,bend angle=10] (c) edge [bend right,looseness=0.7] 
								node [below] {$f$} (b);
					\path[->,bend angle=25,dashed] (d) edge [bend left,looseness=0.7] 
								node [above] {$g_{m\circ f}$} (e);

				\end{tikzpicture}
			\end{center}
			where $g_{m\circ f}\in \Alg(\mathsf{T})(\mathbf{Q_X},\mathbf{A})$ was obtained from the fact that $\mathbf{A}\in \md (\E)$, and 
			$k\in \Alg(\mathsf{T})(\mathbf{Q_X},\mathbf{B})$ was obtained by the diagonal fill--in property of the factorization system $\mathscr{E}/\mathscr{M}$. 
			Since $m$ is mono, from $m\circ k\circ e_X=m\circ f$, we get $k\circ e_X=f$ which implies that $\mathbf{B}\in \md (\E)$. 
		\item $\md(\E)$ is closed under products: Let $\mathbf{A}_i\in \md(\E)$, $i\in I$, and let $\mathbf{A}=\prod_{i\in I}\mathbf{A}_i$ be their product in $\Alg(\mathsf{T})$ 
			with projections $\pi_i:A\to A_i$. Let $f\in \Alg(\mathsf{T}) (\mathbf{TX},\mathbf{A})$ such that $X\in \mathcal{D}_0$, then we have the following commutative diagram:
			\begin{center}	
				\begin{tikzpicture}[>=stealth,shorten >=3pt,
					node distance=2.5cm,on grid,auto,initial text=,
					accepting/.style={thick,double distance=2.5pt}]
					\node (b) at (4,2) {$A$};
					\node (c) at (0,2) {$TX$};
					\node (d) at (2,2.4) {$Q_X$};
					\node (e) at (6,2) {$A_i$};
					\path[->>] (c) edge [] node [above] {$e_X$} (d);
					\path[->] (b) edge [] node [above] {$\pi_i$} (e);
					\path[->,dashed] (d) edge [] node [above] {$g$} (b);
					\path[->,bend angle=10] (c) edge [bend right,looseness=0.7] 
								node [below] {$f$} (b);
					\path[->,bend angle=25,dashed] (d) edge [bend left,looseness=0.7] 
								node [above] {$g_{\pi_i\circ f}$} (e);
				\end{tikzpicture}
			\end{center}
			where $g_{\pi_i\circ f}\in \Alg(\mathsf{T})(\mathbf{Q_X},\mathbf{A}_i)$ was obtained from the fact that 
			$\mathbf{A}_i\in \md (\E)$, and $g\in \Alg(\mathsf{T})(\mathbf{Q_X},\mathbf{A})$ was obtained by the universal property of the product. Finally, we have that 
			$g\circ e_X=f$ since $\pi_i\circ g\circ e_X=\pi_i\circ f$ for every $i\in I$.
	\end{enumerate}
\end{proof}

\noindent Theorem \ref{birkthm} follows from \cite{bana} as follows. We have that the facts in \cite{bana} hold for any $(E,M)$--category as it is mentioned before 
\cite[Example 1]{bana} (see \cite{herrlich} for basic facts and examples about $(E,M)$--categories). Now, by the assumptions of Theorem \ref{birkthm} we have that 
the category $\Alg(\mathsf{T})$ is an $(E,M)$--category, by \cite[Corollary 15.21]{ahs}, and hence a class $K$ of $\mathsf{T}$--algebras, viewed as a full subcategory of 
$\Alg(\mathsf{T})$, is $E$--equational in $\Alg(\mathsf{T})$ if and only if $K$ is a variety. In this case, the inclusion functor $H:K\to \Alg(\mathsf{T})$ has a left adjoint 
$F:\Alg(\mathsf{T})\to K$ such that its unit $\eta: Id_{\Alg(\mathsf{T})}\Rightarrow HF$ is such that all its componets are in $E=\mathscr{E}$ and, by the assumptions, 
the equational $\mathsf{T}$--theory $\E=\{TX\overset{\eta_\mathbf{TX}}{\lepi} FHTX\}_{X\in \mathcal{D}_0}$ defines $K$ (see, e.g., \cite[Proposition 3 and Remark 1]{bana}). 
The fact that $\E$ is an equational $\mathsf{T}$--theory follows from naturality of $\eta$. Finally, by Proposition \ref{hspclosureb} every equational $\mathsf{T}$--theory 
defines a variety and the correspondence between equational $\mathsf{T}$--theories and varieties is bijective by Proposition \ref{injth} and uniqueness of left adjoint.

\subsubsection{Details for Example \ref{exjan11}}\label{aex1}

We prove that the notion of a coequational $\mathsf{B}$--theory coincides with the notion of a ``variety of languages'' given in \cite[Definition 35]{jan1}.

\begin{defi}[\text{\cite[Definition 35]{jan1}}]
	A {\it variety of languages} is an operator $\mathcal{V}$ on $\set$ such that for every $X\in \set$, $\mathcal{V}(X)\subseteq \set(X^*,2)$ and it satisfies the following:
	\begin{enumerate}[i)]
		\item for every $L\in \mathcal{V}(X)$ we have that $\coeq(X^*/\eq \langle L\rangle)\subseteq \mathcal{V}(X)$;
		\item if $\coeq(X^*/C_i)\subseteq \mathcal{V}(X)$, where $C_i$ a monoid congruence of $X^*$, $i\in I$, then we have that $\coeq(X^*/\bigcap_{i\in I}C_i)\subseteq \mathcal{V}(X)$;
		\item for every $Y\in \set$, if $L\in \mathcal{V}(Y)$ and $\eta:Y^*\to Y^*/\eq \langle L\rangle$ denotes the quotient morphism, then for each monoid morphism 
			$\varphi: X^*\to Y^*$ we have $\coeq(X^*/\ker(\eta\circ \varphi))\subseteq \mathcal{V}(X)$.
	\end{enumerate}
\end{defi}

\begin{defi}
	A coequational $\mathsf{B}$--theory is an operator $\mathcal{L}$ on $\set$ such that for every $X\in \set$ we have that:
	\begin{enumerate}[i)]
		\item $\mathcal{L}(X)\in \CABA$ and it is a subalgebra of $\set(A^*,2)$.
		\item $\mathcal{L}(X)$ is closed under left and right derivatives. That is, if $L\in \mathcal{L}(X)$ and $x\in X$ then $_xL,L_x\in \mathcal{L}(X)$, where $_xL(w)=L(wx)$ 
			and $L_x(w)=L(xw)$, $w\in X^*$.
		\item $\mathcal{L}$ is closed under morphic preimages. That is, for every $Y\in \set$, homomorphism of monoids $h:Y^*\to X^*$ and $L\in \mathcal{L}(X)$, we have that 
			$L\circ h\in \mathcal{L}(Y)$.
	\end{enumerate}
\end{defi}

\noindent The equivalence of a coequational $\mathsf{B}$--theory with the operator $\Lan$ defined above follows from Example \ref{exua} (see also Example \ref{exvv1}).

\noindent We show that the two notions above coincide. We prove that every $\mathcal{V}$ above satisfies the conditions of the $\mathcal{L}$ above and vice versa. 

\begin{lemma}\label{lj1}
	For every $X\in \set$ and $L\in \set(X^*,2)$ we have that $\coeq (X^*/\eq \langle L\rangle)=\langle\mkern-3mu\langle L\rangle\mkern-3mu\rangle$ where 
	$\langle\mkern-3mu\langle L\rangle\mkern-3mu\rangle$ is the $\mathsf{B}$--coalgebra generated by $L$. 
\end{lemma}

\begin{proof}
	By \cite[Corollary 8]{jan1} we have that the monoid $X^*/\eq \langle L\rangle$ is the syntactic monoid of $L$. The universal property 
	of the syntactic monoid of $L$ is, by duality, the property that $\coeq (X^*/\eq \langle L\rangle)=\langle\mkern-3mu\langle L\rangle\mkern-3mu\rangle$. 
	This property of $\langle\mkern-3mu\langle L\rangle\mkern-3mu\rangle$ being the dual of the syntactic monoid of $L$ was also mentioned in \cite[Section 6]{gehrke0}.
\end{proof}

\begin{lemma}\label{lj2}
	Let $\mathcal{V}$ be a variety of languages and $X\in \set$, then 
	\[
		\mathcal{V}(X)=\coeq\left(X^*/\bigcap_{L\in \mathcal{V}(X)}\eq \langle L\rangle\right).
	\]
\end{lemma}

\begin{proof}
	($\supseteq$): Follows from properties i) and ii) of $\mathcal{V}$ being a variety of languages.

	\noindent ($\subseteq$): Consider the canonical epimorphism of monoids $e_{_{L'}}:X^*/\bigcap_{L\in \mathcal{V}(X)}\eq \langle L\rangle\to X^*/\eq \langle L'\rangle$, 
	$L'\in \mathcal{V}(X)$. Then, by duality, i.e., applying $\coeq$, gives us the monomorphism 
	$m_{_{L'}}:\langle\mkern-3mu\langle L'\rangle\mkern-3mu\rangle\to \coeq\left(X^*/\bigcap_{L\in \mathcal{V}(X)}\eq \langle L\rangle\right)$ which implies that 
	$L'\in \coeq\left(X^*/\bigcap_{L\in \mathcal{V}(X)}\eq \langle L\rangle\right)$ since $L'\in \langle\mkern-3mu\langle L'\rangle\mkern-3mu\rangle$.	
\end{proof}

\begin{lemma}\label{lj3}
	For every $X\in \set$ and every $L\in \set(X^*,2)$ we have that $L=\bigcup_{w\in L}w/\eq\langle L\rangle$, where $w/\eq\langle L\rangle$ 
	denotes the equivalence class of $w$ in $X^*/\eq \langle L\rangle$.
\end{lemma}

\begin{proof}
	($\subseteq$): obvious.

	\noindent ($\supseteq$):  Let $u\in \bigcup_{w\in L}w/\eq\langle L\rangle$, then there exists $v\in L$ such that $(u,v)\in \eq \langle L\rangle$. In particular, 
	$L_u=L_v$. Now, using the fact that $v\in L$ we get the following implications:
	\[
		v\in L\Rightarrow \epsilon \in L_v=L_u \Rightarrow \epsilon \in L_u
	\]	
	i.e., $u\in L$.
\end{proof}

\noindent The previous lemma basically says that the syntactic monoid of $L$ recognizes $L$.

\noindent Lemma \ref{lj2} says that $\mathcal{V}(X)\in \CABA$ for every $X\in \set$, since $\coeq(X^*/C)\cong \mathcal{P}(X^*/C)$ for every monoid congruence $C$ of $X^*$ 
\cite[Proposition 15]{jan1}. Lemma \ref{lj1} together with property i) of $\mathcal{V}$ being a variety of languages imply that 
$\mathcal{V}(X)$ is closed under left and right derivatives. That is, every variety of languages $\mathcal{V}$ satisfies properties i) and ii) of a coequational $\mathsf{B}$--theory. Now 
we show that $\mathcal{V}$ also satisfies property iii) of a coequational $\mathsf{B}$--theory.

\begin{lemma}
	Let $\mathcal{V}$ be a variety of languages. Then for every $X,Y\in \set$, homomorphism of monoids $h:X^*\to Y^*$ and $L\in \mathcal{V}(Y)$ we have that 
	$L\circ h\in \mathcal{V}(X)$.
\end{lemma}

\begin{proof}
	By property iii) of $\mathcal{V}$ being a variety of languages we have that $\coeq(X^*/\ker(\eta\circ h))\subseteq \mathcal{V}(X)$. We will show that 
	$L\circ h\in \coeq(X^*/\ker(\eta\circ h))\subseteq \mathcal{V}(X)$. In fact, \\

	{\it Claim:} $L\circ h=\bigcup \{ w/\ker (\eta\circ h) \mid w\in X^*\text{ s.t. }h(w)\in L\}$.

	Let $v\in X^*$, then:

	($\subseteq$): $v\in L\circ h\Rightarrow h(v)\in L\Rightarrow v\in \bigcup \{ w/\ker (\eta\circ h) \mid w\in X^*\text{ s.t. }h(w)\in L\}$.

	($\supseteq$): Assume $v\in \bigcup \{ w/\ker (\eta\circ h) \mid w\in X^*\text{ s.t. }h(w)\in L\}$, i.e., there exists $u\in X^*$ with $h(u)\in L$ such that 
	$(v,u)\in \ker(\eta\circ h)$. Now, we have
	\[
		(v,u)\in \ker(\eta\circ h)\Rightarrow (h(v),h(u))\in \ker(\eta)=\eq \langle L\rangle \Rightarrow h(v)\in L
	\]
	where the last implication follows from Lemma \ref{lj3} since $h(u)\in L$. Finally, from $h(u)\in L$ we get $u\in L\circ h$. This finishes the proof of the claim.\\

	\noindent From the claim we have that $L\circ h\in \coeq(X^*/\ker(\eta\circ h))\subseteq \mathcal{V}(X)$.
\end{proof}

\noindent Until now we proved the following.

\begin{prop}
	Let $\mathcal{V}$ be a variety of languages. Then $\mathcal{V}$ is a coequational $\mathsf{B}$--theory.
\end{prop}

Now we prove.

\begin{prop}
	Let $\mathcal{L}$ be a coequational $\mathsf{B}$--theory. Then $\mathcal{L}$ is a variety of languages.
\end{prop}

\begin{proof}
	We have to prove that $\mathcal{L}$ satisfies properties i), ii) and iii) that define a variety of languages. In fact, let $X\in \set$, then:
	\begin{enumerate}[i)]
		\item Properties i) and ii) of $\mathcal{L}$ being a coequational $\mathsf{B}$--theory say that $\mathcal{L}(X)$ is a $\mathsf{B}$--subcoalgebra of 
			$\set(X^*,2)$. In particular, for every $L\in \mathcal{L}(X)$ we have $\coeq(X^*/\eq \langle L\rangle)=\langle\mkern-3mu\langle L\rangle\mkern-3mu\rangle
			\subseteq \mathcal{L}(X)$.
		\item To prove property ii) we show that for a monoid congruence $C_i$ of $X^*$, $i\in I$, the $\mathsf{B}$--coalgebra $\coeq(X^*/\bigcap_{i\in I}C_i)$ 
			is the $\mathsf{B}$--subcoalgebra of $\set(X^*,2)$ generated by the family $\{\coeq(X^*/C_i)\}_{i\in I}$. We show this by duality, i.e., in the category of monoids. 
			We have the following setting:
			\begin{center}	
				\begin{tikzpicture}[>=stealth,shorten >=3pt,
						node distance=2.5cm,on grid,auto,initial text=,
						accepting/.style={thick,double distance=2.5pt}]
						\node (a) at (3,2) {$X^*$};
						\node (b) at (6,0) {$X^*/\bigcap_{i\in I}C_i$};
						\node (c) at (3,0) {$P$};
						\node (d) at (0,0) {$X^*/C_j$};
						\path[right hook->] (b) edge [] node [below] {$m_\eta$} (c);
						\path[->>] (a) edge [] node [above] {$\eta_j$} (d)
									edge [] node [above] {$e_\eta$} (b);
						\path[->] (a) edge [] node [left] {$\eta$} (c)
								(c) edge [] node [below] {$\pi_j$} (d);
				\end{tikzpicture}
			\end{center}
			where:
			\begin{enumerate}[-]
				\item $\eta_j:X^*\to X^*/C_j$ is the canonical homomorphism, $j\in I$, 
				\item $P$ is the product $P=\prod_{i\in I}X^*/C_i$ with projections $\pi_j:P\to X^*/C_j$, $j\in I$, 
				\item $\eta$ is obtained from $\eta_j$, $j\in I$, by the universal property of $P$, and
				\item $\eta=m_\eta\circ e_\eta$ is the factorization of $\eta$, i.e., $\ker(\eta)=\bigcap_{i\in I}C_i$.
			\end{enumerate}
			Now we prove, by duality, that the $\mathsf{B}$--coalgebra $\coeq(X^*/\bigcap_{i\in I}C_i)$ is the least $\mathsf{B}$--subcoalgebra of $\set(X^*,2)$ containing 
			each of $\coeq(X^*/C_i)$. Let $e:X^*\to X^*/C$ be an epimorphism of monoids such that each $\eta_j$ factors through $e$, $j\in I$. That is, there exists 
			$g_j:X^*/C\to X^*/C_j$ such that $\eta_j=g_j\circ e$, $j\in I$. Therefore, $C\subseteq C_j$, $j\in I$, and hence $C\subseteq \bigcap_{i\in I}C_i$, which means that 
			there exists $g:X^*/C\to X^*/\bigcap_{i\in I}C_i$ such that $e_\eta=g\circ e$.

			Now, $\mathcal{L}$ satisfying property ii) of a variety of languages follows from the observation above. In fact, if $\mathcal{L}(X)$ contains 
			$\coeq(X^*/C_i)$, $i\in I$, then, by using the fact that $\mathcal{L}(X)$ is a $\mathsf{B}$--subcoalgebra of $\set(X^*,2)$, it contains the least 
			$\mathsf{B}$--subcoalgebra of $\set(X^*,2)$ containing each of $\coeq(X^*/C_i)$, $i\in I$, which is $\coeq(X^*/\bigcap_{i\in I}C_i)$.
		\item Let $Y\in \set$, $L\in \mathcal{L}(Y)$ and $\eta:Y^*\to Y^*/\eq \langle L\rangle$ be the quotient morphism. Let $\varphi: X^*\to Y^*$ be a monoid morphism. 
			We have to show that $\coeq(X^*/\ker(\eta\circ \varphi))\subseteq \mathcal{L}(X)$. In fact, let $L'\in \coeq(X^*/\ker(\eta\circ \varphi))$, i.e., $L'$ is of the form 
			$L'=\bigcup_{w\in W}w/\ker(\eta\circ \varphi)$ for some $W\subseteq X^*$. Define $L''$ as 
			$L''=\bigcup_{w\in W}\varphi(w)/\ker(\eta)=\bigcup_{w\in W}\varphi(w)/\eq \langle L \rangle$. Then we have that $L''\in \coeq (Y^*/\eq\langle L\rangle)$ 
			which by i) implies that $L''\in \mathcal{L}(Y)$, since $ \coeq (Y^*/\eq\langle L\rangle)\subseteq \mathcal{L}(Y)$. Since $\mathcal{L}$ is a coequational 
			$\mathsf{B}$--theory then $L''\circ \varphi\in \mathcal{L}(X)$. To finish the proof we prove the following:
		
			{\it Claim:} $L'=L''\circ \varphi$.

			Let $u\in X^*$, then:

			($\subseteq$): Assume that $u\in L'$. Then there exists $w\in W$ such that $(u,w)\in \ker(\eta\circ \varphi)$. This implies that 
			$(\varphi(u),\varphi(w))\in \ker(\eta)=\eq\langle L\rangle$ with $w\in W$, i.e., $\varphi(u)\in L''$ which means that $u\in L''\circ \varphi$.
			
			($\supseteq$): Assume that $u\in L''\circ \varphi$, i.e., $\varphi(u)\in L''$. Then there exists $w\in W$ such that $(\varphi(u),\varphi(w))\in \ker(\eta)$. This implies that 
			$(u,w)\in \ker(\eta\circ \varphi)$ with $w\in W$, i.e., $u\in L'$.
	\end{enumerate}

\end{proof}

\subsubsection{Details for Example \ref{exuoa}}\label{auoa}

The duality between $\poset$ and $\ACDL$ is given by the hom-set functors $\poset(\underline{\ \ },\mathbf{2}_c):\poset \to \ACDL$ and 
$\ACDL(\underline{\ \ },\mathbf{2}_c):\ACDL\to \poset$, where $\mathbf{2}_c$ is the two--element chain `schizophrenic' object in $\poset$ and in $\ACDL$. 
Note that for any $\mathbf{P}\in \poset$ the object $\poset(\mathbf{P},\mathbf{2}_c)$ is the set of downsets of $\mathbf{P}$ with the inclusion order and for any 
$\mathbf{A}\in \ACDL$ the object $\ACDL(\mathbf{A},\mathbf{2}_c)$ is the set of all completely join--prime elements of $\mathbf{A}$ with the order inherited from $\mathbf{A}$. 
Remember that an element $a$ of $\mathbf{A}$ is {\it completely join--prime} if $a\leq \bigvee S$ implies $a\leq s$ for some $s\in S$.

\noindent A variety of ordered algebras in $K$ is defined by an equational $\mathsf{T}_K$--theory $\{T_KX\overset{e_X}{\lepi}Q_X\}_{X\in \set}$ which by duality gives us 
the coequational $\mathsf{B}$--theory $\{\poset(e_X,\mathbf{2}_c)\}_{X\in \set}$ which is equivalently defined by the image of every embedding $\poset(e_X,\mathbf{2}_c)$, 
i.e., we define $\Lan(X):=\Imag(\poset(e_X,\mathbf{2}_c))$. Closure of $\Lan(X)$ under derivatives with respect to the type $\T$ follows from the fact that each 
morphism $e_X$ in an equational $\mathsf{T}_{\T}$--theory is a homomorphism of ordered algebras. In fact, similar to Example \ref{exua} b), we have that for every 
for every $g \in \T$ of arity $n_g$, every $1\leq i\leq n_g$, every $t_j\in T_KX$, $1\leq j< n_g$, $t\in T_KX$ and $f\in \poset(Q_X,\mathbf{2}_c)$ we have 
$(f\circ e_X)^{(i)}_{(g,t_1,\ldots ,t_{n_g-1})}=f^{(i)}_{(g,e_X(t_1),\ldots ,e_X(t_{n_g-1}))}\circ e_X$, where the function $f^{(i)}_{(g,e_X(t_1),\ldots ,e_X(t_{n_g-1}))}\in \set(Q_X,2)$ is 
defined for every $q\in Q_X$ as $f^{(i)}_{(g,e_X(t_1),\ldots ,e_X(t_{n_g-1}))}(q)=f(g(e_X(t_1),\ldots,e_X(t_{i-1}),q,e_X(t_{i}),\ldots,e_X(t_{n_g-1})))$. We only need to prove that 
$f^{(i)}_{(g,e_X(t_1),\ldots ,e_X(t_{n_g-1}))}\in \poset(Q_X,\mathbf{2}_c)$. In fact, for any $p\leq q$ in $Q_X$ we have that 
$g(e_X(t_i),\ldots, p,\ldots , e_X(t_{n_g-1}))\leq g(e_X(t_i),\ldots, q,\ldots , e_X(t_{n_g-1}))$, where $u$ and $v$ are in the $i$--th position, which implies that 
\begin{align*}
	f^{(i)}_{(g,e_X(t_1),\ldots ,e_X(t_{n_g-1}))}(p)&=f(g(e_X(t_1),\ldots,p,\ldots,e_X(t_{n_g-1})))\\
				&\leq f(g(e_X(t_1),\ldots,q,\ldots,e_X(t_{n_g-1})))\\
				&=	f^{(i)}_{(g,e_X(t_1),\ldots ,e_X(t_{n_g-1}))}(q).
\end{align*}
since $f\in \poset(Q_X,\mathbf{2}_c)$. Therefore, $\Lan(X)$ is closed under derivatives with respect to the type $\T$.\\
	
\noindent Conversely, any $S\in \ACDL$ closed under derivatives with respect to the type $\T$ such that $S$ is a subalgebra of $\poset(T_KX,\mathbf{2}_c)=\set(T_KX,2)\in \ACDL$ will define, 
by duality, the surjective function  $e_S:T_KX\to \ACDL(S,\mathbf{2}_c)$ such that $e_S(w)(L)=L(w)$, $w\in T_KX$ and $L\in S$, which is a morphism in 
$\poset(T_KX,\ACDL(S,\mathbf{2}_c))$. 
We only need to show that for every $g\in \T$, $t_j\in T_KX$, $1\leq j<n_g$ and $u,v\in T_KX$ the inequality $e_S(u)\leq e_S(v)$ implies 
that $e_S(g(t_1,\ldots, u,\ldots, t_{n_g-1}))\leq e_S(g(t_1,\ldots, v,\ldots, t_{n_g-1}))$ (see \cite[1.3. Proposition]{bloom}). In fact, assume that $e_S(u)\leq e_S(v)$, i.e., for every $L\in S$ we have that $L(u)\leq L(v)$. 
Now, assume by contradiction that $e_S(g(t_1,\ldots, u,\ldots, t_{n_g-1}))\not\leq e_S(g(t_1,\ldots, v,\ldots, t_{n_g-1}))$, i.e., there exists $L'\in S$ such that 
$L'(g(t_1,\ldots, u,\ldots, t_{n_g-1}))=1$ and $L'(g(t_1,\ldots, v,\ldots, t_{n_g-1}))=0$, i.e., $L'^{(i)}_{(g,t_1,\ldots,t_{n_g-1})}(u)=1$ and $L'^{(i)}_{(g,t_1,\ldots,t_{n_g-1})}(v)=0$ 
with $L'^{(i)}_{(g,t_1,\ldots,t_{n_g-1})}\in S$ by closure under derivatives with respect to the type $\T$, which 
contradics the fact that $e_S(u)\leq e_S(v)$. Therefore $e_S$ is a $\mathsf{T}_K$--algebra morphism in $\mathscr{E}$.

\subsubsection{Details for Example \ref{exvv4}}\label{avec}

The duality between $\fvec$ and $\tvec$ is given by the hom-set functors $\fvec(\underline{\ \ },\mathbb{K}):\fvec \to \tvec$ and 
$\tvec(\underline{\ \ },\mathbb{K}):\tvec\to \fvec$. 

\noindent A variety of $\mathbb{K}$--algebras is defined by an equational $\mathsf{T}$--theory $\{\mathsf{V}(X^*)\overset{e_X}{\lepi}Q_X\}_{X\in \mathcal{D}_0}$ which by duality gives us 
the coequational $\mathsf{B}$--theory $\{\fvec(e_X,\mathbb{K})\}_{X\in \mathcal{D}_0}$ which is equivalently defined by the image of every monomorphism $\fvec(e_X,\mathbb{K})$, 
i.e., we define $\Lan(X):=\Imag(\fvec(e_X,\mathbb{K}))$. Closure of $\Lan(X)$ under left and right derivatives follows from the fact that each morphism $e_X$ in an equational 
$\mathsf{T}$--theory is a homomorphism of $\mathbb{K}$--algebras. In fact, for every $v,w\in \mathsf{V}(X^*)$ and $f\in \fvec(Q_X,\mathbb{K})$ we have that
\[
	(f\circ e_X)_v(w)=(f\circ e_X)(vw)=f(e_X(v)\cdot e_X(w))=(f_{e_X(v)} \circ e_X)(w)
\]
where the function $f_{e_X(v)}\in \set(Q_X,\mathbb{K})$ is defined as $f_{e_X(v)}(q)=f(e_X(v)\cdot q)$, where $\cdot$ is the product operation in $\mathbf{Q_X}$, $q\in Q_X$. Note that 
$f_{e_X(v)}\in \fvec(Q_X,\mathbb{K})$ since for any $k\in \mathbb{K}$ and $p,q\in Q_X$ we have that 
\[
	f_{e_X(v)}(kp+q)=f(e_X(v)\cdot (kp+q))= kf(e_X(v)\cdot p)+f(e_X(v)\cdot q)=kf_{e_X(v)}(p)+f_{e_X(v)}(q)
\]
since $f\in \fvec(Q_X,\mathbb{K})$. Therefore, $(f\circ e_X)_x=f_x\circ e_X \in \Lan(X)$, i.e., $\Lan(X)$ is closed under right derivatives. Closure under left derivatives is proved in a 
similar way.\\
	
\noindent Conversely, any $S\in \tvec$ closed under left and right derivatives such that $S$ is a subspace of $\fvec(\mathsf{V}(X^*),\mathbb{K})\in \tvec$ will define, 
by duality, the surjective function  $e_S:\mathsf{V}(X^*)\to \tvec(S,\mathbb{K})$ such that $e_S(w)(L)=L(w)$, $w\in \mathsf{V}(X^*)$ and $L\in S$, which is a morphism in 
$\fvec(\mathsf{V}(X^*),\tvec(S,\mathbb{K}))$. We only need to show that for every $u,v,w\in \mathsf{V}(X^*)$ the equality $e_S(u)=e_S(v)$ implies 
that $e_S(wu)=e_S(wv)$ and $e_S(uw)=e_S(vw)$. In fact, assume that $e_S(u)=e_S(v)$, i.e., for every $L\in S$ we have that $L(u)=L(v)$. Now, assume by contradiction that 
$e_S(wu)\neq e_S(wv)$, i.e., there exists $L'\in S$ such that $L'(wu)\neq L'(wv)$, i.e., $L'_w(u)\neq L'_w(v)$ with $L'_w\in S$ by closure under right derivatives, which is 
a contradiction. The equality $e_S(uw)=e_S(vw)$ is proved in a similar way by using closure under left derivatives. Therefore $e_S$ is a $\mathsf{T}$--algebra morphism in $\mathscr{E}$.

\subsubsection{Details for Example \ref{exvv5}}\label{ajsl}

Define the monad $\mathsf{T}=(T,\eta,\mu)$ on $\JSL$ as $T(X,\lor)=(\mathcal{P}_f(X^*)/\theta,\cup_\theta)$ where $\theta$ is the least equivalence relation on $\mathcal{P}_f(X^*)$ 
such that:
\begin{enumerate}[i)]
	\item for every $x,y\in X$ $\{x\lor y\}\theta \{x,y\}$,
	\item for every $A,B,C,D\in \mathcal{P}_f(X^*)$, $A\theta B$ and $C\theta D$ imply $AC\theta BD$, and 
	\item for every $A,B,C,D\in \mathcal{P}_f(X^*)$, $A\theta B$ and $C\theta D$ imply $A\cup C\theta B\cup D$.
\end{enumerate}
and $\cup_\theta$ is defined as $A/\theta \cup_\theta B/\theta=(A\cup B)/\theta$ which is well--defined by property iii). We should use a notation like $\theta_{(X,\lor)}$ for the relation 
defined above, but we will denote it by $\theta$ for simplicity. It will be clear from the context to which $\theta$ we are refering to in each case. 
If $h\in \JSL((X,\lor),(Y,\lor))$ then $Th$ is defined as 
\[
	(Th)(\{w_1.\ldots,w_n\}/\theta)=\{h^*(w_1),\ldots,h^*(w_n)\}/\theta.
\]
The unit of the monad is defined as $\eta_{(X,\lor)}(x)=\{x\}/\theta$ and the multiplication as:
\[
	\mu_{(X,\lor)}(\{W_1,\ldots W_n\}/\theta)=\left( \bigcup_{i=1}^n \left(\prod_{j=1}^{m_i} W_j^{(i)}\right)\right)\Big/\theta
\]
where each $W_i\in (\mathcal{P}_f(X^*))^*$ is such that $W_i=W_1^{(i)}\cdots W_{m_i}^{(i)}$ and $W_j^{(i)}\in \mathcal{P}_f(X^*)$, $1\leq i\leq n$ and $1\leq j\leq m_i$.\\

\noindent We have that $\Alg(\mathsf{T})$ is the category of idempotent semirings. 

\begin{lemma}
	Consider the object $(\mathcal{P}_f(X),\cup)\in \JSL$, then $T(\mathcal{P}_f(X),\cup)$ is isomorphic to $(\mathcal{P}_f(X^*),\cup)$ in $\JSL$.
\end{lemma}

\begin{proof}
	By definition we have that 
	\[
		T(\mathcal{P}_f(X),\cup)=\left(\mathcal{P}_f\left(\mathcal{P}_f(X)^*\right)/\theta, \cup_\theta \right)
	\]
	Now, every element in $\mathcal{P}_f(X)$ is of the form $\{x_1,\ldots ,x_n\}=\{x_1\}\cup \cdots \cup \{x_n\}$, which by property i) and iii) of the definition of $\theta$ we have that:
	\[
		\{\{x_1,\ldots ,x_n\}\}\theta \{\{x_1\},\ldots ,\{x_n\}\}
	\]
	Therefore, by using the defining properties of $\theta$ we have that every element in $\mathcal{P}_f\left(\mathcal{P}_f(X)^*\right)$ is equivalent to a unique 
	element of the form:
	\[
		\left\{\left\{x_1^{(1)}\right\}\cdots \left\{x_{n_1}^{(1)}\right\},\ldots ,\left\{x_1^{(m)}\right\}\cdots \left\{x_{n_m}^{(m)}\right\}\right\}
	\]
	where uniqueness follows since $(\mathcal{P}_f(X),\cup)$ is the free join semilattice. Hence, the join semilattice homomorphism 
	$\varphi: (\mathcal{P}_f(X^*),\cup)\to T(\mathcal{P}_f(X),\cup)$ given by:
	\[
		\varphi (\{x_1^{(1)}\cdots x_{n_1}^{(1)},\ldots ,x_1^{(m)}\cdots x_{n_m}^{(m)} \})=\left\{\left\{x_1^{(1)}\right\}\cdots \left\{x_{n_1}^{(1)}\right\},\ldots ,
									\left\{x_1^{(m)}\right\}\cdots \left\{x_{n_m}^{(m)}\right\}\right\}\Big/ \theta
	\]
	is an isomorphism in $\JSL$.
\end{proof}

\noindent We considered $\mathcal{D}_0=\{(\mathcal{P}_f(X),\cup)\mid X\in \set\}$. As every semiring is an $\mathscr{E}$--quotient of $(\mathcal{P}_f(X^*),\cup)$, by the previous Lemma 
we have that condition (B3) is satisfied.\\

\noindent Now, in the definition of the operator $\Lan$ we should formally have that $\Lan(X)$ is a subspace of $\JSL(\mathcal{P}_f(X^*),\mathbf{2})$ but for simplicity we work with 
$\set(X^*,2)$ which is isomorphic to $\JSL(\mathcal{P}_f(X^*),\mathbf{2})$ in $\TBJSL$ under the correspondence $f\mapsto f\circ \eta_{X^*}$ and $L\mapsto L^\sharp$, 
$f\in \JSL(\mathcal{P}_f(X^*),\mathbf{2})$ and $L\in \set(X^*,2)$, where $\eta_{X^*}$ and $L^\sharp$ are defined as 
$\eta_{X^*}(w)=\{w\}$ and $L^\sharp (\{w_1,\ldots,w_n\})=\bigvee_{i=1}^n L(w_i)$. \\

\noindent The duality between $\JSL$ and $\TBJSL$ is given by the hom-set functors $\JSL(\underline{\ \ },\mathbf{2}):\JSL \to \TBJSL$ and 
$\TBJSL(\underline{\ \ },\mathbf{2}):\TBJSL\to \JSL$, where $\mathbf{2}$ is the two--element join semilattice.\\

\noindent A variety of idempotent semirings is defined by an equational $\mathsf{T}$--theory $\{\mathcal{P}_f(X^*)\overset{e_X}{\lepi}Q_X\}_{X\in \mathcal{D}_0}$ which by duality gives us 
the coequational $\mathsf{B}$--theory $\{\JSL(e_X,\mathbf{2})\}_{X\in \mathcal{D}_0}$ which is equivalently defined by the image of every monomorphism $\JSL(e_X,\mathbf{2})$, 
i.e., we define $\Lan(X):=\Imag(\JSL(e_X,\mathbf{2}))$. Closure of $\Lan(X)$ under left and right derivatives follows from the fact that each morphism $e_X$ in an equational 
$\mathsf{T}$--theory is a homomorphism of idempotent semirings. In fact, for every $v,w\in X^*$ and $f\in \JSL(Q_X,\mathbf{2})$ we have that
\[
	(f\circ e_X\circ \eta_{X^*})_v(w)=(f\circ e_X\circ \eta_{X^*})(vw)=(f_{e_X(\{v\})}\circ e_X\circ \eta_{X^*})(w)
\]
where the function $f_{e_X(\{v\})}\in \set(Q_X,2)$ is defined as $f_{e_X(\{v\})}(q)=f(e_X(\{v\})\cdot q)$, where $\cdot$ is the product operation in $\mathbf{Q_X}$, $q\in Q_X$. Note that 
$f_{e_X(\{v\})}\in \JSL(Q_X,\mathbf{2})$ since for any $p,q\in Q_X$ we have that 
\[
	f_{e_X(\{v\})}(p\lor q)=f(e_X(\{v\})\cdot (p\lor q))= f(e_X(\{v\})\cdot p)\lor f(e_X(\{v\})\cdot q)=f_{e_X(\{v\})}(p)\lor f_{e_X(\{v\})}(q)
\]
since $f\in \JSL(Q_X,\mathbf{2})$. Therefore, $(f\circ e_X\circ \eta_{X^*})_x=f_x\circ e_X\circ \eta_{X^*} \in \Lan(X)$, i.e., $\Lan(X)$ is closed under right derivatives. Closure under left derivatives is proved in a similar way.\\
	
\noindent Conversely, any $S\in \TBJSL$ closed under left and right derivatives such that $S$ is a subspace of $\JSL(\mathcal{P}_f(X^*),\mathbf{2})\in \TBJSL$ will define, 
by duality, the surjective function  $e_S:\mathcal{P}_f(X^*)\to \TBJSL(S,\mathbf{2})$ such that $e_S(\{w\})(L)=L(\{w\})$, $w\in X^*$ and $L\in S$, which is a morphism in 
$\JSL(\mathcal{P}_f(X^*),\TBJSL(S,\mathbf{2}))$. We only need to show that for every $w\in X^*$ and $U,V\in \mathcal{P}_f(X^*)$ the 
equality $e_S(U)=e_S(V)$ implies that $e_S(\{w\}U)=e_S(\{w\}V)$ and $e_S(U\{w\})=e_S(V\{w\})$. In fact, assume that $e_S(U)=e_S(V)$, i.e., for every $L\in S$ we have that 
$L(U)=L(V)$. Now, assume by contradiction that $e_S(\{w\}U)\neq e_S(\{w\}V)$, i.e., there exists $L'\in S$ such that $L'(\{w\}U)\neq L'(\{w\}V)$, i.e., 
$(L'\circ \eta_{X^*})_w(u)\neq (L'\circ \eta_{X^*})_w(v)$ with $(L'\circ \eta_{X^*})_w\in S$ by closure under right derivatives, which is 
a contradiction. The equality $e_S(U\{w\})=e_S(V\{w\})$ is proved in a similar way by using closure under left derivatives. Therefore $e_S$ is a $\mathsf{T}$ algebra morphism in 
$\mathscr{E}$.

\subsection{Details for Section \ref{secreit}}

\subsubsection{Proof of Birkhoff's Theorem for finite $\mathsf{T}$--algebras}

In this subsection, we provide a proof of Theorem \ref{reitthm}. 
We start by proving that models of pseudoequational $\mathsf{T}$--theories are pseudovarieties of $\mathsf{T}$--algebras.

\begin{proposition}\label{hspclosurer}
	Let $\mathcal{D}$ be a complete concrete category such that its forgetful functor preserves epis, monos and products, $\mathsf{T}$ a monad on $\mathcal{D}$, $\mathcal{D}_0$ 
	a full subcategory of $\mathcal{D}_0$ and $\mathscr{E}/\mathscr{M}$ a factorization system on $\mathcal{D}$. Assume (B$_f$1), (B$_f$2) and (B$_f$4). Let $\Ps$ be 
	a pseudoequational $\mathsf{T}$--theory on $\mathcal{D}_0$. Then $\md_f(\Ps)$ is a pseudovariety of $\mathsf{T}$--algebras.
\end{proposition}

\begin{proof}
	Clearly $\md_f (\Ps)$ is non empty since $1=(1,!_{T1}\colon T1\to 1)\in \md_f (\Ps)$, where $1$ is the terminal object in $\mathcal{D}$, which is finite since the forgetful functor 
	from $\mathcal{D}$ to $\set$ preserves products. Now we have:
	\begin{enumerate}[i)]
		\item $\md_f(\Ps)$ is closed under $\mathscr{E}$--quotients: Let $\mathbf{A},\mathbf{B}\in \Alg(\mathsf{T})$ with $\mathbf{A}\in \md_f(\Ps)$ and let 
			$e\in \Alg(\mathsf{T})(\mathbf{A},\mathbf{B})\cap \mathscr{E}$. We have that $B$ is finite since $e$ is epi and the forgetful functor from $\mathcal{D}$ to 
			$\set$ preserves epis. Let $f\in \Alg(\mathsf{T})(\mathbf{TX},\mathbf{B})$, $X\in \mathcal{D}_0$. Using (B$_f$2), there exists 
			$k\in \Alg(\mathsf{T})(\mathbf{TX},\mathbf{A})$ such that $f=e\circ k$. As $\mathbf{A}\in \md_f(\Ps)$ then 
			$k$ factors through some $e'\in \Ps(X)$ as $k=g\circ e'$. Then $f=e\circ k=e\circ g\circ e'$ with $e'\in \Ps(X)$, i.e., $\mathbf{B}\in \md_f (\Ps)$.
		\item $\md_f(\Ps)$ is closed under $\mathscr{M}$--subalgebras: Let $\mathbf{A},\mathbf{B}\in \Alg(\mathsf{T})$ with $\mathbf{A}\in \md_f(\Ps)$ and let 
			$m\in \Alg(\mathsf{T})(\mathbf{B},\mathbf{A})\cap\mathscr{M}$. We have that $B$ is finite since $m$ is mono and the forgetful functor from $\mathcal{D}$ 
			to $\set$ preserves monos. Let $f\in \Alg(\mathsf{T})(\mathbf{TX},\mathbf{B})$, $X\in\mathcal{D}_0$. As $\mathbf{A}\in \md_f(\Ps)$ then 
			$m\circ f$ factors through some $e\in \Ps(X)$ as $m\circ f=g\circ e$. Let $f=m_f\circ e_f$ and $g=m_g\circ e_g$ be the factorizations of $f$ and $g$, respectively. 
			We have that $(m\circ m_f)\circ e_f=m\circ f=g\circ e=m_g\circ (e_g\circ e)$ where $m\circ m_f, m_g\in \mathscr{M}$ and $e_f,e_g\circ e\in \mathscr{E}$. Then 
			by uniqueness of the factorization we have that there is an isomorphism $\phi$ such that $\phi\circ e_g\circ e=e_f$. Therefore $f=m_f\circ e_f=m_f\circ \phi\circ e_g\circ e$ 
			with $e\in \Ps(X)$, i.e., $\mathbf{B}\in \md_f (\Ps)$. 
		\item $\md_f(\Ps)$ is closed under finite products: Let $\mathbf{A}_i\in \md_f(\Ps)$, $i\in I$ with $I$ finite, and let $\mathbf{A}=\prod_{i\in I}\mathbf{A}_i$ be their product 
			in $\Alg(\mathsf{T})$ with projections $\pi_i:A\to A_i$. We have that $A$ is finite since the forgetful functor from $\mathcal{D}$ to $\set$ preserves products, $I$ 
			is finite, and each $A_i$ is finite. Let $f\in \Alg(\mathsf{T})(\mathbf{TX},\mathbf{A})$, $X\in \mathcal{D}_0$. As $\mathbf{A}_i\in \md_f(\Ps)$ then 
			$\pi_i\circ f$ factors through some $e_i\in \Ps(X)$ as $\pi_i\circ f=g_i\circ e_i$. Since $\Ps$ is a pseudoequational $\mathsf{T}$--theory there exists $e\in \Ps(X)$ 
			such that every $e_i$ factors through $e$ as $h_i\circ e=e_i$, $i\in I$. Let $\mathbf{Q}$ be the codomain of $e$. Now, by definition of $\mathbf{A}$ there exists 
			$h\in \Alg(\mathsf{T})(\mathbf{Q},\mathbf{A})$ such that $\pi_i\circ h=g_i\circ h_i$. As $\pi_i\circ f=\pi_i\circ h\circ e$ for every $i\in I$, then $f=h\circ e$, 
			$e\in \Ps(X)$, which means that $\mathbf{A}\in \md_f(\Ps)$.
	\end{enumerate}
\end{proof}

\noindent Given a class $K$ of algebras in $\Alg_f(\mathsf{T})$ define the operator $\Ps_K$ on $\mathcal{D}_0$ as follows:
\[
	\Ps_K(X)=\text{$\mathsf{T}$--algebra morphisms in $\mathscr{E}$ with domain $\mathbf{TX}$ and codomain in $K$.}
\]

\begin{proposition}
	Let $\mathcal{D}$ be a complete concrete category such that its forgetful functor preserves epis, monos and products $\mathsf{T}$ a monad on $\mathcal{D}$, $\mathcal{D}_0$ 
	a full subcategory of $\mathcal{D}_0$ and $\mathscr{E}/\mathscr{M}$ a factorization system on $\mathcal{D}$. Assume (B$_f$1) and (B$_f$4). Let $K$ be a pseudovariety of 
	$\mathsf{T}$--algebras. Then $\Ps_K$ is a pseudoequational $\mathsf{T}$--theory on $\mathcal{D}_0$.
\end{proposition}

\begin{proof}
	We have to prove properties i), ii), and ii) of Definition \ref{defpet}. In fact:
	\begin{enumerate}[i)]
		\item Let $X\in \mathcal{D}_0$, $I$ a finite set and $f_i\in \Ps_K(X)$, $i\in I$. Let $\mathbf{A}_i\in K$ be the codomain of $f_i$. Let 
			$\mathbf{A}=\prod_{i\in I}\mathbf{A}_i$ with projections $\pi_i\in \Alg(\mathsf{T})(\mathbf{A},\mathbf{A}_i)$. We have $\mathbf{A}\in K$. Now, by definition of 
			$\mathbf{A}$, there exists $f\in \Alg(\mathsf{T})(\mathbf{TX},\mathbf{A})$ such that $\pi_i\circ f=f_i$. Let $f=m_f\circ e_f$ be the factorization of $f$ with 
			$e_f\in \Alg(\mathsf{T})(\mathbf{TX},\mathbf{Q})\cap\mathscr{E}$ and $m_f\in \Alg(\mathsf{T})(\mathbf{Q},\mathbf{A})\cap\mathscr{M}$. 
			We have that $\mathbf{Q}\in K$. Then $e_f$ is a morphism in $\Ps_K(X)$ such that every $f_i$ factors through $e_f$.
		\item Let $X\in \mathcal{D}_0$, $e\in \Ps_K(X)$ with codomain $\mathbf{A}\in K$, and $e'\in \Alg(\mathsf{T})(\mathbf{A},\mathbf{B})\cap\mathscr{E}$. 
			We have that $\mathbf{B}$ is finite and that $\mathbf{B}\in K$. Therefore $e'\circ e\in\Ps_K(X)$.
		\item Let $X,Y\in \mathcal{D}_0$, $f\in \Ps_K(X)$ with codomain $\mathbf{A}\in K$, and $h\in \Alg(\mathsf{T})(\mathbf{TY},\mathbf{TX})$. Let 
			$f\circ h=m_{f\circ h}\circ e_{f\circ h}$ be the factorization of $f\circ h$ such that $e_{f\circ h}\in \Alg(\mathsf{T})(\mathbf{TY},\mathbf{Q})\cap\mathscr{E}$ and 
			$m_{f\circ h}\in \Alg(\mathsf{T})(\mathbf{Q},\mathbf{A})\cap\mathscr{M}$. Then $\mathbf{Q}\in K$, which implies $e_{f\circ h}\in \Ps_K(Y)$.			
	\end{enumerate}
\end{proof}

\begin{lemma}\label{lemmafr}
	Let $\mathcal{D}$ be a complete concrete category such that its forgetful functor preserves epis,monos and products, $\mathsf{T}$ a monad on $\mathcal{D}$, $\mathcal{D}_0$ 
	a full subcategory of $\mathcal{D}_0$ and $\mathscr{E}/\mathscr{M}$ a factorization system on $\mathcal{D}$. Assume (B$_f$1), (B$_f$2) and (B$_f$4). Let $\Ps$ be a 
	pseudoequational $\mathsf{T}$--theory on $\mathcal{D}_0$. Let $X\in \mathcal{D}_0$ and $e\in \Ps(X)$ with codomain $\mathbf{A}\in \Alg_f(\mathsf{T})$, 
	then $\mathbf{A}\in \md_f(\Ps)$.
\end{lemma}

\begin{proof}
	Let $Y\in \mathcal{D}_0$ and $f\in \Alg(\mathsf{T})(\mathbf{TY},\mathbf{A})$. We have to show that $f$ factors through some element in $\Ps(Y)$. 
	We have the following commutative diagram:
	
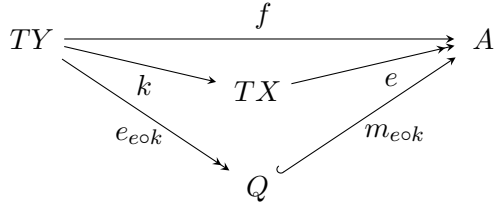
\begin{center}	
	\begin{tikzpicture}[>=stealth,shorten >=3pt,
			node distance=2.5cm,on grid,auto,initial text=,
			accepting/.style={thick,double distance=2.5pt}]
			\node (b) at (3,1.3) {$TX$};
			\node (c) at (3,0) {$Q$};
			\node (d) at (6,2) {$A$};
			\node (e) at (0,2) {$TY$};
			\path[->>] (b) edge [] node [below] {\ \ \ $e$} (d)
					(e) edge [] node [below] {$e_{e\circ k}\ \ $} (c);
			\path[right hook->] (c) edge [] node [below] {\ \ \ \ \ $m_{e\circ k}$} (d);
			\path[->] (e) edge [] node [below] {$k$} (b)
						edge [] node [above] {$f$} (d);
	\end{tikzpicture}
	\end{center}
	where:
	\begin{enumerate}[-]
		\item the morphism $k$ is obtained from $f$ and $e$ by using (B$_f$2),
		\item $e\circ k=m_{e\circ k}\circ e_{e\circ k}$ is the factorization of $e\circ k$.
	\end{enumerate}
	From the previous diagram we have that $e_{e\circ k}\in \Ps(Y)$, since $e\in \Ps(X)$ and $\Ps$ is a pseudoequational $\mathsf{T}$--theory. Therefore 
	$f$ factors through $e_{e\circ k}\in \Ps(Y)$, which implies that $\mathbf{A}\in \md_f(\Ps)$.
\end{proof}

\noindent To finish the proof of Theorem \ref{reitthm} we establish the following one--to--one correspondence between pseudoequational $\mathsf{T}$--theories and 
pseudovarieties of $\mathsf{T}$--algebras. 

\begin{proposition}
	Let $\mathcal{D}$ be a complete concrete category such that its forgetful functor preserves epis, monos and products, $\mathsf{T}$ a monad on $\mathcal{D}$, $\mathcal{D}_0$ 
	a full subcategory of $\mathcal{D}_0$ and $\mathscr{E}/\mathscr{M}$ a factorization system on $\mathcal{D}$. Assume (B$_f$1), (B$_f$2) and (B$_f$4). Let $\Ps$ be a 
	pseudoequational $\mathsf{T}$--theory on $\mathcal{D}_0$ and let $K$ be a pseudovariety of $\mathsf{T}$--algebras. Then:
	\begin{enumerate}[i)]
		\item $\Ps_{\md_f(\Ps)}=\Ps$.
		\item Assume (B$_f$3), then $\md_f(\Ps_K)=K$.
	\end{enumerate}
\end{proposition}

\begin{proof}\hfill
	\begin{enumerate}[i)]
		\item Let $X\in \mathcal{D}_0$, we have to prove that $\Ps_{\md_f(\Ps)}(X)=\Ps(X)$.\\
			($\subseteq$): Let $e\in \Ps_{\md_f(\Ps)}(X)$ with codomain $\mathbf{A}\in \md_f(\Ps)$. As $\mathbf{A}\in \md_f(\Ps)$, there exists $e'\in \Ps(X)$ 
			such that $e$ factors through $e'$ as $g\circ e'=e$. By (B$_f$1) and (B$_f$4) we have that $g$ is a $\mathsf{T}$--algebra morphism. As $g\circ e'=e\in \mathscr{E}$, 
			then $g\in \mathscr{E}$, and, as $\Ps$ is a pseudoequational $\mathsf{T}$--theory, then $g\circ e'=e\in \Ps(X)$.\\ 
			($\supseteq$): Let $e\in \Ps(X)$ with codomain $\mathbf{A}$. By Lemma \ref{lemmafr}, $\mathbf{A}\in \md_f(\Ps)$, i.e., $e\in \Ps_{\md_f(\Ps)}(X)$.
		\item Let $\mathbf{A}$ be an object in $\Alg_f(\mathsf{T})$.\\
			($\supseteq$): Assume that $\mathbf{A}\in K$. We have to show that $\mathbf{A}\in \md_f(\Ps_K)$. In fact, let $X\in \mathcal{D}_0$ 
			and $f\in \Alg(\mathsf{T})(\mathbf{TX},\mathbf{A})$. Let $f=m_f\circ e_f$ be the factorization of $f$ with  
			$e_f\in \Alg(\mathsf{T})(\mathbf{TX},\mathbf{Q})\cap\mathscr{E}$ and $m_f\in\Alg(\mathsf{T})(\mathbf{Q},\mathbf{A})\cap\mathscr{M}$. Then $\mathbf{Q}\in K$, 
			which implies that $e_f\in \Ps_K(X)$, i.e., $\mathbf{A}\in \md_f(\Ps_K)$.\\
			($\subseteq$): Assume that $\mathbf{A}\in \md_f(\Ps_K)$. By (R3) there exists an object $X_A\in \mathcal{D}_0$ and 
			$e\in \Alg(\mathsf{T})(\mathbf{TX_A},\mathbf{A})\cap\mathscr{E}$. As $\mathbf{A}\in \md_f(\Ps_K)$, $e$ factors through some 
			$e'\in \Ps_K(X_A)$ as $e=g\circ e'$. Let $\mathbf{Q}\in K$ be the codomain of $e'$. As $g\circ e'=e\in \mathscr{E}$, then $g\in \mathscr{E}$ and 
			$g\in \Alg(\mathsf{T})(\mathbf{Q},\mathbf{A})$ which implies that $\mathbf{A}\in K$ since $\mathbf{Q}\in K$.		
	\end{enumerate}
\end{proof}

\subsubsection{Details for Example \ref{exuoaf}}

Let $\Ps$ be a pseudoequational $\mathsf{T}_K$--theory on $\mathcal{D}_0$ and let $\Lan$ be an operator on $\mathcal{D}_0$ satisfying the properties i), ii) and iii). Then:
\begin{enumerate}[a)]
	\item Define the operator $\Lan_{\Ps}$ on $\mathcal{D}_0$ as $\Lan_{\Ps}(X):=\bigcup_{e\in \Ps(X)}\Imag(\poset(e,\mathbf{2}_c))$. 
		Then $\Lan_{\Ps}$ satisfies properties i), ii) and iii). The proof is similar to \ref{exuaf} a). Note that the directed union of finite objects in 
		$\ACDL$ that are subobjects of $\poset(T_KX,\mathbf{2}_c)\cong \set(T_KX,2)$ is a distributive sublattice of $\poset(T_KX,\mathbf{2}_c)\cong \set(T_KX,2)$.
	\item Define the operator $\Ps_{\Lan}$ on $\mathcal{D}_0$ such that $\Ps_{\Lan}(X)$ is the collection of all $\mathsf{T}_K$--algebra morphisms 
		$e\in \mathscr{E}$ with domain $\mathbf{T_KX}$ and finite codomain such that $\Imag(\poset(e,\mathbf{2}_c))\subseteq \Lan(X)$. We claim that $\Ps_{\Lan}$ is a pseudoequational 
		$\mathsf{T}_K$--theory. Non--emptiness and properties ii) and iii) from Definition \ref{defpet} are proved in a similar way as in \ref{exuaf} b). Now, to prove property i) in 
		Definition \ref{defpet}, consider a family $\{T_KX\overset{e_i}{\lepi} A_i\}_{i\in I}$ in $\Ps_{\Lan}(X)$ with $I$ finite such that $\Imag(\poset(e_i,\mathbf{2}_f))\subseteq \Lan(X)$, 
		we need to find a morphism $e\in \Ps_{\Lan}(X)$ such that every $e_i$ factors through $e$. In fact, let $\mathbf{A}$ be the product of $\prod_{i\in I}\mathbf{A_i}$ with projections 
		$\pi_i:A\to A_i$, then, by the universal property of $\mathbf{A}$ there exists a $\mathsf{T}_K$--algebra morphism $f:T_KX\to A$ such that $\pi_i\circ f=e_i$, for every $i\in I$. 
		Let $f=m_f\circ e_f$ be the factorization of $f$ in $\Alg(\mathsf{T}_K)$. We claim that $e=e_f$ is a morphism in $\Ps_{\Lan}(X)$ such that every $e_i$ factors through $e$. 
		Clearly, from the construction above, each $e_i$ factors through $e=e_f$. Now, let's prove that $\Imag(\poset(e,\mathbf{2}_c))\subseteq \Lan(X)$. In fact, let $\mathbf{S}$ 
		be the codomain of $e=e_f$ and let $g\in \poset(S,\mathbf{2}_c)$. We have to prove that $g\circ e\in \Lan(X)$ which follows from the following identity:
		\[
			g\circ e=\bigcup_{s\in g}\left(\bigcap_{i\in I} h_{i,s}\circ e_i\right)
		\]
		where $h_{i,s}\in \poset(A_i,\mathbf{2}_c)$ is defined as $h_{i,s}(x)=1$ iff $x\geq \pi_i(m_f(s))$. In fact, for any $w\in T_KX$ we have that $(g\circ e)(w)=1$ 
		implies $(h_{i,e(w)}\circ e_i)(w)=1$ for every $i\in I$, on the other hand, if there is $s\in g$ such that $(h_{i,s}\circ e_i)(w)=1$ for every $i\in I$ then 
		$e_i(w)\geq (\pi_i\circ m_f)(s)$, i.e., $(\pi_i\circ m_f\circ e)(w)\geq (\pi_i\circ m_f)(s)$ for every $i\in I$ (since $e_i=\pi_i\circ m_f\circ e$), which implies that 
		$(m_f\circ e)(w)\geq m_f(s)$ (since the order in $\mathbf{A}$ is componentwise) and the later implies that $e(w)\geq s$ (since $m_f$ is an embedding). Therefore,  
		$(g\circ e)(w)=1$ since $s\in g$ (i.e., $g(s)=1$).\\		

		\noindent Now, for every $s\in S$ and $i\in I$ the composition $h_{i,s}\circ e_i$ belongs to $\Lan(X)$ since 
		$h_{i,s}\circ e_i\in \Imag(\poset(e_i,\mathbf{2}_c))\subseteq \Lan(X)$. As $S$ and $I$ are finite then $g\circ e\in \Lan(X)$ because $\Lan(X)$ is a distributive lattice.
	\item We have that $\Ps=\Ps_{\Lan_{\Ps}}$. In fact, for every $X\in \mathcal{D}_0$ the inclusion $\Ps(X)\subseteq \Ps_{\Lan_{\Ps}}(X)$ is obvious. Now, to prove that 
		$\Ps_{\Lan_{\Ps}}(X)\subseteq \Ps(X)$, let $e'\in \Alg(\mathsf{T}_K)(\mathbf{T_KX},\mathbf{A})\cap\mathscr{E}$ with finite codomain such that $e'\in \Ps_{\Lan_{\Ps}}(X)$, i.e., 
		$\Imag(\poset(e'.\mathbf{2}_f))\subseteq \bigcup_{e\in \Ps(X)}\Imag(\poset(e,\mathbf{2}_c))$. Then the previous inclusion means that for every 
		$f\in \poset(\mathbf{A},\mathbf{2}_c)$ there exists $e_f\in \Ps(X)$ and $g_f$ such 
		that $f\circ e'=g_f\circ e_f$. As $\{e_f\mid f\in \poset(\mathbf{A},\mathbf{2}_c)\}$ is finite, then there exists $e\in \Ps(X)$ such that each $e_f$ factors through $e$. 
		We will prove that $e'$ 	factors through $e\in \Ps(X)$ which will imply that $e'\in \Ps(X)$, since $\Ps$ is a pseudoequational $\mathsf{T}_K$--theory. It is enough to show that 
		for all $u,v\in T_KX$ $e(u)\leq e(v)$ implies $e'(u)\leq e'(v)$. 
		In fact, assume that $e(u)\leq e(v)$ and define $f'\in\poset(\mathbf{A},\mathbf{2}_c)$ as $f'(x)=1$ iff $e'(u)\leq x$. Then, as $e_{f'}$ factors through $e$ we have that 
		$e_{f'}(u)\leq e_{f'}(v)$. By applying $g_{f'}$ to the last inequality, and using the fact that $f'\circ e'=g_{f'}\circ e_{f'}$, we get $1=f'(e'(u))\leq f'(e'(v))$ which implies that 
		$e'(u)\leq e'(v)$ by definition of $f'$.
	\item Similar to \ref{exuaf} d) by making the obvious changes.
\end{enumerate}

\subsubsection{Details for Example \ref{expvv4}}

Let $\Ps$ be a pseudoequational $\mathsf{T}$--theory on $\mathcal{D}_0$ and let $\Lan$ be an operator on $\mathcal{D}_0$ satisfying the properties i), ii) and iii). Then:

\begin{enumerate}[a)]
	\item Define the operator $\Lan_{\Ps}$ on $\mathcal{D}_0$ as $\Lan_{\Ps}(X):=\bigcup_{e\in \Ps(X)}\Imag(\fvec(e,\mathbb{K}))$. 
		We claim that $\Lan_{\Ps}$ satisfies properties i), ii) and iii). The proof is similar to \ref{exuaf} a). Note that the directed union of finite objects in 
		$\tvec$ that are subobjects of $\fvec(\mathsf{V}(X^*),\mathbb{K})$ is a $\mathbb{K}$--vector space which is a subspace of $\fvec(\mathsf{V}(X^*),\mathbb{K})$.
	\item Define the operator $\Ps_{\Lan}$ on $\mathcal{D}_0$ such that $\Ps_{\Lan}(X)$ is the collection of all $\mathsf{T}$--algebra morphisms 
		$e\in \mathscr{E}$ with domain $\mathbf{TX}$ and finite codomain such that $\Imag(\fvec(e,\mathbb{K}))\subseteq \Lan(X)$. We claim that $\Ps_{\Lan}$ is a pseudoequational 
		$\mathsf{T}$--theory. Non--emptiness and properties ii) and iii) from Definition \ref{defpet} are proved in a similar way as in \ref{exuaf} b). Now, to prove property i) in 
		Definition \ref{defpet}, consider a family $\{TX\overset{e_i}{\lepi} A_i\}_{i\in I}$ in $\Ps_{\Lan}(X)$ with $I$ finite such that $\Imag(\fvec(e_i,\mathbb{K}))\subseteq \Lan(X)$, 
		we need to find a morphism $e\in \Ps_{\Lan}(X)$ such that every $e_i$ factors through $e$. In fact, let $\mathbf{A}$ be the product of $\prod_{i\in I}\mathbf{A_i}$ with projections 
		$\pi_i:A\to A_i$, then, by the universal property of $\mathbf{A}$ there exists a $\mathsf{T}$--algebra morphism $f:TX\to A$ such that $\pi_i\circ f=e_i$, for every $i\in I$. 
		Let $f=m_f\circ e_f$ be the factorization of $f$ in $\Alg(\mathsf{T})$. We claim that $e=e_f$ is a morphism in $\Ps_{\Lan}(X)$ such that every $e_i$ factors through $e$. 
		Clearly, from the construction above, each $e_i$ factors through $e=e_f$. Now, let's prove that $\Imag(\fvec(e,\mathbb{K}))\subseteq \Lan(X)$. In fact, 
		let $\mathbf{S}$ be the codomain of $e=e_f$ and let $g\in \fvec(\mathbf{S},\mathbb{K})$. Let $\hat{g}\in \fvec(\mathbf{A},\mathbb{K})$ such that $\hat{g}\circ m_f=g$ 
		(this can be done since $\mathbb{K}$ is injective. In fact, define $\hat{g}$ as zero in $A\smallsetminus \Imag(m_f)$) and let $\iota_i\in \fvec(\mathbf{A}_i,\mathbf{A})$ 
		such that $\pi_i\circ \iota_i=id_{A_i}$ and $(\pi_{i'}\circ \iota_i)(y)=0$ if $i'\neq i$. Note that for every $x\in A$ we have $x=\sum_{i\in I} (\iota_i\circ \pi_i)(x)$. We have to prove 
		that $g\circ e\in \Lan(X)$ which follows from the following identity:
		\[
			g\circ e=\sum_{i\in I} \hat{g}\circ \iota_i\circ e_i
		\]
		In fact, for any $x\in \mathsf{V}(X^*)$ we have:
		\begin{align*}
			(g\circ e)(x)&=(\hat{g}\circ m_f\circ e)(x)=\hat{g}(m_f(e(x)))=\hat{g}\left(\sum_{i\in I}(\iota_i\circ \pi_i)(m_f(e(x)))\right)\\
				&=\sum_{i\in I}(\hat{g}\circ \iota_i\circ \pi_i\circ m_f\circ e)(x)=\sum_{i\in I}(\hat{g}\circ \iota_i\circ e_i)(x)
		\end{align*}
		\noindent From that we get that $g\circ e\in \Lan(X)$ since each $\hat{g}\circ \iota_i\circ e_i \in \Lan(X)$ and $\Lan(X)$ is a subspace of $\fvec(\mathsf{V}(X^*),\mathbb{K})$. 
	\item We have that $\Ps=\Ps_{\Lan_{\Ps}}$. In fact, for every $X\in \mathcal{D}_0$ the inclusion $\Ps(X)\subseteq \Ps_{\Lan_{\Ps}}(X)$ is obvious. Now, to prove that 
		$\Ps_{\Lan_{\Ps}}(X)\subseteq \Ps(X)$, let $e'\in \Alg(\mathsf{T})(\mathbf{TX},\mathbf{A})\cap\mathscr{E}$ with finite codomain such that $e'\in \Ps_{\Lan_{\Ps}}(X)$, i.e., 
		$\Imag(\fvec(e'.\mathbb{K}))\subseteq \bigcup_{e\in \Ps(X)}\Imag(\fvec(e,\mathbb{K}))$. Then the previous inclusion means that for every 
		$f\in \fvec(\mathbf{A},\mathbb{K})$ there exists $e_f\in \Ps(X)$ and $g_f$ such 
		that $f\circ e'=g_f\circ e_f$. As $\{e_f\mid f\in \fvec(\mathbf{A},\mathbb{K})\}$ is finite, then there exists $e\in \Ps(X)$ such that each $e_f$ factors through $e$. 
		We will prove that $e'$ 	factors through $e\in \Ps(X)$ which will imply that $e'\in \Ps(X)$, since $\Ps$ is a pseudoequational $\mathsf{T}$--theory. It is enough to show that 
		for all $u,v\in X^*$ $e(u)=e(v)$ implies $e'(u)=e'(v)$. 
		In fact, assume that $e(u)=e(v)$ and suppose by contradiction that $e'(u)\neq e'(v)$, then there exist $f'\in\fvec(\mathbf{A},\mathbb{K})$ 
		such that $(f'\circ e')(u)\neq (f'\circ e')(v)$, but then $e(u)=e(v)$ implies $(g_{f'}\circ e_{f'})(u)=(g_{f'}\circ e_{f'})(v)$, since $e_{f'}$ factors through $e$, 
		which is a contradiction since $g_{f'}\circ e_{f'}=f'\circ e'$.
	\item Similar to \ref{exuaf} d) by making the obvious changes. 
\end{enumerate}

\subsubsection{Details for Example \ref{expvv5}}

Let $\Ps$ be a pseudoequational $\mathsf{T}$--theory and let $\Lan$ be an operator on $\mathcal{D}_0$ satisfying the properties i), ii) and iii). Then:
\begin{enumerate}[a)]
	\item Define the operator $\Lan_{\Ps}$ on $\mathcal{D}_0$ as $\Lan_{\Ps}(X):=\bigcup_{e\in \Ps(X)}\Imag(\JSL(e,\mathbf{2}))$. 
		We claim that $\Lan_{\Ps}$ satisfies properties i), ii) and iii). The proof is similar to \ref{exuaf} a). Note that the directed union of finite objects in 
		$\TBJSL$ that are subobjects of $\JSL(\mathcal{P}_f(X^*),\mathbf{2})\cong \set(X^*,2)$ is a join subsemilattice of $\JSL(\mathcal{P}_f(X^*),\mathbf{2})\cong \set(X^*,2)$.
	\item Define the operator $\Ps_{\Lan}$ on $\mathcal{D}_0$ such that $\Ps_{\Lan}(X)$ is the collection of all $\mathsf{T}$--algebra morphisms 
		$e\in \mathscr{E}$ with domain $\mathbf{TX}$ and finite codomain such that $\Imag(\JSL(e,\mathbf{2}_c))\subseteq \Lan(X)$. We claim that $\Ps_{\Lan}$ is a pseudoequational 
		$\mathsf{T}$--theory. Non--emptiness and properties ii) and iii) from Definition \ref{defpet} are proved in a similar way as in \ref{exuaf} b). Now, to prove property i) in 
		Definition \ref{defpet}, consider a family $\{TX\overset{e_i}{\lepi} A_i\}_{i\in I}$ in $\Ps_{\Lan}(X)$ with $I$ finite such that $\Imag(\JSL(e_i,\mathbf{2}))\subseteq \Lan(X)$, 
		we need to find a morphism $e\in \Ps_{\Lan}(X)$ such that every $e_i$ factors through $e$. In fact, let $\mathbf{A}$ be the product of $\prod_{i\in I}\mathbf{A_i}$ with projections 
		$\pi_i:A\to A_i$, then, by the universal property of $\mathbf{A}$ there exists a $\mathsf{T}$--algebra morphism $f:TX\to A$ such that $\pi_i\circ f=e_i$, for every $i\in I$. 
		Let $f=m_f\circ e_f$ be the factorization of $f$ in $\Alg(\mathsf{T})$. We claim that $e=e_f$ is a morphism in $\Ps_{\Lan}(X)$ such that every $e_i$ factors through $e$. 
		Clearly, from the construction above, each $e_i$ factors through $e=e_f$. Now, let's prove that $\Imag(\JSL(e,\mathbf{2}))\subseteq \Lan(X)$. In fact, 
		let $\mathbf{S}$ be the codomain of $e=e_f$ and let $g\in \JSL(S,\mathbf{2})$. Let $\hat{g}\in \JSL(\mathbf{A},\mathbf{2})$ such that $\hat{g}\circ m_f=g$ 
		(this can be done since $\mathbf{2}$ is an injective semilattice, see \cite[Lemma 1]{bl}) and let $\iota_i\in \Alg(\mathsf{T})(\mathbf{A}_i,\mathbf{A})$ such that 
		$\pi_i\circ \iota_i=id_{A_i}$ and $(\pi_{i'}\circ \iota_i)(y)=0$ if $i'\neq i$. Note that for every $x\in A$ we have 
		$x=\bigvee_{i\in I} (\iota_i\circ \pi_i)(x)$. We have to prove that $g\circ e\in \Lan(X)$ which follows from the following identity:
		\[
			g\circ e=\bigvee_{i\in I} \hat{g}\circ \iota_i\circ e_i
		\]
		In fact, for any $W\in \mathcal{P}_f(X^*)$ we have:
		\begin{align*}
			(g\circ e)(W)&=(\hat{g}\circ m_f\circ e)(W)= \hat{g}\left(\bigvee_{i\in I}(\iota_i\circ \pi_i \circ m_f \circ e)(W)\right)\\
					&\bigvee_{i\in I}(\hat{g}\circ \iota_i\circ \pi_i \circ m_f \circ e)(W)=\left(\bigvee_{i\in I} \hat{g}\circ \iota_i\circ e_i \right)(W)
		\end{align*}
		\noindent From that we get that $g\circ e\in \Lan(X)$ since each $\hat{g}\circ \iota_i\circ e_i \in \Lan(X)$ and $\Lan(X)$ is a join subsemilattice of 
		$\JSL(\mathbf{TX},\mathbf{2})\cong \set(X^*,2)$. 
	\item We have that $\Ps=\Ps_{\Lan_{\Ps}}$. In fact, for every $X\in \mathcal{D}_0$ the inclusion $\Ps(X)\subseteq \Ps_{\Lan_{\Ps}}(X)$ is obvious. Now, to prove that 
		$\Ps_{\Lan_{\Ps}}(X)\subseteq \Ps(X)$, let $e'\in \Alg(\mathsf{T})(\mathbf{TX},\mathbf{A})\cap\mathscr{E}$ with finite codomain such that $e'\in \Ps_{\Lan_{\Ps}}(X)$, i.e., 
		$\Imag(\JSL(e'.\mathbf{2}))\subseteq \bigcup_{e\in \Ps(X)}\Imag(\JSL(e,\mathbf{2}))$. Then the previous inclusion means that for every 
		$f\in \JSL(\mathbf{A},\mathbf{2})$ there exists $e_f\in \Ps(X)$ and $g_f$ such 
		that $f\circ e'=g_f\circ e_f$. As $\{e_f\mid f\in \JSL(\mathbf{A},\mathbf{2})\}$ is finite, then there exists $e\in \Ps(X)$ such that each $e_f$ factors through $e$. 
		We will prove that $e'$ 	factors through $e\in \Ps(X)$ which will imply that $e'\in \Ps(X)$, since $\Ps$ is a pseudoequational $\mathsf{T}$--theory. It is enough to show that 
		for all $u,v\in TX$ $e(u)=e(v)$ implies $e'(u)=e'(v)$. In fact, assume that $e(u)=e(v)$ and suppose by contradiction that $e'(u)\neq e'(v)$, then there exist 
		$f'\in\JSL(\mathbf{A},\mathbf{2})$ such that $(f'\circ e')(u)\neq (f'\circ e')(v)$, but then $e(u)=e(v)$ implies 
		$(g_{f'}\circ e_{f'})(u)=(g_{f'}\circ e_{f'})(v)$, since $e_{f'}$ factors through $e$, which is a contradiction since $g_{f'}\circ e_{f'}=f'\circ e'$.
	\item Similar to \ref{exuaf} d) by making the obvious changes.
\end{enumerate}

\subsection{Details for Section \ref{seclocal}}

\subsubsection{Proof of local Birkhoff's Theorem for $\mathsf{T}$--algebras}

\begin{lemma}\label{llemmafree}
	Let $\mathcal{D}$ be a category, $\mathscr{E}/\mathscr{M}$ a factorization system on $\mathcal{D}$, $\mathsf{T}=(T,\eta,\mu)$ a monad on $\mathcal{D}$ and  
	$X\in\mathcal{D}$. Assume (b2). Let $TX\overset{e_X}{\lepi} Q_X$ be a local equational $\mathsf{T}$--theory on $X$. Then $\mathbf{Q_X}\in \md(e_X)$.
\end{lemma}

\begin{proof}
	Same as in Lemma \ref{lemmafree} by considering $\mathcal{D}_0=\{X\}$.
\end{proof}

\begin{proposition}\label{linjth}
	Let $\mathcal{D}$ be a category, $\mathscr{E}/\mathscr{M}$ a factorization system on $\mathcal{D}$, $\mathsf{T}=(T,\eta,\mu)$ a monad on $\mathcal{D}$ and  
	$X\in\mathcal{D}$. Assume (b1) and (b2). Let $TX\overset{(e_i)_X}{\lepi} (Q_i)_X$ be a local equational $\mathsf{T}$--theory on $X$, $i=1,2$. If 
	$(e_1)_X\neq (e_2)_X$ then $\md ((e_1)_X)\neq \md((e_2)_X)$.
\end{proposition}

\begin{proof} 
	Same as in Proposition \ref{injth} by considering $\mathcal{D}_0=\{X\}$.
\end{proof}

\begin{proposition}\label{lhspclosureb}
	Let $\mathcal{D}$ be a complete category, $\mathsf{T}=(T,\eta,\mu)$ a monad on $\mathcal{D}$, $\mathscr{E}/\mathscr{M}$ a factorization system on $\mathcal{D}$ and 
	$X\in \mathcal{D}$. Assume (b1), (b2) and (b3). Let $e_X$ be a local equational $\mathsf{T}$--theory on $X$. Then $\md(e_X)$ is a local variety of 
	$X$--generated $\mathsf{T}$--algebras.
\end{proposition}

\begin{proof}
	$\md (e_X)$ is nonempty  by Lemma \ref{llemmafree}. Put $TX\overset{e_X}{\lepi} Q_X$, then:
	\begin{enumerate}[i)]
		\item $\md(e_X)$ is closed under $\mathscr{E}$--quotients: similar proof to that of Proposition \ref{hspclosureb}. Note that an $\mathscr{E}$--quotient of an $X$--generated 
			$\mathsf{T}$--algebra is $X$--generated.
		\item $\md(e_X)$ is closed under $X$--generated $\mathscr{M}$--subalgebras: similar proof to that of Proposition \ref{hspclosureb}.
		\item $\md(e_X)$ is closed under subdirect products: Let $\mathbf{A}_i\in \md(e_X)$, $i\in I$, and let $\mathbf{S}$ be the subdirect product of the family 
			$\{(\mathbf{A}_i,e_i)\}_{i\in I}$, where $e_i\in \Alg(\mathsf{T})(\mathbf{TX},\mathbf{A}_i)\cap\mathscr{E}$, $i\in I$. Let $f\in \Alg(\mathsf{T}) (\mathbf{TX},\mathbf{S})$ 
			then we have the following commutative diagram:
			\begin{center}	
				\begin{tikzpicture}[>=stealth,shorten >=3pt,
					node distance=2.5cm,on grid,auto,initial text=,
					accepting/.style={thick,double distance=2.5pt}]
					\node (a) at (0,0) {$TX$};
					\node (b) at (4,0) {$S$};
					\node (c) at (2,1) {$R$};
					\node (d) at (0,-1.5) {$Q_X$};
					\node (e) at (2,-1) {$P$};
					\node (f) at (4,-1.5) {$\prod_{i\in I}A_i$};
					\node (g) at (4,-3) {$A_j$};
					\node (h) at (6.5,-1.5) {$TX$};
					\path[->>] (a) edge [] node [left] {$e_X$} (d)
								edge [] node [above] {$e_f$} (c)
							(d) edge [] node [above] {$e_g$} (e)
							(h) edge [] node [above] {$e_e$} (b)
								edge [] node [below] {$e_j$} (g);
					\path[->] (a) edge [] node [above] {$f$} (b)
							(h) edge [] node [above] {$e$} (f)
							(d) edge [] node [below] {$g$} (f)
								edge [] node [below] {$g_j$} (g)
							(f) edge [] node [right] {$\pi_j$} (g);
					\path[right hook->] (e) edge [] node [above] {$m_g$} (f)
									(c) edge [] node [above] {$m_f$} (b)
									(b) edge [] node [right] {$m_e$} (f);
				\end{tikzpicture}
			\end{center}
			where:
			\begin{enumerate}[-]
				\item the two right triangles are obtained from the construction of $\mathbf{S}$,
				\item $f=m_f\circ e_f$ is the factorization of $f$,
				\item $g_j$ is obtained from the property that $\mathbf{A}_j\in \md(e_X)$, i.e., $\pi_j\circ m_e\circ f=g_j\circ e_X$, $j\in I$,
				\item $g$ is obtained from the morphisms $g_j$ by using the universal property of the product $\prod_{i\in I}\mathbf{A}_i$, and 
				\item $g=m_g\circ e_g$ is the factorization of $g$.
			\end{enumerate}
			Then, by the universal property of the product, we have that $m_e\circ m_f\circ e_f=m_g\circ e_g\circ e_X$, where $e_g\circ e_X, e_f\in \mathscr{E}$ and 
			$m_g,m_e\circ m_f\in \mathscr{M}$. Hence, by uniquenes of factorization we have that there is an isomorphism $\phi$ such that 
			$e_f=\phi\circ e_g\circ e_X$, which implies that $f=m_f\circ e_f=m_f\circ \phi\circ e_g\circ e_X$, i.e., $\mathbf{S}\in \md(e_X)$.
	\end{enumerate}
\end{proof}

\begin{proposition}\label{lpropvartoeq}
	Let $\mathcal{D}$ be a complete category, $\mathscr{E}/\mathscr{M}$ a factorization system on $\mathcal{D}$, $\mathsf{T}=(T,\eta,\mu)$ a monad on $\mathcal{D}$ and 
	$X\in \mathcal{D}$. Assume (b1), (b3) and (b4).
	Let $V$ be a local variety of $X$--generated $\mathsf{T}$--algebras. Then $V=\md(e_X)$ for some local equational $\mathsf{T}$--theory $e_X$ on $X$.
\end{proposition}

\noindent Notice that, if we assume (b2), the local equational $\mathsf{T}$--theory $e_X$ on $X$ is unique by Proposition \ref{linjth}.

\begin{proof}
	We prove the proposition in two steps: i) the construction of $e_X$, and ii) to show that $V=\md(e_X)$. In fact:
	\begin{enumerate}[i)]
		\item Let $H=\{TX\overset{e_i}{\lepi} P_i\}_{i\in I}$ be the collection of all $\mathsf{T}$--algebra morphisms, up to isomorphism, 
			in $\mathscr{E}$ with domain $TX$ and codomain in the variety $V$. By (b4), $H$ is a set. Put $\mathbf{P}=\prod_{i\in I}\mathbf{P}_i$ and let 
			$\pi_i\in \Alg(\mathsf{T})(\mathbf{P},\mathbf{P}_i)$ be the $i$th--projection. Then we have the following commutative diagram in $\Alg(\mathsf{T})$:
			\begin{center}	
				\begin{tikzpicture}[>=stealth,shorten >=3pt,
					node distance=2.5cm,on grid,auto,initial text=,
					accepting/.style={thick,double distance=2.5pt}]
					\node (a) at (2,0) {$Q_X$};
					\node (b) at (6,0.5) {$P_i$};
					\node (c) at (0,0.5) {$TX$};
					\node (d) at (4,0.5) {$P$};
					\path[right hook->,dashed] (a) edge [] node [below] {$\ \ m_X$} (d);
					\path[->>,dashed] (c) edge [] node [below] {$e_X$} (a);
					\path[->,dashed] (c) edge [] node [above] {$k$} (d);
					\path[->] (d) edge [] node [above] {$\pi_i$} (b);
					\path[->>,bend angle=25] (c) edge [bend left,looseness=0.7] 
							node [above] {$e_i$} (b);
				\end{tikzpicture}
			\end{center}	
			where $k\in \Alg(\mathsf{T})(\mathbf{TX},\mathbf{P})$ is obtained from the universal property of the product $\mathbf{P}$ and $k=m_X\circ e_X$ is the 
			factorization of $k$, i.e., $m_X\in \mathscr{M}$ and $e_X\in \mathscr{E}$. Observe that $\mathbf{Q_X}\in V$ since it is a subdirect product of elements in $V$.\\
	
			\noindent {\it Claim:} $TX\overset{e_X}{\lepi} Q_X$ is a local equational $\mathsf{T}$--theory on $X$.\\

			\noindent Let $g\in \Alg(\mathsf{T})(\mathbf{TX},\mathbf{TX})$. We have to prove that there exists $g'\in \Alg(\mathsf{T})(\mathbf{Q_X},\mathbf{Q_Y})$ such 
			that $g'\circ e_X=e_Y\circ g$. In fact, we have the following commutative diagram:
			\begin{center}	
				\begin{tikzpicture}[>=stealth,shorten >=3pt,
					node distance=2.5cm,on grid,auto,initial text=,
					accepting/.style={thick,double distance=2.5pt}]
					\node (b) at (2,0) {$P$};
					\node (c) at (0,1.5) {$TX$};
					\node (d) at (0,0) {$Q_X$};
					\node (e) at (6.5,1.5) {$TX$};
					\node (f) at (6.5,0) {$Q_X$};
					\node (g) at (4,0) {$S=P_j$};
					\path[right hook->] (g) edge [] node [below] {$m_{e_X\circ g}$} (f)
							(d) edge [] node [below] {$m_X$} (b);
					\path[->>] (c) edge [] node [left] {$e_X$} (d)
							(e) edge [] node [right] {$e_X$} (f)
							(c) edge [] node [right] {$\ \ e_{e_X\circ g}=e_j$} (g);
					\path[->] 	(c) edge [] node [above] {$g$} (e)
							(b) edge [] node [below] {$\pi_j$} (g);			
				\end{tikzpicture}
			\end{center}	
			where $e_X\circ g=m_{e_X\circ g}\circ e_{e_X\circ g}$ is the factorization of $e_X\circ g$ and $\mathbf{S}$ is the codomain of $e_{e_X\circ g}$. 
			From that we have then that $\mathbf{S}$ is an $X$--generated $\mathscr{M}$--subalgebra of $\mathbf{Q_X}\in V$. Hence $\mathbf{S}\in V$ and therefore 
			$\mathbf{S}=\mathbf{P}_j$ and $e_{e_X\circ g}=e_j$ for some $j\in I$. Finally, commutativity of the triangle follows from the definition of $\mathbf{Q_X}$ above.
			Therefore, $e_X$ is a local equational $\mathsf{T}$--theory on $X$.
		\item Let us prove that $V=\md(e_X)$.\\
			$(\supseteq)$: Let $\mathbf{A}\in \Alg(\mathsf{T})$ such that $\mathbf{A}\in \md (e_X)$. Since $\mathbf{A}$ is $X$--generated, there exists 
			$s_A\in \Alg(\mathsf{T})(\mathbf{TX},\mathbf{A})\cap\mathscr{E}$. As $\mathbf{A}\in \md (e_X)$, the morphism $s_A$ factors through $e_X$ as 
			$s_A=g_{s_A}\circ e_X$. As we have that $g_{s_A}\circ e_X=s_A\in \mathscr{E}$ then, by (b1) and Lemma \ref{lemmafs}, we have that 
			$g_{s_A}\in \Alg(\mathsf{T})(\mathbf{Q_X},\mathbf{A})\cap \mathscr{E}$, and hence $\mathbf{A}\in V$ since it is an $\mathscr{E}$--quotient of $\mathbf{Q_X}\in V$.\\

			\noindent $(\subseteq)$: Let $\mathbf{A}\in \Alg(\mathsf{T})$ such that $\mathbf{A}\in V$. Let $f\in \Alg(\mathsf{T})(\mathbf{TX},\mathbf{A})$, then we have the following 
			commutative diagram:
			\begin{center}	
				\begin{tikzpicture}[>=stealth,shorten >=3pt,
					node distance=2.5cm,on grid,auto,initial text=,
					accepting/.style={thick,double distance=2.5pt}]
					\node (a) at (0,0) {$TX$};
					\node (b) at (3.5,0) {$Q_X$};
					\node (c) at (3.5,1.5) {$P$};
					\node (d) at (3.5,3) {$Z=P_i$};
					\node (e) at (0,3) {$A$};
					\path[->>] (a) edge [] node [below] {$e_X$} (b)
							edge [] node [left] {$e_{f}=e_i$} (d);
					\path[->] (a) edge [] node [left] {$f$} (e)
							(c) edge [] node [right] {$\pi_i$} (d);
					\path[right hook->] (b) edge [] node [right] {$m_X$} (c)
								(d) edge [] node [above] {$m_{f}$} (e);
				\end{tikzpicture}
			\end{center}	
			where $f=m_{f}\circ e_{f}$ is the factorization of $f$ with $m_{f}\in \mathscr{M}$ and $e_{f}\in \mathscr{E}$, which implies that $\mathbf{Z}\in V$ since 
			it is an $X$--generated $\mathscr{M}$--subalgebra of $\mathbf{A}\in V$. Therefore, $\mathbf{Z}=\mathbf{P}_i$ and $e_{f}=e_i$ for some $i\in I$. Hence the 
			factorization of $f$ through $e_X$ follows from the definition of $e_X$ (see i) above) which implies that $\mathbf{A}\in \md(e_X)$.
	\end{enumerate}
\end{proof}

\subsubsection{Proof of local Birkhoff's Theorem for finite $\mathsf{T}$--algebras}

\begin{proposition}\label{lhspclosurer}
	Let $\mathcal{D}$ be a concrete complete category such that its forgetful functor preserves epis, monos and products, $\mathsf{T}$ a monad on $\mathcal{D}$, $X\in \mathcal{D}$ 
	and $\mathscr{E}/\mathscr{M}$ a factorization system on $\mathcal{D}$. Assume (b$_f$1) to (b$_f$3). Let $\mathtt{P}_X$ be a local pseudoequational $\mathsf{T}$--theory 
	on $X$. Then $\md_f(\mathtt{P}_X)$ is a local pseudovariety of $X$--generated $\mathsf{T}$--algebras.
\end{proposition}

\begin{proof}
	$\md_f (\mathtt{P}_X)$ is nonempty by Lemma \ref{lemmafr} with $\mathcal{D}_0=\{X\}$. Now we have:
	\begin{enumerate}[i)]
		\item $\md_f(\mathtt{P}_X)$ is closed under $\mathscr{E}$--quotients: similar proof to that of Proposition \ref{hspclosurer}. Note that 
			an $\mathscr{E}$--quotiend of an $X$--generated algebra is also $X$--generated.
		\item $\md_f(\mathtt{P}_X)$ is closed under $X$--generated $\mathscr{M}$--subcoalgebras: similar proof to that of Proposition \ref{hspclosurer}.
		\item $\md_f(\mathtt{P}_X)$ is closed under finite subdirect products: Let $I$ be a finite set and let $\mathbf{A}_i\in \md_f(\mathtt{P}_X)$, $i\in I$. Let $\mathbf{S}$ be 
			the subdirect product of the family $\{(\mathbf{A}_i,e_i)\}_{i\in I}$, where $e_i\in \Alg(\mathsf{T})(\mathbf{TX},\mathbf{A}_i)\cap\mathscr{E}$, $i\in I$. Let 
			$f\in \Alg(\mathsf{T}) (\mathbf{TX},\mathbf{S})$ then we have the following commutative diagram:
			
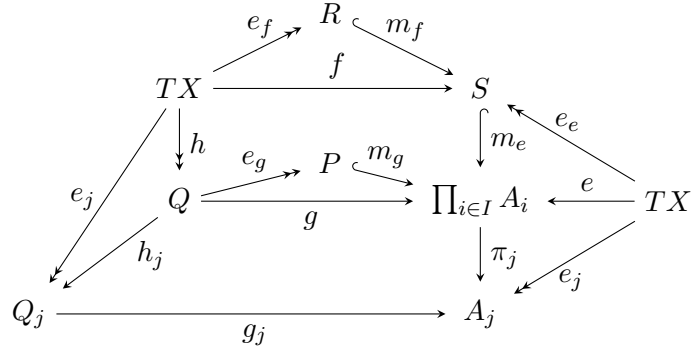
\begin{center}	
				\begin{tikzpicture}[>=stealth,shorten >=3pt,
					node distance=2.5cm,on grid,auto,initial text=,
					accepting/.style={thick,double distance=2.5pt}]
					\node (a) at (0,0) {$TX$};
					\node (b) at (4,0) {$S$};
					\node (c) at (2,1) {$R$};
					\node (d) at (-2,-3) {$Q_j$};
					\node (e) at (2,-1) {$P$};
					\node (f) at (4,-1.5) {$\prod_{i\in I}A_i$};
					\node (g) at (4,-3) {$A_j$};
					\node (h) at (0,-1.5) {$Q$};
					\node (i) at (6.5,-1.5) {$TX$};
					\path[->>] (a) edge [] node [left] {$e_j$} (d)
								edge [] node [right] {$h$} (h)
								edge [] node [above] {$e_f$} (c)
							(h) edge [] node [above] {$e_g$} (e)
							(i) edge [] node [above] {$e_e$} (b)
								edge [] node [below] {$e_j$} (g);
					\path[->] 	(a) edge [] node [above] {$f$} (b)
							(h) edge [] node [below] {$g$} (f)
							(d) edge [] node [below] {$g_j$} (g)
							(f) edge [] node [right] {$\pi_j$} (g)
							(h) edge [] node [right] {\ \ $h_j$} (d)
							(i) edge [] node [above] {$e$} (f);
					\path[right hook->] (e) edge [] node [above] {$m_g$} (f)
									(b) edge [] node [right] {$m_e$} (f)
									(c) edge [] node [above] {$m_f$} (b);
				\end{tikzpicture}
			\end{center}
			where:
			\begin{enumerate}[-]
				\item the two right triangles are obtained from the construction of $\mathbf{S}$,
				\item $f=m_f\circ e_f$ is the factorization of $f$,
				\item $g_j$ is obtained from the property that $\mathbf{A}_j\in \md_f(\mathtt{P}_X)$, i.e., $\pi_j\circ m_e\circ f=g_j\circ e_j$, $j\in I$ and $e_j\in \mathtt{P}_X$,
				\item $h\in\mathtt{P}_X$ is obtained from the morphisms $e_j\in \mathtt{P}_X$ by using the property that $\mathtt{P}$ is a local pseudoequational 
					$\mathsf{T}$--theory on $X$, i.e., $h_j\circ h=e_j$, $j\in I$,
				\item $g$ is obtained from the morphisms $g_j\circ h_j$ by using the universal property of the product $\prod_{i\in I}\mathbf{A}_i$, and 
				\item $g=m_g\circ e_g$ is the factorization of $g$.
			\end{enumerate}
			Then, by the universal property of the product, we have that $m_e\circ m_f\circ e_f=m_g\circ e_g\circ h$. Now, as $e_g\circ h, e_f\in \mathscr{E}$ and 
			$m_g,m_e\circ m_f\in \mathscr{M}$, by uniquenes of factorization we have that there is an isomorphism $\phi$ such that $e_f=\phi\circ e_g\circ h$, which implies 
			$f=m_f\circ e_f=m_f\circ \phi\circ e_g\circ h $, i.e., $\mathbf{S}\in \md_f(\mathtt{P}_X)$.
	\end{enumerate}
\end{proof}

\noindent Given a class $K$ of $X$--generated algebras in $\Alg_f(\mathsf{T})$ define the collection of morphisms $\mathtt{P}_X(K)$ as follows:
\[
	\mathtt{P}_X(K)=\text{$\mathsf{T}$--algebra morphisms in $\mathscr{E}$ with domain $\mathbf{TX}$ and codomain in $K$.}
\]

\begin{proposition}
	Let $\mathcal{D}$ be a concrete complete category such that its forgetful functor preserves epis, monos and products, $\mathsf{T}$ a monad on $\mathcal{D}$, $X\in \mathcal{D}$ 
	and $\mathscr{E}/\mathscr{M}$ a factorization system on $\mathcal{D}$. Assume (b$_f$1) and (b$_f$3). Let $K$ be a local pseudovariety of 
	$X$--generated $\mathsf{T}$--algebras. Then $\mathtt{P}_X(K)$ is a local pseudoequational $\mathsf{T}$--theory on $X$.
\end{proposition}

\begin{proof}
	We have to prove properties i), ii), and ii) of Definition \ref{deflpet}. In fact:
	\begin{enumerate}[i)]
		\item Let $I$ be a finite set and $f_i\in \mathtt{P}_X(K)$, $i\in I$. Let $\mathbf{A}_i\in K$ be the codomain of $f_i$. Let 
			$\mathbf{A}=\prod_{i\in I}\mathbf{A}_i$ with projections $\pi_i\in \Alg(\mathsf{T})(\mathbf{A},\mathbf{A}_i)$. Now, by definition of 
			$\mathbf{A}$, there exists $f\in \Alg(\mathsf{T})(\mathbf{TX},\mathbf{A})$ such that $\pi_i\circ f=f_i$. Let $f=m_f\circ e_f$ be the factorization of $f$ with 
			$e_f\in \Alg(\mathsf{T})(\mathbf{TX},\mathbf{Q})\cap\mathscr{E}$ and $m_f\in \Alg(\mathsf{T})(\mathbf{Q},\mathbf{A})\cap\mathscr{M}$. 
			We have that $\mathbf{Q}\in K$ since it is the subdirect product of $\{(\mathbf{A}_i,f_i)\}_{i\in I}$. Hence, $e_f\in \mathtt{P}_X(K)$ and every $f_i$ factors through it.
		\item Follows from the property that $K$ is closed under $\mathscr{E}$--quotients.
		\item Let $f\in \mathtt{P}_X(K)$ with codomain $\mathbf{A}\in K$, and $h\in \Alg(\mathsf{T})(\mathbf{TX},\mathbf{TX})$. Let 
			$f\circ h=m_{f\circ h}\circ e_{f\circ h}$ be the factorization of $f\circ h$ such that $e_{f\circ h}\in \Alg(\mathsf{T})(\mathbf{TX},\mathbf{Q})\cap\mathscr{E}$ and 
			$m_{f\circ h}\in \Alg(\mathsf{T})(\mathbf{Q},\mathbf{A})\cap\mathscr{M}$. Then $\mathbf{Q}\in K$ since it is an $X$--generated $\mathscr{M}$--subcoalgebra of 
			$\mathbf{A}\in K$, which implies $e_{f\circ h}\in \mathtt{P}_X(K)$.
	\end{enumerate}
\end{proof}

\begin{proposition}
	Let $\mathcal{D}$ be a concrete complete category such that its forgetful functor preserves epis, monos and products, $\mathsf{T}$ a monad on $\mathcal{D}$, $X\in \mathcal{D}$ 
	and $\mathscr{E}/\mathscr{M}$ a factorization system on $\mathcal{D}$. Assume (b$_f$1) to (b$_f$3). Let $\mathtt{P}_X$ be a local pseudoequational $\mathsf{T}$--theory on $X$ 
	and let $K$ be a local pseudovariety of $X$--generated $\mathsf{T}$--algebras. Then:
	\begin{enumerate}[i)]
		\item $\mathtt{P}_X({\md_f(\mathtt{P}_X)})=\mathtt{P}_X$.
		\item $\md_f(\mathtt{P}_X(K))=K$.
	\end{enumerate}
\end{proposition}

\begin{proof}\hfill
	\begin{enumerate}[i)]
		\item ($\subseteq$): Let $e\in \mathtt{P}_X({\md_f(\mathtt{P}_X)})$ with codomain $\mathbf{A}\in \md_f(\mathtt{P}_X)$. As $\mathbf{A}\in \md_f(\mathtt{T})$, 
			there exists $e'\in \mathtt{P}_X$ such that $e$ factors through $e'$ as $g\circ e'=e$. By (b$_f$2) and (b$_f$4) we have that $g$ is a $\mathsf{T}$--algebra morphism. 
			As $g\circ e'=e\in \mathscr{E}$, then $g\in \mathscr{E}$, and, as $\mathtt{P}_X$ is a pseudoequational $\mathsf{T}$--theory, then $g\circ e'=e\in \mathtt{P}_X$.\\ 
			($\supseteq$): Let $e\in \mathtt{P}_X$ with codomain $\mathbf{A}$. Using Lemma \ref{lemmafr} with $\mathcal{D}_0=\{X\}$, $\mathbf{A}\in\md_f(\mathtt{P}_X)$, 
			i.e., $e\in \mathtt{P}_X({\md_f(\mathtt{P}_X)})$.
		\item Let $\mathbf{A}$ be an $X$--generated algebra in $\Alg_f(\mathsf{T})$.\\
			($\supseteq$): Assume that $\mathbf{A}\in K$. We have to show that $\mathbf{A}\in \md_f(\mathtt{P}_X(K))$. In fact, let $f\in \Alg(\mathsf{T})(\mathbf{TX},\mathbf{A})$ 
			and $f=m_f\circ e_f$ be the factorization of $f$ with $e_f\in \Alg(\mathsf{T})(\mathbf{TX},\mathbf{Q})\cap\mathscr{E}$ and 
			$m_f\in\Alg(\mathsf{T})(\mathbf{Q},\mathbf{A})\cap\mathscr{M}$. Then $\mathbf{Q}\in K$ since it is an $X$--generated $\mathscr{M}$--subcoalgebra of 
			$\mathbf{A}\in K$, which implies that $e_f\in \mathtt{P}_X(K)$, i.e., $\mathbf{A}\in \md_f(\mathtt{P}_X(K))$.\\
			($\subseteq$): Assume that $\mathbf{A}\in \md_f(\mathtt{P}_X(K))$. Since $\mathbf{A}$ is an $X$--generated $\mathsf{T}$--algebra, there exists 
			$e\in \Alg(\mathsf{T})(\mathbf{TX},\mathbf{A})\cap\mathscr{E}$. As $\mathbf{A}\in \md_f(\mathtt{P}_X(K))$, $e$ factors through some $e'\in \mathtt{P}_X(K)$ as 
			$e=g\circ e'$. Let $\mathbf{Q}\in K$ be the codomain of $e'$. 	As $g\circ e'=e\in \mathscr{E}$, then $g\in \mathscr{E}$ and 
			$g\in \Alg(\mathsf{T})(\mathbf{Q},\mathbf{A})$ which implies that $\mathbf{A}\in K$ since it is an $\mathscr{E}$--quotient of $\mathbf{Q}\in K$.
	\end{enumerate}
\end{proof}

\end{document}